\documentclass[12pt]{article}
\usepackage[left=1in, right=1in, top=1in, bottom=1in]{geometry}
\usepackage{makeidx}         % allows index generation
\usepackage{graphicx}        % standard LaTeX graphics tool
\usepackage{multicol}        % used for the two-column index
\usepackage[bottom]{footmisc}% places footnotes at page bottom

\usepackage[dvipsnames]{xcolor}

\usepackage{tabularx}

\usepackage{amsmath}
\usepackage{bm}
\usepackage{amssymb}
\usepackage{mathrsfs}
\usepackage{amsfonts}
\usepackage{enumerate}
\usepackage{hyperref}
\hypersetup{
    colorlinks=true,        % Enable colored links
    linkcolor=blue,         % Set the color of cross-references (\ref{})
    citecolor=blue,        % Set the color of citation links (\cite{})
    filecolor=black,      % Set the color of file links
    urlcolor=cyan           % Set the color of URLs
}

\makeindex             % used for the subject index

\usepackage[authoryear]{natbib}

\newtheorem{theorem}{Theorem}[section]
\newtheorem{remark}[theorem]{Remark}
\newtheorem{corollary}[theorem]{Corollary}
\newtheorem{lemma}[theorem]{Lemma}
\newtheorem{assumption}{Assumption}
\newtheorem{proposition}[theorem]{Proposition}

\newtheorem{algorithm}[theorem]{Algorithm}

\newtheorem{example}[theorem]{Example}

\usepackage{float}
\usepackage{subfiles}

\usepackage[utf8]{inputenc}

\usepackage{fancyhdr}
\usepackage{setspace}
\usepackage{amsmath}

\usepackage[most]{tcolorbox}

\usepackage[font=small,labelfont=bf,labelsep=period]{caption}
\newcommand{\prob}{\mathbb{P}}

\usepackage{colortbl}
\usepackage{xcolor}

\definecolor{siggray}{gray}{0.85}
\usepackage{import}
\usepackage{booktabs}

\begin{document}

\title{Heavy Tails and Predictive Ability Testing
\thanks{We thank Mikkel Bennedsen, Bent Jesper Christensen, Michael Jansson, Søren Johansen, Thomas Mikosch, Anders Rahbek, Frederik Vilandt Rasmussen, Mikkel Sølvsten, Peter N. Sørensen, and seminar participants at Aarhus University for comments and suggestions. Frederiksen and Pedersen are partly supported by the Independent Research Fund Denmark (DFF Grant 7015-00028). Support from the Danish Finance Institute (DFI) is gratefully acknowledged by Pedersen. Matsui is partly supported by the JSPS Grant-in-Aid for Scientific Research C (19K11868).} \\
}

\author{Jonas F. Frederiksen\footnote{University of Copenhagen; email: \tt{jff@econ.ku.dk}} \and Muneya Matsui\footnote{University of Osaka; email: \tt{mmuneya@gmail.com}} \and Rasmus S. Pedersen\footnote{University of Copenhagen and Danish Finance Institute; email: \tt{rsp@econ.ku.dk} }}

\maketitle

\date{}

\begin{abstract}
We study the asymptotic behavior of widely used tests for evaluating and comparing predictive accuracy when forecast errors exhibit heavy tails. In particular, when loss differentials have infinite variance, the Diebold–Mariano test statistic converges to a nonstandard limit involving non-Gaussian stable random variables. As a consequence, conventional critical values can yield severely distorted inference: a nominal 5\% test may reject a true null as often as 70\% of the time.
To establish these results, we develop a new stable limit theorem for strongly mixing, infinite-variance  time series processes. Building on this theory, we consider  subsampling-based inference that remains valid irrespective of tail-heaviness and requires no estimation of long-run variances or tail indices. An application to risk forecasts for emerging-market exchange rates shows that accounting for heavy tails can substantially alter conclusions about predictive performance relative to standard procedures.

\end{abstract}


\newpage
\numberwithin{equation}{section}

\newpage

\section{Introduction}

The Diebold--Mariano (DM) test of \cite{diebold1995paring} and the 
superior predictive ability (SPA) tests of \cite{white2000snooping} and 
\cite{hansen2005SPA} are workhorse tools for forecast comparison in 
economics and finance. Their asymptotic justification rests on the 
central limit theorem (CLT) of \cite{ibragimov1962clt} for strongly mixing processes with \textit{finite} variance.

This finite-variance requirement is at odds with a basic stylized fact 
about the data. Many of the underlying variables -- asset returns, 
exchange rate fluctuations, and measures of financial volatility -- 
have power-law (heavy) tails with index $\kappa>0$, and the empirically 
relevant range routinely includes $\kappa<2$, corresponding to infinite 
(or undefined) variance; see, e.g., \cite{gabaix2009power} and 
\citet[Section 1.2]{ibragimov2015heavytails}. If extreme tail 
realizations are not fully captured by a forecast series (or if the forecast series itself is heavy-tailed) the 
corresponding forecast errors and loss differentials inherit these 
heavy tails.

The mismatch has substantial consequences for inference. In a 
simulation experiment, a standard DM test at the 5\% nominal level 
rejects a true null hypothesis as often as 70\% of the time, 
irrespective of the sample size, under empirically plausible 
heavy-tailed data-generating processes. In our empirical application to 
risk forecasts for emerging-market foreign exchange rates -- where the 
loss differentials appear to have infinite variance -- the standard DM 
test would lead a researcher to conclude that a 
GARCH-based risk model statistically dominates each of nine competing 
forecasts. A heavy-tail-robust alternative reaches a substantively 
different conclusion: simple non-parametric rolling-window risk 
forecasts cannot be rejected as inferior to any of the more 
sophisticated GARCH-based or score-driven competitors.

A central problem in the forecast evaluation literature is to conduct 
inference about the mean of a time series $(X_t)_{t=1,\dots,n}$, where 
$X_t$ is a forecast error, a forecast loss, a loss differential, or a 
transformation thereof. We focus first on testing
\begin{equation}\label{eq:zero_mean}
    \mathsf{H}_0:\mathbb{E}[X_t]=0,
\end{equation}
against $\mathsf{H}_\mathrm{A}:\mathbb{E}[X_t]\neq 0$ with 
$\mathbb{E}[|X_t|]<\infty$. Depending on how $X_t$ is constructed, 
\eqref{eq:zero_mean} covers hypotheses such as forecast unbiasedness, 
forecast optimality, forecast encompassing, and, most prominently, the 
equal predictive ability (EPA) hypothesis of \cite{diebold1995paring}\footnote{We focus on the unconditional predictive ability hypothesis; the conditional version of \cite{GiacominiWhite2006} is beyond the scope of our analysis. We also abstract from parameter estimation error in the forecasts; see \cite{west1996asymptotic} for asymptotic theory under estimation uncertainty.}. 
Standard practice is to base inference on
\begin{equation}\label{DM_test}
    T_n=\dfrac{\overline{X}_n}{\hat{\sigma}_n/\sqrt{n}},
\end{equation}
where $\overline{X}_n$ is the sample average and $\hat{\sigma}_n^2$ is 
typically a heteroskedasticity- and autocorrelation consistent (HAC) 
estimator of the long-run variance of $\overline{X}_n$, and to compare 
$T_n$ against standard normal critical values. The \cite{ibragimov1962clt} CLT 
provides the foundational justification for this practice, and 
underlies a large body of work on forecast evaluation, including 
\cite{white2000snooping}, \cite{hansen2005SPA}, \cite{hansen_lunde2005}, \cite{GiacominiWhite2006}, \cite{hansen_lunde_nason2011MCS}, \cite{Patton_timmermann_2012rationality}, 
\cite{BarendsePatton2022}, and \cite{odendahl2023forecast}.

This paper makes three main contributions. \textit{First}, we develop a 
new stable limit theorem for infinite-variance strongly mixing time series. The theorem parallels the structure of the 
\cite{ibragimov1962clt} CLT: it imposes regular variation ($\kappa<2$), a mild 
anti-clustering condition on extremes, and weak rates on mixing 
coefficients -- conditions that hold for a range of standard linear and 
nonlinear processes, including autoregressive, GARCH, and stochastic 
volatility models. The result thus provides the foundational analog to 
\cite{ibragimov1962clt}, that the existing forecast-evaluation literature 
requires, when loss differentials have infinite variance.

\textit{Second}, we use this result to characterize the limiting 
distribution of the statistic $T_n$ in \eqref{DM_test} under 
$\mathsf{H}_0$ when $(X_t)_{t=1,\dots,n}$ has tail index $\kappa<2$. We 
show that $T_n$ converges in distribution to a ratio of two dependent 
non-Gaussian stable random variables, whose joint distribution is fully 
characterized by $\kappa$ and the extremal cluster process of 
$(X_t)_{t=1,\dots,n}$. The resulting ratio may be skewed, and its 
quantiles can differ substantially from those of the standard normal 
distribution used in practice.

\textit{Third}, we consider an inferential procedure that is valid under 
\textit{both} finite- and infinite-variance regimes. We construct 
critical values by subsampling \citep{politis2012subsampling} and 
establish their asymptotic validity from primitives. The procedure is 
fully data-driven: it does not require knowledge of $\kappa$, the 
convergence rate of the sample mean, or consistent estimation of 
long-run variances. The same approach yields confidence intervals for 
$\mathbb{E}[X_t]$, which may itself be of interest in forecast 
evaluation.

We extend the analysis in two further directions. We turn first to the 
SPA tests of \cite{white2000snooping} and \cite{hansen2005SPA}, which 
address the more demanding problem of deciding whether a benchmark 
forecast is dominated by any one of $m$ competing forecasts. We provide 
a multivariate version of our limit theory and derive the limiting 
distribution of the SPA test statistic under heavy tails, which lays 
the groundwork for a heavy-tail-robust SPA test based on subsampling. 
This complements the conditional SPA test of \cite{Li2021CSPA}, which, 
like the original, assumes finite higher-order moments. 

We then consider alternatives under which $\mathbb{E}[|X_t|]$ is infinite, so that $\mathbb{E}[X_t]$ does not exist or equals $\pm\infty$. We show that the statistic $T_n$ in \eqref{DM_test} does not diverge under such alternatives, so that any test based on it -- whether using standard 
normal or subsampled critical values -- is inconsistent. We propose a modified test statistic that diverges whenever $\mathbb{E}[X_t]$ differs from zero or is undefined, together with an accompanying subsampling-based critical value construction. To our knowledge, such alternatives have not been addressed in the existing literature on inference about the mean of heavy-tailed time series.

Our work is most closely related to the working paper by 
\cite{kim2021forecast}, who also consider the DM test under heavy 
tails. \cite{kim2021forecast} obtain a non-Gaussian stable limit for 
the DM statistic by direct application of \cite{Davis1983stablelimit} 
and \cite{davis1995point} and propose subsampling for inference, 
invoking the validity result of \cite{kokoszka2004subsampling}. They 
also provide an analytical treatment of how the tail index of the loss 
differential depends on the tail indices of the target variable, the 
forecasts, and the loss function, and they illustrate empirically the 
relevance of subsampling for realized volatility forecast comparisons. 
Our analysis differs in three respects: (i) building upon 
\cite{matsui2025self}, we provide a new stable limit theorem for 
strongly mixing processes that explicitly accounts for the extremal dependence of the underlying time series process. Our assumptions are 
strictly more general than those used in \cite{kim2021forecast}, 
\cite{McElroy_Politis_2002_subsampling}, and 
\cite{kokoszka2004subsampling}, and we verify them for a range of 
standard linear and nonlinear time series processes including 
autoregressive, GARCH, and stochastic volatility models; (ii) our 
analysis extends beyond the EPA setup to forecast encompassing, 
optimality, and SPA testing, as well as alternatives under which the 
mean may fail to exist; and (iii) we provide rigorous arguments for the 
asymptotic validity of the subsampling-based tests.

The remainder of the paper is organized as follows. 
Section \ref{sec:limit_theory} presents the stable limit theory and the 
resulting limiting distribution of $T_n$ (Proposition \ref{Master_Lemma}). 
Section \ref{sec:subsampling} presents the subsampling algorithm and 
establishes its asymptotic validity. Section \ref{sec:extensions} 
contains the extensions to alternatives with undefined mean and to the 
SPA test. Sections \ref{sec:Monte_Carlo} and \ref{sec:empirical} report 
the simulation experiments and the empirical illustration. The proof of the main Theorem \ref{theo:inf_var_jointCLT} in the Appendix; additional proofs and simulation results are in 
the Supplemental Appendix.

\subsection*{Notation}
For a random variable $X$, unless stated otherwise, $\mathbb{E}[X] \neq 0$ 
if either $\mathbb{E}[|X|]=\infty$ or that $\mathbb{E}[|X|]<\infty$ with 
$\mathbb{E}[X]$ non-zero. For two real-valued functions $f$ and $g$ we 
write $f(x)\sim g(x)$ as $x\to \infty$ if $\lim_{x\to \infty}f(x)/g(x)=1$. 
We generically denote real-valued, \textit{slowly varying} functions by 
$L$, that is, the functions satisfy that for any constant $c>0$, 
$L(cx)\sim L(x)$ as $x\to\infty$. For short, a strictly stationary process 
is referred to as a stationary process.
\newpage

\section{Set-up and limit theory}\label{sec:limit_theory}

This section develops the limit theory underlying our inferential procedure. Section \ref{sec:setup} provides some motivating examples and previews the main result. Sections \ref{sec:limit_partial_sums}--\ref{sec:limit_partial_sums_hv} state the partial-sum limit theory under finite variance and heavy tails, respectively. Section \ref{sec:limit_T_n} combines these into the limiting distribution of $T_n$ in \eqref{DM_test}, and Section \ref{sec:AR_details} illustrates the theory in the context of linear autoregressive processes.

\subsection{Set-up and preview of main result}\label{sec:setup}
Given a random sample $(X_t)_{t=1,\dots,n}$ with a time-invariant mean, we consider testing the hypothesis $\mathsf{H}_0$ in \eqref{eq:zero_mean} against $\mathsf{H}_\textrm{A}:\mathbb{E}[X_t]\neq0$ with $\mathbb{E}[|X_t|]<\infty$. Following much of the existing literature on forecast evaluation and comparisons, we throughout take $X_t$ as a \textit{primitive}. Depending on the construction of $X_t$, the hypothesis test resembles those often encountered in the forecasting literature. We begin with three motivating examples.

\begin{example}[Forecast Optimality]
Let $(f_{t})_{t=1,\dots,n}$ denote a forecast series of a time series $(Y_t)_{t=1,\dots,n}$, and let $e_{t}=Y_t-f_{t}$ denote the forecast error. The forecast is said to be \textit{unbiased}, if $\mathbb{E}[e_t]=0$, and it is said to be \textit{efficient} if $\mathbb{E}[e_{t}f_{t}]=0$; see, e.g., \citet[p.~6]{odendahl2023forecast}. Hence, forecast unbiasedness or forecast efficiency corresponds to the hypothesis $\mathsf{H}_0$ in \eqref{eq:zero_mean} when $X_t$ is given by $e_t$ or $e_{t}f_{t}$, respectively. In both cases, $X_t$ may be heavy-tailed if the forecast errors are heavy-tailed. \hfill $\blacklozenge$
\end{example}

\begin{example}[Forecast Encompassing]
Let $(f_{1,t})_{t=1,\dots,n}$ and $(f_{2,t})_{t=1,\dots,n}$ denote two forecast series of a time series $(Y_t)_{t=1,\dots,n}$, and let $e_{i,t}=Y_t-f_{i,t}$ denote the forecast error of series $i=1,2$. Following \cite{harvey1998encompassing}, forecast series 1 is said to \textit{encompass} forecast series 2 if $\mathbb{E}[(e_{1,t}-e_{2,t})\,e_{1,t}]=0$. Hence, forecast encompassing means that  $\mathsf{H}_0$ holds with $X_t=(e_{1,t}-e_{2,t})\,e_{1,t}$,
which may be heavy-tailed whenever at least one of the forecast errors are. \hfill $\blacklozenge$
\end{example}

\begin{example}[Equal Predictive Ability (EPA)]\label{ex:EPA}
Consider two forecast series $(f_{1,t})_{t=1,\dots,n}$ and $(f_{2,t})_{t=1,\dots,n}$ of a time series $(Y_t)_{t=1,\dots,n}$. For a given loss function $\mathcal{L}:\mathbb{R}^2\to\mathbb{R}_+$, the two forecast series are said to have equal predictive ability (EPA) if
\begin{equation*}
    \mathbb{E}[\mathcal{L}(Y_t,f_{1,t})-\mathcal{L}(Y_t,f_{2,t})]=0;
\end{equation*}
see \cite{diebold1995paring}. Hence, the EPA hypothesis corresponds to the hypothesis $\mathsf{H}_0$ in \eqref{eq:zero_mean} with $X_t=\mathcal{L}(Y_t,f_{1,t})-\mathcal{L}(Y_t,f_{2,t})$. 

A widely used loss function is the squared error, $\mathcal{L}(Y_t,f_{1,t})=(Y_t-f_{1,t})^2=e_{1,t}^2$. Here, if the forecast error $e_{1,t}$ has tail index $\kappa$, then -- by construction -- the squared error $e_{1,t}^2$ has tail index $\kappa/2$, so that $X_t$ may have infinite variance even when $e_{1,t}$ or $e_{2,t}$ themselves have finite variance. 
\medskip

Heavy-tailed losses may also occur when evaluating quantile forecasts as considered in the empirical application in Section \ref{sec:empirical}. Let $f_{1,t}$ denote a $\tau$-level (conditional) quantile forecast of $Y_t$. The associated \textit{tick loss} \citep{giacomini_komunjer2005quantile} is
\begin{equation*}
    \mathcal{L}_{\mathrm{tick},\tau}(Y_t,f_{1,t})=(\tau-\mathbf{1}\{Y_t-f_{1,t}<0\})(Y_t-f_{1,t})=(\tau-\mathbf{1}\{e_{1,t}<0\})e_{1,t}.
\end{equation*}
A common way to construct quantile forecasts for financial returns (or losses) is by means of GARCH-type models. Under mild conditions, if $Y_t$ follows a GARCH process, then $Y_t$ has a tail index $\kappa>0$ and the tick losses inherit this tail index; see Section \ref{sec:empirical} for further details and an empirical illustration with emerging-market foreign exchange rates, for which the range $\kappa\in(1,2)$ is empirically relevant. \hfill $\blacklozenge$
\end{example}

 In order to test $\mathsf{H}_0$, we consider the statistic $T_n$ in \eqref{DM_test} with $\hat{\sigma}_n^2=n^{-1}\sum_{t=1}^nX_t^2$. In contrast to \cite{diebold1995paring}, we do not use a HAC estimator. This is deliberate: under heavy tails the sample mean has infinite (or undefined) long-run variance, so consistent estimation of this quantity is neither possible nor needed for our purposes. The test statistic $T_n$ is, hence, given by the studentized sum
\begin{equation}\label{self_sum}
    T_n=\dfrac{S_n}{\gamma_{n}},
\end{equation}
with
\begin{equation} \label{eq:def_S_gamma}
    S_n := \sum_{t=1}^nX_t\qquad\text{and}\qquad \gamma_{n}:=\sqrt{\sum_{t=1}^nX_t^2}\;.
\end{equation}
The normalization of $S_n$ by $\gamma_n$ is important, since these quantities are of the same stochastic order irrespective of the tail-heaviness of $X_t$; in particular, the statistic does not require knowledge about the rate of convergence of $S_n/n$.\smallskip

To preview the main results of this section, we briefly summarize the limiting behavior of $T_n$ under $\mathsf{H}_0$. Given appropriate regularity conditions,
\begin{itemize}
    \item as predominantly exploited in the existing body of literature on forecast evaluation and comparisons, if $X_t$ has \textit{finite} variance, then $T_n$ has a \textit{Gaussian} limit;
    \item if $X_t$ is heavy-tailed with tail index $\kappa\in(1,2)$, so that $X_t$ has infinite variance, then $T_n$ converges in distribution to the ratio $\xi_\kappa/\zeta_{\kappa/2}$, where $\xi_\kappa$ is a \textit{non-Gaussian} $\kappa$-stable random variable and $\zeta_{\kappa/2}^2$ is a $(\kappa/2)$-stable random variable. 
\end{itemize}
The formal statement is given in Proposition \ref{Master_Lemma} in Section \ref{sec:limit_T_n}, which is a direct consequence of the partial-sum limit theorems in Sections \ref{sec:limit_partial_sums}--\ref{sec:limit_partial_sums_hv}. Complementing this result, in Proposition \ref{prop:AR_HAC} we show -- in the context of heavy-tailed linear autoregressive processes -- that the same non-Gaussian limit governs the HAC-based variant of $T_n$ that is routinely used in applied forecast comparisons, so that standard normal critical values are potentially unreliable also for this statistic.

\begin{remark}
    Throughout, we focus on testing $\mathsf{H}_0$ in \eqref{eq:zero_mean}, but our results also enable the construction of confidence intervals for $\mathbb{E}[X_t]$ when $X_t$ may have infinite variance. This is useful, e.g., for assessing the accuracy of a forecast series: setting $X_t = \mathcal{L}(Y_t, f_t)$, the quantity $\mathbb{E}[X_t]$ represents the expected loss, such as the mean squared forecast error. See Remark \ref{rem:confidence_interval} for details.
\end{remark}

\subsection{Limit theory for partial sums under finite variance}\label{sec:limit_partial_sums}
We assume throughout that the data-generating process (DGP) for $X_t$ is \textit{stationary}, and we write $(X_t)_{t\in\mathbb{Z}}$ to signify that the process is indexed over all integers.

We start by considering the finite variance case. Before stating our assumptions, we recall the definition of strong mixing coefficients,  characterizing the time-dependence of a real-valued stationary process $(X_t)_{t\in\mathbb{Z}}$. Let $g_1$ and $g_2$ be real-valued, bounded functions acting on $(X_t)_{t\in\mathbb{Z}}$ and define
\begin{equation}\label{mix_coef}
\alpha_k:=\sup_{|g_1|,|g_2|\leqslant 1}\left|\mathrm{cov}\left[g_1(\dots,X_{-1},X_{0}),g_2(X_{k},X_{k+1},\dots)\right]\right|,\qquad k \geqslant 0.
\end{equation}
Following \citet[p.3]{Doukhan1994mixing} the numbers $(\alpha_k)_{k\geqslant0}$ are the strong mixing coefficients of $(X_t)_{t\in\mathbb{Z}}$, and the process is said to be strongly mixing (or $\alpha$-mixing) if $\alpha_k\to 0$ as $k\to \infty$.
In order to derive the asymptotic properties of $T_n$ in \eqref{self_sum} under finite variance of $X_t$, we consider the following assumption.
\begin{assumption}\label{mix_coef_ass}
    There exists $\varepsilon>0$ such that $\mathbb{E}[|X_t|^{2+\varepsilon}]<\infty$ and the  $\alpha$-mixing coefficients of $(X_t)_{t\in\mathbb{Z}}$ satisfy 
    \begin{equation}\label{eq:mix_summability}
        \sum_{k=1}^\infty\alpha_k^{\varepsilon/(2+\varepsilon)}<\infty.
    \end{equation}
\end{assumption}

Assumption \ref{mix_coef_ass} is standard in the literature on forecast evaluations. The summability condition in \eqref{eq:mix_summability}  holds for any $\varepsilon>0$, if the mixing coefficients $(\alpha_k)_{k\geqslant0}$ decay at a geometric rate. Therefore, a wide range of well-known time series processes, including causal stationary solutions to linear autoregressive, GARCH-type and stochastic volatility processes, satisfy the summability condition; see Section \ref{sec:DGPs} in the Supplemental Appendix for details. On the other hand, the moment condition $\mathbb{E}[|X_t|^{2+\varepsilon}]<\infty$ typically imposes additional restrictions on the process, particularly that $X_t$ has finite variance. With $\overline{\varepsilon}:=\inf\{\varepsilon \in(0,\infty]:\mathbb{E}[|X_t|^{2+\varepsilon}]=\infty\}$, we note that a smaller $\overline{\varepsilon}$ imposes a faster decay on $\alpha_k$ for \eqref{eq:mix_summability} to hold. Consequently, Assumption \ref{mix_coef_ass} balances the order of finite moments and the rate of decay of the mixing coefficients. The assumption implies that a CLT and the ergodic theorem apply to $(X_t - \mathbb{E}[X_t])$ and $(X_t - \mathbb{E}[X_t])^2$, respectively: 
\begin{theorem}[\cite{ibragimov1962clt}]\label{theo:rio}
    Let $(X_t)_{t\in\mathbb{Z}}$ be a stationary process satisfying Assumption \ref{mix_coef_ass}. Then
    \begin{equation*}
        \dfrac{1}{\sqrt{n}}\left(\sum\nolimits_{t=1}^n(X_t-\mathbb{E}[X_t]),\sqrt{\sum\nolimits_{t=1}^n{(X_t-\mathbb{E}[X_t])}^2}\right)\xrightarrow{d}\left(Z,\sqrt{\mathrm{var}(X_t)}\right),\qquad n\rightarrow\infty,
    \end{equation*}
    where $Z$ is a $\mathrm{N}(0,\sigma^2)$-distributed random variable with 
     \begin{equation}\label{eq:def_sigma2}
        \sigma^2:= \mathrm{var}(X_t)+2\sum_{t=1}^\infty\;\mathrm{cov}(X_0,X_t),\qquad 0 <\sigma^2<\infty.
    \end{equation}
\end{theorem}

\smallskip

\subsection{Limit theory for partial sums under heavy tails}\label{sec:limit_partial_sums_hv}

We now turn to the case where $X_t$ has infinite second moment. To formalize what we mean by heavy tails, we recall the notion of regular variation: a real-valued random variable $X$ is said to be regularly varying with tail index $\kappa>0$ if
\begin{equation}\label{RV_marg}
    \mathbb{P}(X>x)\sim p_+\dfrac{L(x)}{x^\kappa}\quad\quad\text{and}\quad\quad \mathbb{P}(X<-x)\sim p_-\dfrac{L(x)}{x^\kappa},\quad\quad x\to\infty,
\end{equation}
where $L$ is a slowly varying function and the tail-balance coefficients $p_+,p_-\geqslant0$ satisfy $p_++p_-=1$. Importantly, under \eqref{RV_marg}, $\mathbb{E}[{|X|}^{p}]=\infty$ $(<\infty)$ for every $p>\kappa$ $(p<\kappa)$, and, hence, a smaller $\kappa$ implies that $X$ has heavier tails.

Following \cite{basrak2009regularly}, a stationary real-valued process $(X_t)_{t\in\mathbb{Z}}$ is said to be regularly varying with index $\kappa>0$, if $X_0$ is regularly varying with index $\kappa>0$ in the sense of \eqref{RV_marg} and there exists a real-valued stochastic process $(\mathit{\Theta}_t)_{t\in\mathbb{Z}}$ such that for all integers $h\geqslant0$,
\begin{equation}\label{RV_TS}
    \mathbb{P}(|X_0|^{-1}(X_{-h},\dots,X_{h})\in\cdot\;|\;|X_0|>x)\xrightarrow{w}\mathbb{P}((\mathit{\Theta}_{-h},\dots,\mathit{\Theta}_{h})\in\cdot),\quad\quad x\to\infty,
\end{equation}
where $\overset{w}{\to}$ denotes weak convergence of probability measures. The process $(\mathit{\Theta}_t)_{t\in\mathbb{Z}}$ is called the \emph{spectral tail process} of $(X_t)_{t\in\mathbb{Z}}$ and characterizes the extremal dependence structure of $(X_t)_{t\in\mathbb{Z}}$. Informally, conditional on a very large excursion at time $0$, $(\mathit{\Theta}_t)_{t\in\mathbb{Z}}$ records the shape of the surrounding cluster of large values, normalized by $|X_0|$. An iid\ process has $\mathit{\Theta}_t=0$ almost surely for all $t\neq0$, while persistent processes give rise to non-degenerate $\mathit{\Theta}_t$. By construction, $\mathbb{P}(|\mathit{\Theta}_0|=1)=1$ with $\mathbb{P}(\mathit{\Theta}_0=\pm1)=p_{\pm}\geqslant 0$. For more details on regularly varying time series we refer to \citet[Chapter 4]{mikosch2024extreme}.

With this notion in hand, we impose the following assumption.
\begin{assumption}\label{ass_RV}
    The process $(X_t)_{t\in\mathbb{Z}}$ is regularly varying with index $\kappa\in(0,1)\cup(1,2)$ in the sense of \eqref{RV_TS}.
\end{assumption}
We note that for $\kappa\in(1,2)$, $X_t$ has infinite variance, whereas the variance is not defined for $\kappa\in(0,1)$.\footnote{As is standard in the literature on stable limit theory for dependent processes, we rule out the case $\kappa=1$. The case $\kappa\in(0,1)$ is covered both for completeness and in view of the undefined-mean alternatives studied in Section \ref{sec:extensions}.} Under Assumption \ref{ass_RV} there exists a deterministic sequence $(a_n)_{n\geqslant1}$ satisfying $a_n\to\infty$ and
\begin{equation} \label{eq:def_a_n}
    n\prob(|X_t|>a_n)\rightarrow1,\qquad n\rightarrow\infty.
\end{equation}
In particular, by regular variation of $X_t$, it holds that $a_n\sim L(n)n^{1/\kappa}$ as $n\to\infty$ for some slowly varying function $L$.
Given Assumption \ref{ass_RV}, we  make the following assumption about the dependence structure of $(X_t)_{t\in\mathbb{Z}}$.  
\begin{assumption}\label{ass:AC+MX}
    Let $(a_n)_{n\geqslant1}$ satisfy \eqref{eq:def_a_n}. There exists a deterministic sequence $(r_n)_{n\geqslant1}$ with $r_n\to\infty$ and $r_n=o(n)$ as $n\to\infty$ such that
    \begin{equation}\label{AC}
\lim_{k\rightarrow\infty}\limsup_{n\rightarrow\infty}n\sum_{t=k}^{r_n}\mathbb{E}[\min\{|a_n^{-1}X_t|,1\}\;\min\{|a_n^{-1}X_0|,1\}]=0.
    \end{equation}
    Furthermore, the strong mixing coefficients $(\alpha_k)_{k\geqslant0}$ of $(X_t)_{t\in\mathbb{Z}}$ satisfy
    \begin{equation}\label{alpha_coefs_req}
        \dfrac{n}{r_n}\alpha_{\ell_n}\rightarrow0,\qquad n\to\infty,
    \end{equation}
    for some deterministic sequence $(\ell_n)_{n\geqslant1}$ with $\ell_n\rightarrow\infty$ and $\ell_n=o(r_n)$ as $n\to\infty$.
\end{assumption}

The condition in \eqref{AC} is a so-called anti-clustering condition; see e.g. \citet[Section 2.3]{bartkiewicz2011stable}. To provide some intuition for this condition, note that \eqref{AC} implies
\begin{equation}\label{eq:AC_implication}
    \lim_{k\rightarrow\infty}\limsup_{n\rightarrow\infty}n\sum_{t=k}^{r_n} \mathbb{P}\left( \vert X_t\vert >a_n \ , \ \vert X_0 \vert > a_n  \right) = 0.
\end{equation}
Hence, informally, if we observe a very large $|X_0| > a_n$, then $(X_t)_{t\geqslant 1}$ cannot stay above this threshold $a_n$ for too long with non-negligible probability; the extremes may cluster, but long clusters become increasingly rare.  

The two parts of Assumption \ref{ass:AC+MX} play complementary roles in the limit theory, and it is useful to separate them conceptually. Regular variation (Assumption \ref{ass_RV}) together with the anti-clustering condition \eqref{AC} delivers a stable limit for the sum of $k_n = n/r_n$ independent block sums of $X_t$  with each block having length $r_n$. The mixing condition \eqref{alpha_coefs_req} then ensures that this sum of independent block sums is a good approximation of the full partial sum $S_n=\sum_{t=1}^nX_t$; we refer to Appendix \ref{sec:appendix_main_theorem} for additional details about this approximation argument.

In addition to playing complementary roles, note that the anti-clustering condition \eqref{AC} and the mixing condition \eqref{alpha_coefs_req} depend on the same deterministic sequence $(r_n)_{n\geqslant1}$. Similar to the role of the constant $\varepsilon>0$ in Assumption \ref{mix_coef_ass}, balancing the tail-heaviness and the rate of decay of the mixing coefficients in the finite variance case, the sequence $(r_n)_{n\geqslant1}$ balances the clustering of extremes against the mixing coefficients: if $(X_t)_{t\in\mathbb{Z}}$ is likely to exhibit long clusters of extreme observations such that $\mathbb{P}\left( \vert X_t\vert >a_n \ , \ \vert X_0 \vert > a_n  \right)>0$ for large $a_n$ and $t$, the anti-clustering condition requires that $r_n$ grows at a sufficiently slow rate. This slow rate must be compensated for by a faster decay of $\alpha_k$ in order for the mixing condition in \eqref{alpha_coefs_req} to hold. 

Assumption \ref{ass:AC+MX} can be shown to hold under mild conditions for a wide range of well-known time series processes, including linear autoregressive processes (as considered in Section \ref{sec:AR_details}), GARCH processes, and other processes given by affine stochastic recurrence equations, as well as stochastic volatility processes; see Section \ref{sec:DGPs} of the Supplemental Appendix for additional details and technical conditions.

Assumptions \ref{ass_RV}--\ref{ass:AC+MX} ensure that $(X_t)_{t\in\mathbb{Z}}$ obeys the following stable limit theorem.
\begin{theorem}\label{theo:inf_var_jointCLT}
    Let $(X_t)_{t\in\mathbb{Z}}$ satisfy Assumptions \ref{ass_RV}--\ref{ass:AC+MX}, and define $\mu:=\mathbb{E}[X_t]\mathbf{1}(\kappa>1)$. With $(a_n)_{n\geqslant1}$ defined by \eqref{eq:def_a_n}, it holds that
    \begin{equation}\label{eq:inf_var_jointCLT}
        a_n^{-1}\left( \sum\nolimits_{t=1}^{n}(X_t-\mu), \sqrt{\sum\nolimits_{t=1}^n(X_t-\mu)^2}  \right) \overset{d}{\to} (\xi_\kappa,\;\zeta_{\kappa/2}),\qquad n\to\infty,
    \end{equation}
    where $\xi_\kappa$ is a $\kappa$-stable random variable, $\mathbb{P}(\zeta_{\kappa/2}>0)=1$, and $\zeta_{\kappa/2}^2$ is a  $\kappa/2$-stable random variable.
\end{theorem}

To the best of our knowledge, the above result is new. Whereas many existing stable limit results rely on point process convergence (e.g., \cite{davis1995point}), our result builds upon recent theoretical results of \cite{matsui2025self} relying on convergence of so-called hybrid characteristic function--Laplace transforms. For a real-valued random variable $V$ and a non-negative random variable $W$, the hybrid characteristic function--Laplace transform is the map
\begin{equation*}
    (u,\lambda)\mapsto \mathbb{E}\big[\mathrm{e}^{iuV-\lambda W}\big],\qquad u\in\mathbb{R},\ \lambda\geq0,
\end{equation*}
which fully characterizes the joint distribution of $(V,W)$ and reduces to the characteristic function of $V$ when $\lambda=0$ and to the Laplace transform of $W$ when $u=0$. This approach typically applies under milder conditions than those imposed to establish point process convergence, and we carefully adapt the results of \cite{matsui2025self} to strongly mixing processes and allow for $\mu\neq 0$. For further details about hybrid characteristic function-Laplace transforms, we refer to \citet[Appendix D.2]{mikosch2024extreme}, and we refer to the textbook by \cite{samorodnitsky1994stable} for a treatment of non-Gaussian stable distributions.
\medskip

The distribution of the limiting vector in \eqref{eq:inf_var_jointCLT} is fully characterized by the tail index and the extremal dependence structure of $(X_t)_{t\in\mathbb{Z}}$. To see this, recall that by Assumption \ref{ass_RV}, $(X_t)_{t\in\mathbb{Z}}$ is regularly varying and hence, by the definition in \eqref{RV_TS}, has a spectral tail process $(\mathit{\Theta_t})_{t\in\mathbb{Z}}$. Define the \textit{extremal cluster process} as
\begin{equation}\label{eq:extremal_cluster_process}
Q_t=\mathit{\Theta}_t/(\sum_{j\in\mathbb{Z}}|\mathit{\Theta}_j|^\kappa)^{1/\kappa},\qquad t\in\mathbb{Z},
\end{equation}
which is the spectral tail process re-normalized so that $\sum_{t\in\mathbb{Z}}|Q_t|^\kappa=1$.
Then it holds that $(\xi_\kappa,\zeta_{\kappa/2}^2)$ has hybrid characteristic function--Laplace transform given by
\begin{align}
        &\mathbb{E}\big[\mathrm{e}^{iu\xi_\kappa-\lambda\zeta_{\kappa/2}^2}\big] \nonumber \\
        &=\exp\Big(\int_0^\infty \mathbb{E}\Big[\mathrm{e}^{iyu\sum_{t\in\mathbb{Z}}Q_t-y^2\lambda\sum_{t\in\mathbb{Z}}|Q_t|^2}-1-iyu\sum_{t\in\mathbb{Z}}Q_t\boldsymbol{1}\big\{\kappa\in(1,2)\big\}\Big]d(-y^{-\kappa})\Big), \label{eq:hybrid_chf_Laplace}
\end{align}
with $u\in\mathbb{R}$ and $\lambda\geq0$. In particular, it holds that $\xi_\kappa$ and $\zeta_{\kappa/2}$ are dependent. Setting $\lambda=0$ yields the characteristic function of the $\kappa$-stable random variable $\xi_\kappa$,
\begin{equation}\label{eq:chr fct of xi}
    \mathbb{E}[\mathrm{e}^{iu\xi_\kappa}]=\exp\big(-|u|^\kappa\tilde{\sigma}^\kappa(1-i\;\beta\;\textrm{sign}(u)\tan(\kappa\pi/2))\big),\qquad u\in\mathbb{R},
\end{equation}
with \textit{scale} and \textit{skewness} parameters given, respectively, by
\begin{equation}\label{eq: scale and skew}
    \tilde{\sigma}=\left(c_\kappa\;\mathbb{E}\big[\big|\sum\nolimits_{t\in\mathbb{Z}}Q_t\big|^\kappa\big]\right)^{1/\kappa}\qquad\text{and}\qquad  \beta=\dfrac{\mathbb{E}\big[\big(\sum\nolimits_{t\in\mathbb{Z}}Q_t\big)_+^\kappa-\big(\sum\nolimits_{t\in\mathbb{Z}}Q_t\big)_-^\kappa\big]}{\mathbb{E}\big[\big|\sum\nolimits_{t\in\mathbb{Z}}Q_t\big|^\kappa\big]},
\end{equation}
where
 \begin{equation}\label{eq:def_c_kappa}
  c_\kappa:=\frac{\Gamma(2-\kappa)\cos(\kappa\pi/2)}{1-\kappa},   
 \end{equation}
 and $\Gamma(\cdot)$ the gamma function\footnote{A slightly different expression for the characteristic function for $\xi_\kappa$ can be found in \citet[Section 9.2]{mikosch2024extreme}. For the sake of completeness, we provide detailed derivations of \eqref{eq:chr fct of xi} and \eqref{eq: scale and skew} in Section \ref{sec:chf_scale_skew} of the Supplemental Appendix.}. The distribution of $\xi_\kappa$ is asymmetric whenever the imaginary part of its characteristic function is non-zero, that is, if $\beta\neq0$. Consequently, whether or not the distribution is symmetric depends on the spectral tail process $(\mathit{ \Theta}_t)_{t\in\mathbb{Z}}$ via \eqref{eq:extremal_cluster_process}. For certain processes, such as linear autoregressive processes considered in Section \ref{sec:AR_details} below, $\beta$ is determined entirely by the tail-balance coefficients $p_{\pm} = \mathbb{P}(\mathit{\Theta}_0=\pm1)$. Moreover, note that the imaginary part of the characteristic function tends to zero as $\kappa\uparrow2$. Intuitively, this property is due to the fact that a stable distribution with index of stability $\kappa = 2$ and a location parameter of zero is a normal distribution with zero mean and hence symmetric.

\subsection{Limiting behavior of $T_n$}\label{sec:limit_T_n}

Combining Theorems \ref{theo:rio} and \ref{theo:inf_var_jointCLT} yields the limiting behavior of the test statistic $T_n$ in \eqref{self_sum} under both finite and infinite variance:
\begin{proposition}\label{Master_Lemma}
    Let $(X_t)_{t\in\mathbb{Z}}$ be a stationary process of real-valued random variables, and let $T_n$ be given by \eqref{self_sum}.
    \begin{itemize}
        \item [$(1)$] Suppose that $(X_t)_{t\in\mathbb{Z}}$ satisfies Assumption \ref{mix_coef_ass} and that $\mathbb{E}[X_t]=0$. Then 
        \begin{equation*}
            T_n\xrightarrow{d}\mathrm{N}(0,\sigma^2/\mathbb{E}[X_t^2]),\qquad n\to\infty,
        \end{equation*}
        where $\sigma^2>0$ is given in \eqref{eq:def_sigma2}.
        \item [$(2)$] Suppose that $(X_t)_{t\in\mathbb{Z}}$ satisfies Assumptions \ref{ass_RV}--\ref{ass:AC+MX} with index $\kappa\in(1,2)$ and $\mathbb{E}[X_t]=0$. Then 
        \begin{equation*}
            T_n\xrightarrow{d}\xi_\kappa\;/\zeta_{\kappa/2},\qquad n\to\infty,
        \end{equation*}
        with $\xi_\kappa$ and $\zeta_{\kappa/2}^2$ non-Gaussian stable random variables whose joint distribution is characterized by the hybrid characteristic function--Laplace transform given in \eqref{eq:hybrid_chf_Laplace}.
        \item [$(3)$] Suppose that $(X_t)_{t\in\mathbb{Z}}$ satisfies Assumption \ref{mix_coef_ass} or Assumptions  \ref{ass_RV}--\ref{ass:AC+MX} with index $\kappa\in(1,2)$. If $\mathbb{E}[X_t]\neq0$, then $|T_n|\xrightarrow{\mathbb{P}}\infty$ as $n\to\infty$.
    \end{itemize}
\end{proposition}
Cases (1) and (2) give the limiting distribution of $T_n$ under $\mathsf{H}_0$ for different tail regimes, and the two limits may differ substantially in shape. Because $\zeta_{\kappa/2}$ is strictly positive (as discussed in Section \ref{sec:limit_partial_sums_hv}), the limit in case (2) is asymmetric whenever the skewness parameter $\beta$ in \eqref{eq: scale and skew} is non-zero, whereas the Gaussian limit in case (1) is always symmetric. This is investigated in more detail in Section \ref{sec:AR_details} in the context of linear autoregressive processes. 
Despite being non-Gaussian, the ratio $\xi_\kappa/\zeta_{\kappa/2}$ can be shown to have all moments finite for a wide class of processes, see \cite{matsui2025moments}; nevertheless, as documented in the simulation study in Section \ref{sec:Monte_Carlo}, its quantiles can depart markedly from those of a standard normal distribution, so that comparing $T_n$ against Gaussian critical values can produce a substantially distorted test. %\textcolor{red}{CURRENT: Moreover, in the simulation study in Section \ref{sec:Monte_Carlo}, we find that the quantiles of the case-(2) limit can depart markedly from those of a standard normal distribution, so that comparing $T_n$ against Gaussian critical values can produce a substantially distorted test.} %To obtain a test based on $T_n$ that is valid in both cases (1) and (2), we consider in Section \ref{sec:subsampling} a subsampling scheme for critical value construction. 

\subsection{Example: Linear autoregressive processes}\label{sec:AR_details}
We end this section by considering additional details about the results in Theorem \ref{theo:inf_var_jointCLT} and Proposition \ref{Master_Lemma} in the case where $(X_t)_{t\in\mathbb{Z}}$ follows an autoregressive process of order one. This example also provides some supporting intuition for the results of the simulation experiment in Section \ref{sec:Monte_Carlo}.
Let the real-valued $X_t$ obey the recursion
\begin{equation}\label{eq:AR1_sec2}
    X_t=\delta +\varphi X_{t-1}+Z_t,\qquad t\in\mathbb{Z},
\end{equation}
where $\delta\in\mathbb{R},\;\varphi\in[0,1)$ and $(Z_t)_{t\in\mathbb{Z}}$ is an iid sequence of random variables with strictly positive Lebesgue density on $\mathbb{R}$ and $\mathbb{E}[{|Z_t|}^\varepsilon]<\infty$ for some $\varepsilon>0$.\footnote{The non-negativity restriction on the autoregressive coefficient $\varphi$ is made entirely for the ease of exposition.} We also assume that $\mathbb{E}[Z_t]=0$ whenever $\mathbb{E}[|Z_t|]<\infty$. These conditions ensure that there exists a  stationary causal solution $(X_t)_{t\in\mathbb{Z}}$ to \eqref{eq:AR1_sec2}, which is strongly mixing with geometric rate; see, e.g., Proposition 2.2.4 in \cite{buraczewski2016stochastic}. Consequently, if $\mathbb{E}[{|Z_t|}^{2+\varepsilon}]<\infty$ for some $\varepsilon>0$, then $(X_t)_{t\in\mathbb{Z}}$ satisfies Assumption \ref{mix_coef_ass}. On the other hand, if $Z_t$ is regularly varying with some tail index $\kappa>0$, then $(X_t)_{t\in\mathbb{Z}}$ is regularly varying with the same index $\kappa$: 
\begin{proposition}\label{prop:AR1 - RV}
    Assume that $Z_t$ is regularly varying with tail index $\kappa>0$ and tail-balance coefficients $p_\pm\geqslant0$ with $p_++p_-=1$. Then there exists a stationary and regularly varying solution $(X_t)_{t\in\mathbb{Z}}$ to \eqref{eq:AR1_sec2} with the same index $\kappa$ and has spectral tail process given by
    \begin{align}\label{AR_spec_process}
        \mathit{\Theta}_t=\mathit{\Theta}_0\;\varphi^t\mathbf{1}(t\geqslant -J),\qquad t\in\mathbb{Z},
    \end{align}
    where $J$ is a non-negative integer-valued random variable satisfying
    \begin{equation}
        \mathbb{P}(J=j)=\varphi^{\kappa j}(1-\varphi^\kappa),\qquad j=0,1,\dots,
    \end{equation}
    and $\mathbb{P}(\mathit{\Theta}_0=\pm1)=p_\pm$. If $\kappa\in(0,2)$ then $(X_t)_{t\in\mathbb{Z}}$ satisfies Assumption \ref{ass:AC+MX}.
\end{proposition}
By Proposition \ref{prop:AR1 - RV}, $(X_t)_{t\in\mathbb{Z}}$ satisfies Assumptions \ref{ass_RV}--\ref{ass:AC+MX} whenever $Z_t$ is regularly varying with tail index $\kappa\in(0,1)\cup(1,2)$, and we have the following result. 

\begin{proposition}\label{prop:AR_lim_moments}
    Suppose that $(X_t)_{t\in\mathbb{Z}}$ is given by \eqref{eq:AR1_sec2} and satisfies the conditions of Proposition \ref{prop:AR1 - RV} with tail index $\kappa\in(0,1)\cup(1,2)$ and tail-balance coefficients $p_\pm\in[0,1]$. Assume also that the constant term $\delta=0$ if $\kappa\in(1,2)$. Then
    \begin{equation}\label{eq:weak_AR1}
    T_n\xrightarrow{d}\xi_\kappa/\zeta_{\kappa/2},\qquad n\to\infty.
    \end{equation}
    The random variable $\xi_\kappa/\zeta_{\kappa/2}$ has all moments finite,  and  the skewness parameter of $\xi_\kappa$ is given by $\beta=p_+-p_-$. Moreover, if $\kappa\in(1,2)$, then
    \begin{align}
        \mathbb{E}\left[\dfrac{\xi_\kappa}{\zeta_{\kappa/2}}\right]&=(p_+-p_-)\;\dfrac{\sqrt{1+\varphi}}{\sqrt{1-\varphi}}\;\dfrac{\Gamma((1-\kappa)/2)}{\sqrt{\pi}\;\Gamma(1-\kappa/2)}, \label{eq:limit_mean} \\[0.1in]
        \mathbb{E}\left[\Big(\dfrac{\xi_\kappa}{\zeta_{\kappa/2}}\Big)^2\right]&=\dfrac{1+\varphi}{1-\varphi}\left(1+(p_+-p_-)^2\;\dfrac{\kappa}{2}\;\dfrac{\Gamma((1-\kappa)/2)}{\Gamma(1-\kappa/2)}\right), \label{eq:limit_second_moment}
    \end{align}
 and $\xi_\kappa/\zeta_{\kappa/2}$ has a continuous Lebesgue density.
\end{proposition}

Proposition \ref{prop:AR_lim_moments} states that (under suitable conditions) the distribution of $\xi_\kappa/\zeta_{\kappa/2}$ has all moments finite, suggesting that the distribution is thin-tailed. Note that, in contrast to the standard normal distribution, the mean of $\xi_\kappa/\zeta_{\kappa/2}$ may be non-zero, depending on the tail-balance coefficients. % This feature was investigated by \cite{matsui2025moments} who demonstrate that the tails of the density of $\xi_\kappa/\zeta_{\kappa/2}$ (for $\kappa\in(1,2)$) are (asymptotically) of the same order as the tails of a standard normal density, and that certain $\epsilon$-quantiles of $\xi_\kappa/\zeta_{\kappa/2}$ may be smaller (in absolute value) than the $\epsilon$-quantiles of the (suitably scaled) normal distribution.\medskip

In the simulation experiment in Section \ref{sec:Monte_Carlo}, we consider the properties of the standard \cite{diebold1995paring}  test when $X_t$ is heavy tailed. The Diebold--Mariano test is based on the test statistic in \eqref{DM_test} with $\hat{\sigma}_n^2$ some HAC estimator and with critical values given by quantiles from the standard normal distribution. Typically, $\hat{\sigma}_n^2$ is the estimator by \cite{newey_west_1987}, and for the sake of simplicity we here let
\begin{align}\label{eq:T_HAC} 
    T_n^{\mathrm{HAC}}=\frac{S_n}{\sqrt{n}\hat{\sigma}_n}, \qquad \hat{\sigma}^2_n=n^{-1}\sum_{t=1}^nX_t^2+2\sum_{j=1}^{q}w_j\left(n^{-1}\sum_{t=1}^{n-j}X_tX_{t+j}\right),
\end{align}
for a fixed integer $q>0$  and fixed weights $w_j>0$, $j=1,\dots,q$.
We have the following result.
\begin{proposition}\label{prop:AR_HAC}
Suppose that $(X_t)_{t\in\mathbb{Z}}$ is given by \eqref{eq:AR1_sec2} and satisfies the conditions of Proposition \ref{prop:AR1 - RV} with tail index $\kappa\in(0,1)\cup(1,2)$, and that $\delta=0$ if $\kappa\in(1,2)$. With $T_n^{\mathrm{HAC}}$ given by \eqref{eq:T_HAC}, it holds that
\begin{equation}\label{eq:sim_limit_ar}         T_n^{\mathrm{HAC}}\xrightarrow{d}\left(1+2\sum_{j=1}^qw_j\varphi^j\right)^{-1/2}\dfrac{\xi_\kappa}{\zeta_{\kappa/2}},\qquad n\to\infty,
\end{equation}
with $\xi_\kappa/\zeta_{\kappa/2}$ the limiting distribution in \eqref{eq:weak_AR1}.
\end{proposition}

\section{Critical value construction}\label{sec:subsampling}
In this section we consider subsampling-based techniques for testing the hypothesis $\mathsf{H}_0$ in \eqref{eq:zero_mean} based on the test statistic $T_n$. Subsampling-based inference about the mean of a heavy-tailed time series has been considered in, e.g., \cite{McElroy_Politis_2002_subsampling}, \cite{kokoszka2004subsampling} and \cite{bai2016unified}. While these papers consider the construction of confidence intervals for $\mathbb{E}[X_t]$, our emphasis is on hypothesis testing. In particular, we extend the setting of \citet[Chapter 3.5]{politis2012subsampling} to hypothesis testing for heavy-tailed time series with unknown rate of convergence of the sample mean, and we construct critical values based on subsample versions of the original test statistic. Consider the following algorithm.
\begin{algorithm} \label{algo}
    Given a sample $(X_t)_{t=1,\dots,n}$ of size $n > 1$, choose some integer $b_n\in (0,n)$ and some nominal level $\eta\in(0,1)$. 
    \begin{enumerate}
        \item Compute the statistic given in \eqref{self_sum},
        \begin{equation*}
        T_n=\dfrac{\sum_{t=1}^nX_t}{(\sum_{t=1}^nX_t^2)^{1/2}}.    
        \end{equation*}
        \item With $q_n := n-b_n+ 1$, compute the subsample statistics
        \begin{equation}\label{block_stat}
        T_{i,b_n} := \dfrac{\sum_{t=i}^{i+b_n-1} X_t}{(\sum_{t=i}^{i+b_n-1} X_t^2)^{1/2}},\qquad i=1,\dots,q_n.
        \end{equation}
        \item Define 
        \begin{equation}\label{eq:def_L_n_b}
        L_{n,b_n}(x):= q_n^{-1}\sum_{i=1}^{q_n}\mathbf{1}\big(T_{i,b_n}\leqslant x\big),
        \end{equation}
        and compute the $\eta/2$ and $(1-\eta/2)$ empirical quantiles of $(T_{i,b_n})_{i=1,\dots,q_n}$, i.e.,
        \begin{equation}\label{eq:algo_C_n_b_1}
            C_{n,b_n}(\eta/2)=\inf\{x:L_{n,b_n}(x)\geqslant \eta/2\}
        \end{equation}
        and
        \begin{equation}\label{eq:algo_C_n_b_2}
            C_{n,b_n}(1-\eta/2)=\inf\{x:L_{n,b_n}(x)\geqslant1-\eta/2\}.
        \end{equation}

    \end{enumerate}
    Then a two-sided equal tailed test of level $\eta$ rejects $\mathsf{H}_0$ if $T_n \notin [C_{n,b_n}(\eta/2), C_{n,b_n}(1-\eta/2)]$. 
\end{algorithm}
The algorithm is tailored to work irrespective of $X_t$ having finite variance (Assumption \ref{mix_coef_ass}) or being regularly varying with $\kappa\in(1,2)$ (Assumptions \ref{ass_RV}--\ref{ass:AC+MX}). In particular, the algorithm does not rely on estimation of the long-run variance $\sigma^2$ in Theorem \ref{theo:rio}, and it does not require any knowledge about (or estimation of) the scaling sequence $(a_n)_{n\geqslant1}$ in \eqref{eq:def_a_n}.

The asymptotic validity of the subsampling-based test in Algorithm \ref{algo} is ensured under the strong mixing conditions in Assumptions \ref{mix_coef_ass} and \ref{ass:AC+MX}. In particular, we have the following result.
\begin{theorem}\label{cor:validity_subsampling}
With $(b_n)_{n\geqslant 1}$ provided by Algorithm \ref{algo}, suppose that 
 $b_n\rightarrow\infty$ and $b_n=o(n)$ as $n\to\infty$.   
Given a level $\eta>0$, let $C_{n,b_n}(\eta/2)$ and $C_{n,b_n}(1-\eta/2)$ be given by \eqref{eq:algo_C_n_b_1} and \eqref{eq:algo_C_n_b_2}, respectively. 
With $(\xi_\kappa,\zeta_{\kappa/2})$ the limiting random variables in Theorem \ref{theo:inf_var_jointCLT}, define $C(y)=\inf\{x:\mathbb{P}(\xi_\kappa/\zeta_{\kappa/2}\leqslant x)\geqslant y)\}$ for $y\in(0,1)$.
Under $\mathsf{H}_0$, suppose either that
\begin{enumerate}
    \item   Assumption \ref{mix_coef_ass} holds, or
    \item   Assumptions  \ref{ass_RV}--\ref{ass:AC+MX} hold with $\kappa \in (1,2)$, and $\mathbb{P}(\xi_\kappa/\zeta_{\kappa/2}\leqslant x)$  is continuous in neighborhoods around $C(\eta/2)$ and  $C(1-\eta/2)$.
    
\end{enumerate}
Then  
        \begin{equation*}
            \mathbb{P}\big(T_n\notin [C_{n,b_n}(\eta/2), C_{n,b_n}(1-\eta/2)]\big)\rightarrow \eta,\qquad n\to\infty.
        \end{equation*}

\noindent Under either Assumption \ref{mix_coef_ass} or Assumptions  \ref{ass_RV}--\ref{ass:AC+MX} with $\kappa \in (1,2)$, if $\mathbb{E}[X_t]\neq0$ then
        \begin{equation*}
            \mathbb{P}\big(T_n \notin [C_{n,b_n}(\eta/2), C_{n,b_n}(1-\eta/2)]\big)\rightarrow 1,\qquad n\to\infty.
        \end{equation*}

\end{theorem}

\begin{remark}\label{rem:limit_T_n}
    Theorem \ref{cor:validity_subsampling} implies that the subsampling-based test in Algorithm \ref{algo} is asymptotically valid for a range of processes. In the infinite variance case ($\kappa\in(1,2)$), Assumptions \ref{ass_RV}--\ref{ass:AC+MX} cover a more general class of processes than those considered in existing work on subsampling-based inference about the mean of heavy-tailed time series. For instance, \cite{McElroy_Politis_2002_subsampling} focus entirely on heavy-tailed linear processes. \cite{kokoszka2004subsampling} establish \eqref{eq:inf_var_jointCLT}, but their result hinges on a technical assumption about point process convergence.\footnote{That is, Condition (15) in  \cite{kokoszka2004subsampling}. To show that this condition holds for a given process, one typically has to argue $(i)$ that $(X_t)_{t\in\mathbb{Z}}$ is reguarly varying (as in our Assumption \ref{ass_RV}), $(ii)$ that an anti-clustering condition similar to our Assumption \ref{ass:AC+MX} holds \citep[Condition (2.8) in Theorem 2.7]{davis1995point}, and $(iii)$ that the limiting point process is non-null.} In addition, they assume that $\mathbb{E}[X_t\boldsymbol{1}(|X_t|\leqslant x)X_s\boldsymbol{1}(|X_s|\leqslant x)]=0$ for all $x>0$, $t\neq s$.\footnote{Their condition (14). The assumption is made in order to prove a so-called small-vanishing-values condition for the case $\kappa\in(1,2)$, c.f. condition (3.2) of Theorem 3.1 in \cite{davis1995point}. This condition is typically hard to verify for a given process, see \citet[Section 2.4]{bartkiewicz2011stable} for a discussion.} This condition rules out a general class of processes, including non-trivial linear processes, and appears unnecessarily restrictive.
\end{remark}

\begin{remark}\label{rem:singularities}
    For $\kappa\in(1,2)$ and $\mathbb{E}[X_t]=0$, Theorem \ref{cor:validity_subsampling} imposes that $\mathbb{P}(\xi_\kappa/\zeta_{\kappa/2}\leqslant x)$ is continuous at $C(\eta/2)$ and  $C(1-\eta/2)$. This kind of assumption is standard in the  literature on subsampling for heavy-tailed time series. Whereas $\xi_\kappa$ is $\kappa$-stable and hence has a continuous distribution, the distribution of the ratio $\xi_\kappa/\zeta_{\kappa/2}$ may have singularities; see, e.g., \cite{Logan1973self-norm}. In the context of  heavy-tailed linear AR(1) processes, as considered in Section \ref{sec:AR_details}, it holds by Proposition \ref{prop:AR_lim_moments} that $\xi_\kappa/\zeta_{\kappa/2}$ does not have any singularity points whenever $\kappa\in(1,2)$.
\end{remark}
Theorem \ref{cor:validity_subsampling} states that the subsampling-based test in Algorithm \ref{algo} is asymptotically valid for \textit{any} block size $b_n=o(n)$ with $b_n\to\infty$. As pointed out by, e.g., \cite{kokoszka2004subsampling}, the finite sample properties of a subsampling procedure depend on $b_n$ chosen by the user. Several data-driven approaches have been considered.  If one is willing to assume a specific process (e.g., AR(1) or GARCH) for $X_t$, one may determine $b_n$ by means of simulations as in \cite[Section 9.4.1]{politis2012subsampling}. The so-called Minimum Volatility Method \citep[Section 9.4.2]{politis2012subsampling} avoids this assumption, but requires a user-chosen grid of candidate values. \cite{kim2021forecast} propose a  
rule depending on estimated tail-index and skewness parameters in the 
iid stable-distribution setting. We adopt the simple rule
\begin{equation}\label{eq:b_n_choice}
    b_n = \lfloor1.5n^{0.5}\rfloor,
\end{equation}
 similar to block sizes used in \citet[Section 5]{bai2016unified} and \citet[Section 3]{zhang2013blocksampling}, for simplicity and to 
maintain consistency with existing work on subsampling for heavy-tailed 
dependent data. The Monte Carlo experiments in Section \ref{sec:Monte_Carlo} indicate that this rule yields reasonable 
finite-sample rejection frequencies across a range of tail indices 
and sample sizes.

\begin{remark}
    Algorithm \ref{algo} provides critical values for an equal-tailed two-sided test. Alternatively, one could consider a symmetric two-sided test where the critical value is given by the $(1-\eta)$ empirical quantile of $|T_{i,b_n}|$, and the test rejects if $|T_n|$ exceeds this critical value. Likewise, one could alternatively consider one-sided alternatives. The asymptotic validity of both such tests follows directly from the results provided in the Supplemental Appendix. 
\end{remark}

\begin{remark}\label{rem:confidence_interval}
As mentioned, most of the existing literature on subsampling-based inference about the mean of a regularly varying time series considers confidence intervals for the mean. In particular, with $\overline{X}_n = n^{-1}\sum_{t=1}^nX_t$ and
    \begin{equation*}
    T_{i,b_n}^c= \dfrac{\sum_{t=i}^{i+b_n-1} (X_t-\overline{X}_n)}{\big(\sum_{t=i}^{i+b_n-1}(X_t-\overline{X}_n)^2\big)^{1/2}},\qquad i=1,\dots,q_n,
\end{equation*}
define 
\begin{equation*}
        L_{n,b_n}^c(x)= q_n^{-1}\sum_{i=1}^{q_n}\mathbf{1}\big(T_{i,b_n}^c\leqslant x\big),\qquad x\in\mathbb{R},
        \end{equation*}
        and
        \begin{equation*}
            C_{n,b_n}^c(y)=\inf\{x:L_{n,b_n}^c(x)\leqslant y\},\qquad y\in(0,1).
        \end{equation*}
   Then an equal-tailed two-sided confidence interval for $\mathbb{E}[X_t]$ is given by
   \begin{equation}\label{eq:subsamling_CI}
       \mathrm{CI}_{b_n,n}=[\overline{X}_n+(\gamma_n/n)C_{n,b_n}^c(\eta/2) \;, \; \overline{X}_n+(\gamma_n/n)C_{n,b_n}^c(1-\eta/2)].
   \end{equation}
   Under the same assumptions as in Theorem \ref{cor:validity_subsampling}, but with $\mathbb{E}[X_t]$ potentially non-zero,
   \begin{equation*}
       \mathbb{P}\left(\mathbb{E}[X_t]\in \mathrm{CI}_{b_n,n}\right)\to 1-\eta ,\qquad n\to\infty,
   \end{equation*}
   that is, the confidence interval has asymptotic correct coverage. %In the context of forecast evaluations, suppose, for instance, that $X_t =(Y_t-f_{t})^2$. Then $\mathrm{CI}_{b_n,n}$ yields a confidence interval for the mean squared forecast error. 
   Note that an alternative to the test in Algorithm \ref{algo} is achieved by rejecting $\mathsf{H}_0$ if $0\notin  \mathrm{CI}_{b_n,n}$ \citet[p.54]{politis2012subsampling}. The validity of such a test follows directly from the technical proofs in the Supplemental Appendix.
\end{remark}

\section{Extensions}\label{sec:extensions}
In this section we consider extensions of the limit theory and related subsampling-based inference considered in Sections \ref{sec:limit_theory}--\ref{sec:subsampling}. In Section \ref{sec:infinite_mean} we consider testing the hypothesis $\mathsf{H}_0$ allowing for alternatives where $\mathbb{E}[|X_t|]=\infty$, that is, the mean may be infinite or may not exist. In Section \ref{sec:multivariate} we consider testing a hypothesis about a mean of a vector with particular attention to tests for superior predictive ability (SPA).

\subsection{Allowing for alternatives with $\mathbb{E}[|X_t|]=\infty$}\label{sec:infinite_mean}
 A straightforward implication of Theorem \ref{theo:inf_var_jointCLT} is that when the tail index\footnote{The case $\kappa\in(0,1)$ may indeed be empirically relevant: \cite{kim2021forecast} find empirical evidence for such values of the tail index when considering loss differentials of realized volatility forecasts under squared loss, but do not address the inferential consequences. Their finding may also be motivated theoretically: \cite{gabaix2006_QJE} argue that stock market returns have tail index of $3$ (the ``cubic law of returns''). Consequently, realized volatility -- constructed by summing squared returns -- should have tail index $3/2$, so that squared realized volatility forecast errors may have tail index $\kappa=3/4<1$.} $\kappa\in(0,1)$,
   \begin{equation*}
            T_n\xrightarrow{d}\xi_\kappa\;/\zeta_{\kappa/2},\qquad n\to\infty.
   \end{equation*}
   Consequently, under alternatives where $\mathbb{E}[|X_t|]=\infty$, the test statistic $T_n$ does not diverge. Noting that the limiting distribution of the subsample statistics are identical to the one of $T_n$, with the critical values $C_{n,b_n}(\eta/2)$ and $C_{n,b_n}(1-\eta/2)$ determined according to Algorithm \ref{algo}, we have that
   \begin{equation}
       \mathbb{P}(T_n\notin[C_{n,b_n}(\eta/2),C_{n,b_n}(1-\eta/2)])\to \eta,\qquad n\to\infty;
   \end{equation}
   see Proposition \ref{subsample_special} in the Supplemental Appendix.
   That is, the two-sided equal-tailed test is inconsistent when $\kappa\in(0,1)$. To the best of our knowledge, this issue has not been considered in the existing body of literature on subsampling techniques for heavy-tailed observations.

To construct a test that delivers non-trivial rejection frequencies under alternatives where $\kappa\in(0,1)$, we consider a modification of the original test statistic $T_n$. Define 
\begin{equation}\label{eq:def_gamma_1}
    \overline{\gamma}_{n}=\sum_{t=1}^n|X_t|.
\end{equation}
Under Assumption \ref{mix_coef_ass} or Assumptions \ref{ass_RV}--\ref{ass:AC+MX}, if $\mathbb{E}[|X_t|]<\infty$, then by the ergodic theorem $\overline{\gamma}_{n}/n\overset{\mathbb{P}}{\to}\mathbb{E}[|X_t|]$, as $n\to \infty$. In Section \ref{sec:proof_infinite_mean} we show that, with $a_n$ given by \eqref{eq:def_a_n}, $\overline{\gamma}_n/a_n $ has a strictly positive $\kappa$-stable limit if $\kappa\in(0,1)$, while $a_n/n\to\infty$ as $n\to\infty$. This implies that $\overline{\gamma}_{n}/n$ diverges for $\kappa\in(0,1)$ and motivates the modified statistic
\begin{equation}\label{eq:modified_statistic}
    \widetilde{T}_n=(\overline{\gamma}_{n}/n)T_n.
\end{equation}
We have the following result.
\begin{proposition}
   
Let the statistic $\widetilde{T}_n$ be given by \eqref{eq:modified_statistic}.
\begin{enumerate}
    \item  Under Assumption \ref{mix_coef_ass} and $\mathbb{E}[X_0]=0$,
    \begin{equation*}
        \widetilde{T}_n\overset{d}{\to} \mathbb{E}[|X_t|]Z,\qquad n\to\infty,
    \end{equation*}
    with $Z$ a Gaussian random variable given in Theorem \ref{theo:rio}.
    \item Under Assumptions \ref{ass_RV}--\ref{ass:AC+MX} with $\kappa\in(1,2)$ and $\mathbb{E}[X_0]=0$,
    \begin{equation*}
        \widetilde{T}_n\overset{d}{\to} \mathbb{E}[|X_t|] \frac{\xi_\kappa}{\zeta_{\kappa/2}},\qquad n\to\infty,
    \end{equation*}
    with $(\xi_\kappa,\zeta_{\kappa/2})$ given in Theorem \ref{theo:inf_var_jointCLT}.
    \item Under Assumption \ref{mix_coef_ass} or Assumptions \ref{ass_RV}--\ref{ass:AC+MX}, if $\mathbb{E}[X_t]\neq0$,
    \begin{equation}
        |\widetilde{T}_n|\overset{\mathbb{P}}{\to}\infty,\qquad n\to\infty.
    \end{equation}
\end{enumerate}
\end{proposition}
Critical values for a test based on $\widetilde{T}_n$ can be computed by the subsampling Algorithm \ref{algo_infinite_mean} in the Supplemental Appendix. Importantly, under $\kappa\in(0,1)$ we show in Section \ref{sec:proof_infinite_mean} that the subsample statistics $\widetilde{T}_{i,{b_n}}$ satisfy 
\begin{equation}\label{eq:modified_subsampling_validity}
    |\widetilde{T}_{i,{b_n}}| = o_\mathbb{P}(|\widetilde{T}_n|),
\end{equation}
under $\mathbb{E}[X_t]\neq0$, that is, the original statistic $\widetilde{T}_n$ diverges faster than the subsample statistics, ensuring non-trivial rejection frequencies. The finite-sample properties of a test based on the modified statistic $\widetilde{T}_n$ and subsampling-based critical values, are investigated in a simulation experiment in Section \ref{sim:infinite_mean} in the Supplemental Appendix. Notably, the simulations indicate that for processes with undefined mean ($\kappa <1$), the standard Diebold--Mariano test has rejection frequencies below the nominal level irrespective of the sample size, whereas the subsampling-based test in Algorithm \ref{algo_infinite_mean} achieves non-trivial rejection frequencies for sufficiently large sample sizes. 

\subsection{Testing for superior predictive ability} \label{sec:multivariate}
  In Section \ref{sec:appendix_multivariate} of the Supplemental Appendix we provide a multivariate version of Theorem \ref{theo:inf_var_jointCLT}. The result, Theorem \ref{theo:mikosch_multivariate}, may be used for testing a hypothesis about the mean vector of a multivariate heavy-tailed stationary time series. We illustrate this with an application to comparing multiple forecast series. Motivated by the work of \cite{hansen2005SPA}, see also \cite{white2000snooping}, suppose that we have $m+1\geqslant 2$ forecast series $(f_{j,t})_{t=1,\dots,n}$, $j=0,1\dots,m$, of a target variable $(Y_t)_{t=1,\dots,n}$ with $m$ fixed. Given a loss function $\mathcal{L}(\cdot , \cdot)$, let $\mathbf{X}_t=(X_{t,1},\dots,X_{t,m})'$ be defined by
\begin{equation}
  X_{t,j}:=\mathcal{L}(Y_t,f_{0,t})-\mathcal{L}(Y_t,f_{j,t}),\qquad j=1,\dots,m,
\end{equation}
that is $X_{t,j}$ is the loss differential between forecast series $0$ and $j$ at time $t$. Given that $\mathbf{X}_t$ has time-invariant mean, the forecast series $0$ is said to have \textit{superior predictive ability} (SPA) over the other $m$ competing forecast series if 
\begin{equation}\label{eq:def_H_SPA}
    \mathsf{H}_{\mathrm{SPA}}:\mathbb{E}[\mathbf{X}_t]\leqslant \mathbf{0}
\end{equation}
is true, where the inequality holds entry-wise. \cite{hansen2005SPA} proposes a test for $\mathsf{H}_{\mathrm{SPA}}$ against
\begin{equation}\label{eq:H_alternative_SPA}
 \mathsf{H}_{\mathrm{non-SPA}}: \mathbb{E}[X_{t,j}]>0 \ \text{for some} \ j\in\{1,\dots,m\}   
\end{equation}
under the assumptions of $\Vert\mathbf{X}_t\Vert$ having bounded variance and suitable strong mixing coefficients of $(\mathbf{X}_t)_{t\in\mathbb{Z}}$, see Assumption \ref{ass:mixing_multivariate} in Section \ref{sec:appendix_multivariate} of the Supplemental Appendix. Here $\Vert \mathbf{x} \Vert = (\sum_{j=1}^mx_j^2)^{1/2}$ for a vector $\mathbf{x}=(x_1,\dots,x_m)'\in\mathbb{R}^m$. In order to account for the possibility of $\Vert\mathbf{X}_t\Vert$ having infinite variance, that is, being regularly varying with index $\kappa\in(1,2)$ (c.f., Assumption \ref{ass_RV_multivariate} in the  Supplemental Appendix), we consider a modification of the test statistic proposed by \cite{hansen2005SPA}. Define
\begin{eqnarray}\label{eq:T_n_multivariate}
    \mathbf{S}_n := \sum_{t=1}^n\mathbf{X}_t,\quad \tilde{\gamma}_{n}:=\sqrt{\sum_{t=1}^n\Vert \mathbf{X}_t\Vert^2},\quad  \mathbf{T}_n := (T_{n,1},\dots,T_{n,m})'= \tilde{\gamma}_{n}^{-1}\mathbf{S}_n,
\end{eqnarray}
and the test statistic
\begin{equation}\label{eq:def_V_SPA}
    V_n^{\mathrm{SPA}}:=\max\left[\max_{j=1,\dots,m}T_{n,j},0\right].
\end{equation}
We have the following result.
\begin{theorem}\label{thm:SPA}
      Suppose that $(\mathbf{X}_t)_{t\in\mathbb{Z}}$ is stationary. Let $ V_n^{\mathrm{SPA}}$ be given by \eqref{eq:def_V_SPA}.
    \begin{itemize}
        \item [1.] Suppose that $\mathsf{H}_{\mathrm{SPA}}$ in \eqref{eq:def_H_SPA} holds.
        \begin{itemize}
            \item [a.] Under Assumption \ref{ass:mixing_multivariate},
            \begin{equation*}
                V_n^{\mathrm{SPA}} \overset{d}{\to} \frac{\max\left[\max_{j=1,\dots,m}\tilde{Z}_j,0\right]}{\sqrt{\mathbb{E}[\Vert \mathbf{X}_0 \Vert^2]}},\qquad n\to\infty.
            \end{equation*}
            where 
            \begin{equation}
                \tilde{Z_{j}}=Z_j\mathbf{1}_{(\mathbb{E}[X_{0,j}]=0)},\qquad j=1,\dots,m,
            \end{equation}
            with $\mathbf{Z} = (Z_1,\dots,Z_m)'$ a Gaussian vector provided by Theorem \ref{thm:rio_multivariate} in the  Appendix.
            \item [b.] Under Assumptions \ref{ass_RV_multivariate}--\ref{AC,MX_multivariate} with tail index $\kappa\in(1,2)$,
            \begin{equation*}
                V_n^{\mathrm{SPA}} \overset{d}{\to} \frac{\max\left[\max_{j=1,\dots,m}\tilde{\xi}_{\kappa,j},0\right]}{\tilde{\zeta}_{\kappa,2}},\qquad n\to\infty.
            \end{equation*}
            where 
            \begin{equation}
              \tilde{\xi}_{\kappa,j}=\xi_{\kappa,j}\mathbf{1}_{(\mathbb{E}[X_{0,j}]=0)}, \qquad j=1,\dots,m,
            \end{equation}
            with $\bm{\xi}_\kappa = (\xi_{\kappa,1},\dots,\xi_{\kappa,m})'$ a $\kappa$-stable random vector and $\tilde{\zeta}_{\kappa,2}^2$ a strictly positive $\kappa/2$-stable random variable both provided by Theorem \ref{theo:mikosch_multivariate} in the  Appendix.
        \end{itemize}
        \item [2.] Suppose that $\mathsf{H}_{\mathrm{non-SPA}}$ in \eqref{eq:H_alternative_SPA} holds. Then under either Assumption \ref{ass:mixing_multivariate} or Assumptions \ref{ass_RV_multivariate}--\ref{AC,MX_multivariate} with tail index $\kappa\in(1,2)$,
            \begin{equation*}
                V_n^{\mathrm{SPA}}\overset{\mathbb{P}}{\to}\infty,\qquad n\to\infty.
            \end{equation*}        
        
    \end{itemize}
\end{theorem}
As earlier, critical values for a test based on $V_n^{\mathrm{SPA}}$ can be computed by a straightforward multivariate extension of the subsampling procedure in Algorithm \ref{algo}.

\begin{remark}
    Under Assumption \ref{ass_RV_multivariate} the entries of $\mathbf{X}_t$ may have different tail heaviness. In this case, the entries of $\mathbf{S}_n$ are of different stochastic orders whereas the normalizing quantity $\tilde{\gamma}_n$ has stochastic order determined by the \textit{smallest} of the tail indices, namely the index of regular variation of $\Vert \mathbf{X}_t \Vert$, $\kappa>0$. Consequently, the limiting distribution of $\mathbf{T}_n$ in \eqref{eq:T_n_multivariate} may be singular with respect to the Lebesgue measure on $\mathbb{R}^m$ when some entries of  $\mathbf{X}_t$ have zero mean. This may affect the power properties of the test based on the statistic $V_n^{\mathrm{SPA}}$. An alternative test statistic, similar to the one considered by \cite{hansen2005SPA}, could be based on $\mathbf{S}_n \oslash \bm{\gamma}_n$ with $\bm{\gamma}_n:=(\gamma_{n,1},\dots,\gamma_{n,m})'$, $\gamma_{n,j}=\sqrt{\sum_{t=1}^nX_{t,j}^2}$, and $\oslash$ denoting entry-wise division. It remains an open task to derive stable limit results for $(\mathbf{S}_n,\bm{\gamma}_n)$ under potentially heterogeneous tail heaviness across the entries of $\mathbf{X}_t$ and, consequently, with entry-wise different rates of convergence.
\end{remark}

\section{Monte-Carlo simulations}\label{sec:Monte_Carlo}
In this section, we investigate the finite-sample properties of the standard Diebold--Mariano test as well as the test based on subsampling Algorithm \ref{algo} with block length $b_n$ given by \eqref{eq:b_n_choice}. We refer to Section \ref{sim:infinite_mean} in the Supplemental Appendix for additional results for the Diebold--Mariano test and the subsampling based test allowing for $\mathbb{E}[|X_t|]=\infty$ as considered in Section \ref{sec:infinite_mean}. 

Throughout, we consider a DGP for $X_t$ given in terms of the AR(1) process
\begin{equation}\label{eq:sim_AR}
    X_t=\delta+0.5\; X_{t-1}+Z_t,\qquad t=1,\dots,n,
\end{equation}
with iid noise $Z_t$.\footnote{The initial $X_0$ is chosen random using a burn-in period of $10^4$ observations initiated at $X_{-10000}=0$.}\footnote{We also considered different values of the autoregressive coefficient ranging from $0.1$ to $0.9$. The simulation results were qualitatively the same as for the ones reported here with $0.5$.} Initially, we consider two different distributions for $Z_t$: a symmetric stable and an asymmetric stable, denoted \textrm{Stable}$(\kappa,0,1,0)$  and \textrm{Stable}$\left(\kappa,4/5,1,4\tan(\kappa\pi/2)/5\right)$, respectively, with $\kappa\in(1,2)$. Both distributions imply that $\mathbb{E}[Z_t]=0$, and, consequently, $\mathbb{E}[X_t]=2\delta$ so that the hypothesis $\mathsf{H}_0$ in \eqref{eq:zero_mean} is true if and only if $\delta=0$. Moreover,  $Z_t$ is regularly varying with index $\kappa$ for both choices of distribution. The symmetric distribution has tail balance coefficients $p_+=p_-=1/2$, whereas the asymmetric distribution has $p_+=1-p_-=9/10$. By Proposition \ref{prop:AR1 - RV} $(X_t)_{t\in\mathbb{Z}}$ is regularly varying with index $\kappa$. %Moreover, by Proposition \ref{prop:AR_lim_moments}, the limiting distribution $\xi_\kappa/\zeta_{\kappa/2}$ is asymmetric in the latter case. Furthermore, $\xi_\kappa/\zeta_{\kappa/2}$ has a continuous probability density function, and moments of any order are finite for both distributional specifications of $Z_t$.  Notably, for the case $p_+=0.9$, it holds by \eqref{eq:limit_mean} that $\xi_\kappa/\zeta_{\kappa/2}$ has a strictly positive mean.
Moreover, by Proposition \ref{prop:AR_lim_moments}, the limiting distribution of the statistic $T_n$, $\xi_\kappa/\zeta_{\kappa/2}$, has a continuous density and finite moments of all orders for both distributional specifications of $Z_t$. In the asymmetric case ($p_+ = 0.9$), $\xi_\kappa$ has strictly positive skewness parameter $\beta$, and $\xi_\kappa/\zeta_{\kappa/2}$ has strictly negative mean.

In terms of the Diebold--Mariano test, we make use of the statistic
\begin{align}\label{eq:T_DM} 
    T_n^{\mathrm{DM}}=\frac{S_n}{\sqrt{n}\hat{\sigma}_{n,\mathrm{NW}}},
\end{align}
with $S_n$ given by \eqref{eq:def_S_gamma} and with $\hat{\sigma}^2_{n,\mathrm{NW}}$ the standard \cite{newey_west_1987} HAC variance estimator with lag length given by $\lfloor4(n/100)^{2/9}\rfloor$, as considered in \citet[Section IV]{newey_west_1994_automatic}. The test rejects the hypothesis $\mathsf{H}_0$ at a 5\% nominal level if $|T_n^{\mathrm{DM}}|>1.96$. Apart from the centering with respect to the sample mean in $\hat{\sigma}^2_{n,\mathrm{NW}}$ and the increasing lag length, the test statistic is identical to $T_n^{\mathrm{HAC}}$ in \eqref{eq:T_HAC}. Consequently, from Proposition \ref{prop:AR_HAC} we expect that, under $\mathsf{H}_0$, $T_n^{\mathrm{DM}}$ is potentially poorly approximated by a standard normal distribution when $X_t$ has tail index $\kappa <2$.

\subsection{Rejection frequencies under $\mathsf{H}_0$}
Table \ref{tab:RFs under null} contains rejection frequencies under $\mathsf{H}_0$ for the two tests for different values of the tail index $\kappa\in(1,2)$. In the case where the innovation $Z_t$ has a symmetric distribution, Panel A of Table \ref{tab:RFs under null} states that the rejection frequencies of the two tests are quite comparable and fairly close to the nominal level of 5\% across tail indices and sample sizes.

Turning to Panel B, where the innovations exhibit extremal skewness, the Diebold--Mariano test is heavily distorted, with rejection frequencies remarkably above the nominal level for $\kappa \leqslant 1.5$. This holds irrespective of the sample size, and the distortion increases with the tail-heaviness ($\kappa \downarrow 1$). This is likely explained by the asymmetry of the limiting distribution $\xi_\kappa/\zeta_{\kappa/2}$, which suggests that critical values based on the (symmetric) standard normal distribution are unreliable. The subsampling-based test does not perform as well as in the symmetric case, but its rejection frequencies are generally substantially closer to the nominal level than those of the Diebold--Mariano test, and converge toward it as the sample size increases. For tail indices $\kappa \geqslant 1.5$, the rejection frequencies of the subsampling-based test appear reasonable across all sample sizes.

\begin{table}[H]
\centering
\caption{Rejection frequencies (in percent) under $\mathsf{H}_0$ for the Diebold--Mariano test and the subsampling-based test in Algorithm \ref{algo}. The Diebold--Mariano test is based on the statistic $T_n^{\mathrm{DM}}$ in \eqref{eq:T_DM}. The DGP $(X_t)_{t=1,\dots,n}$ is given by  \eqref{eq:sim_AR} with $\delta=0$. In Panel A the iid noise $(Z_t)_{t=1,\dots,n}$ is $\textrm{Stable}(\kappa,0,1,0)$, and in Panel B it is \textrm{Stable}$\left(\kappa,4/5,1,4\tan(\kappa\pi/2)/5\right)$. All tests are carried out at a 5\% nominal level. The rejection frequencies are based on $M=10^4$ Monte-Carlo replications.}\label{tab:RFs under null}
\resizebox{0.95\textwidth}{!}{
\begin{tabular}{l|ccccc|ccccc}
\hline\hline
 & \multicolumn{5}{c|}{Diebold--Mariano} & \multicolumn{5}{c}{Subsampling Algorithm \ref{algo}} \\
 \hline
 $n$: & 1000 & 2000 & 5000 & 10000 & 100000 & 1000 & 2000 & 5000 & 10000 & 100000 \\
 \hline\hline
 \multicolumn{11}{c}{Panel A - \textrm{Stable}$(\kappa,0,1,0)$ noise.} \\
 \hline
 $\kappa=1.1$  & 4.9 & 4.4 & 3.8 & 3.5 & 3.0 & 7.3 & 6.5 & 5.7 & 5.6 & 5.2 \\
 $\kappa=1.3$  & 6.1 & 5.5 & 5.0 & 4.2 & 4.1 & 6.7 & 6.2 & 6.0 & 5.2 & 5.3 \\
 $\kappa=1.5$  & 6.8 & 5.3 & 5.6 & 5.0 & 4.5 & 5.7 & 4.6 & 4.9 & 4.7 & 4.8 \\
 $\kappa=1.7$  & 7.4 & 6.8 & 6.3 & 6.0 & 4.8 & 5.0 & 4.5 & 4.4 & 4.6 & 4.2 \\
 $\kappa=1.9$  & 7.6 & 7.2 & 6.7 & 6.3 & 5.5 & 4.5 & 4.0 & 3.8 & 4.0 & 3.7 \\
 \hline\hline
 \multicolumn{11}{c}{Panel B - \textrm{Stable}$\left(\kappa,4/5,1,4\tan(\kappa\pi/2)/5\right)$ noise.} \\
 \hline
 $\kappa=1.1$ & 71.2 & 71.8 & 71.4 & 71.0 & 70.6 & 50.4 & 45.1 & 37.8 & 31.2 & 15.1 \\
 $\kappa=1.3$ & 35.2 & 36.1 & 34.9 & 34.3 & 33.0 & 15.3 & 12.7 & 9.4 & 8.1 & 6.0 \\
 $\kappa=1.5$ & 18.0 & 18.0 & 17.2 & 16.2 & 16.3 & 8.3 & 7.1 & 5.8 & 5.6 & 4.9 \\
 $\kappa=1.7$ & 11.1 & 10.1 & 9.4 & 9.3 & 8.1 & 6.1 & 5.0 & 4.7 & 4.8 & 4.3 \\
 $\kappa=1.9$ & 8.0 & 6.8 & 7.2 & 6.2 & 6.0 & 4.6 & 3.5 & 3.8 & 3.6 & 4.0 \\
\hline\hline
\end{tabular}
}
\end{table}

\subsection{Rejection frequencies under alternatives}\label{sec:algo_power}
Figure \ref{fig:power_finite_mean} contains rejection frequencies for the subsampling-based test when the null hypothesis is violated. We consider DGPs given by \eqref{eq:sim_AR} with $Z_t$ \textrm{Stable}$(\kappa,0,1,0)$-distributed and $\delta\in(0,2.5]$ such that $\mathbb{E}[X_t]\in(0,5]$. As one would expect, the rejection frequencies are increasing in the sample size. We also note that the rejection frequencies tend to increase quite slowly in the sample size and in $\mathbb{E}[X_t]$ as $\kappa$ tends to one. This is likely explained by the fact that for $\kappa$ close to one, we are near a setting where $X_t$ has undefined mean, as considered in detail in Section \ref{sec:infinite_mean}. When $X_t$ has undefined mean ($\kappa\in(0,1)$) the test statistic $T_n$ does not diverge irrespective of the value of $\delta$, and the subsampling- based test is not expected to provide non-trivial rejection frequencies.
\begin{figure}[H]
    \centering
    \includegraphics[width=1\linewidth]{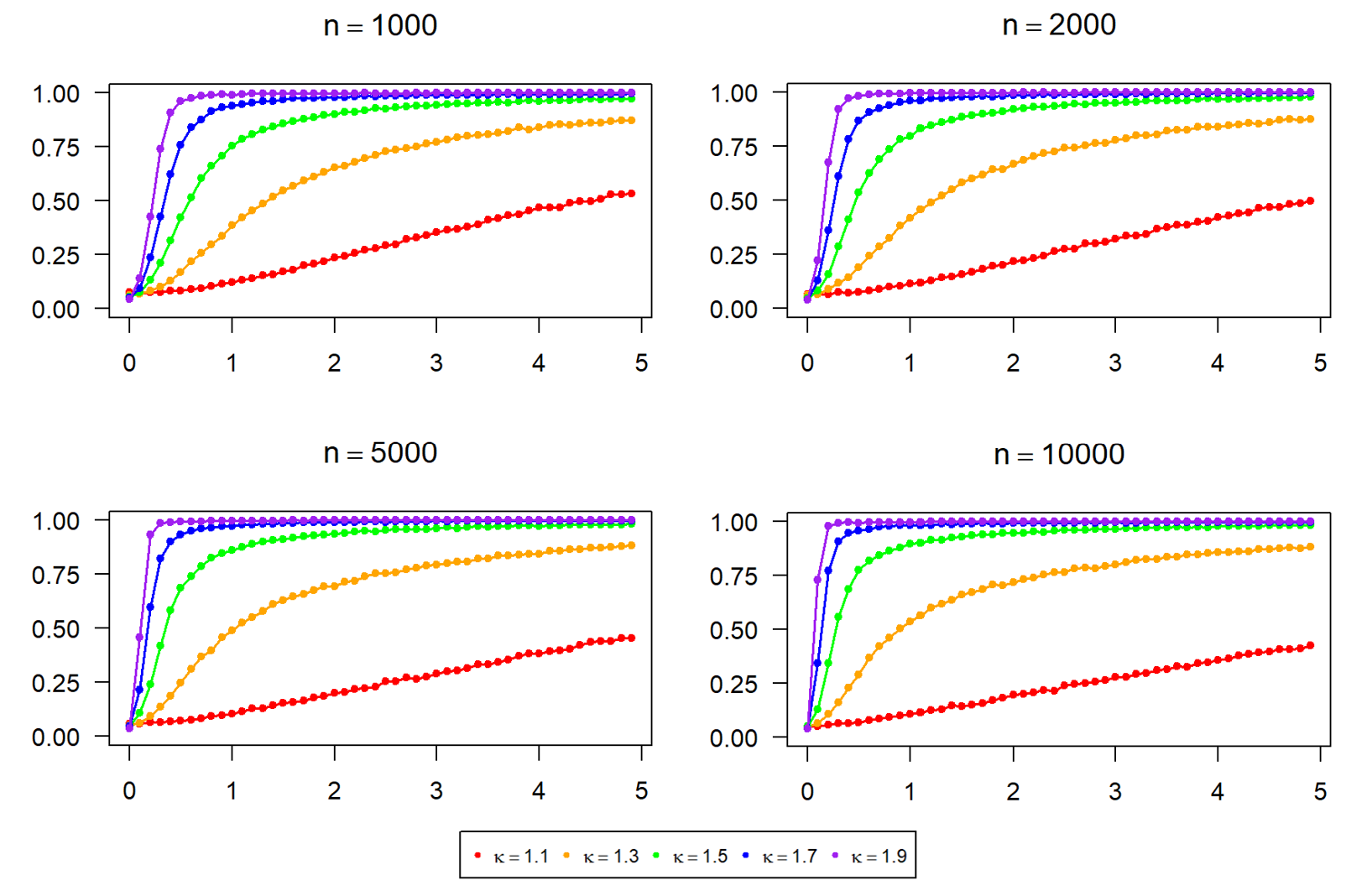}
    \caption{Rejection frequencies under alternatives $\mathbb{E}[X_t]\neq 0$ for the subsampling-based test in Algorithm \ref{algo}. The DGP $(X_t)_{t=1,\dots,n}$ is given by  \eqref{eq:sim_AR} with $\delta\in[0,2.5]$ and with iid noise $(Z_t)_{t=1,\dots,n}$ is $\textrm{Stable}(\kappa,0,1,0)$-distributed, $\kappa\in(1,2)$. The test is carried out at a 5\% nominal level. The x-axes indicate the values of $\mathbb{E}[X_t]=2\delta$. The rejection frequencies are based on $M=10^4$ Monte-Carlo replications.}\label{fig:power_finite_mean}
\end{figure}
We end this section by comparing the performances of the Diebold--Mariano and the subsampling-based tests when $X_t$ has finite variance, i.e., when the Diebold--Mariano test is expected to work. To do so, we modify the DGP so that $Z_t$ has a Student's $t$-distribution with $\kappa > 2$ degrees of freedom: $Z_t$ either has a standard (symmetric) Student's $t$-distribution ($\mathrm{Student}(\kappa)$) or a skewed Student's $t$-distribution ($\mathrm{Skt}(\kappa)$).\footnote{We refer to \cite{fernandez1998bayesian} for details about the skewed Student's $t$-distribution.} For the skewed case, $Z_t$ is scaled such that $V(Z_t) = \kappa/(\kappa-2)$, which is also the variance of the symmetric version. Both distributions are regularly varying with index $\kappa$ and have $\mathbb{E}[Z_t] = 0$. For the symmetric case, the tail balancing coefficients are $p_+ = p_- = 0.5$, and the skewed distribution is chosen such that $p_+ = 1 - p_- = 0.9$. Since $\kappa > 2$, the DGPs obey Theorem \ref{theo:rio}, so that $T_n$ has a Gaussian limiting distribution; we expect the same to hold for the Diebold--Mariano statistic $T_n^{\mathrm{DM}}$.

Figure \ref{fig:Algo vs Newey} contains rejection frequencies for the Diebold--Mariano and subsampling-based tests when the null hypothesis is violated. The first row corresponds to symmetric $t$-distribution and the second to the asymmetric one. For the heaviest-tailed symmetric distributions ($\kappa = 2.1, 2.5$), the Diebold--Mariano test yields larger rejection frequencies than the subsampling-based test, whereas the opposite holds under asymmetric distributions. Under lighter tails ($\kappa \geqslant 3$), the rejection frequencies of the two tests are quite similar, especially for large sample sizes ($n = 5000$). Overall, the two tests exhibit comparable power under finite-variance distributions.

\begin{figure}[H]
    \centering
    \includegraphics[width=1\linewidth]{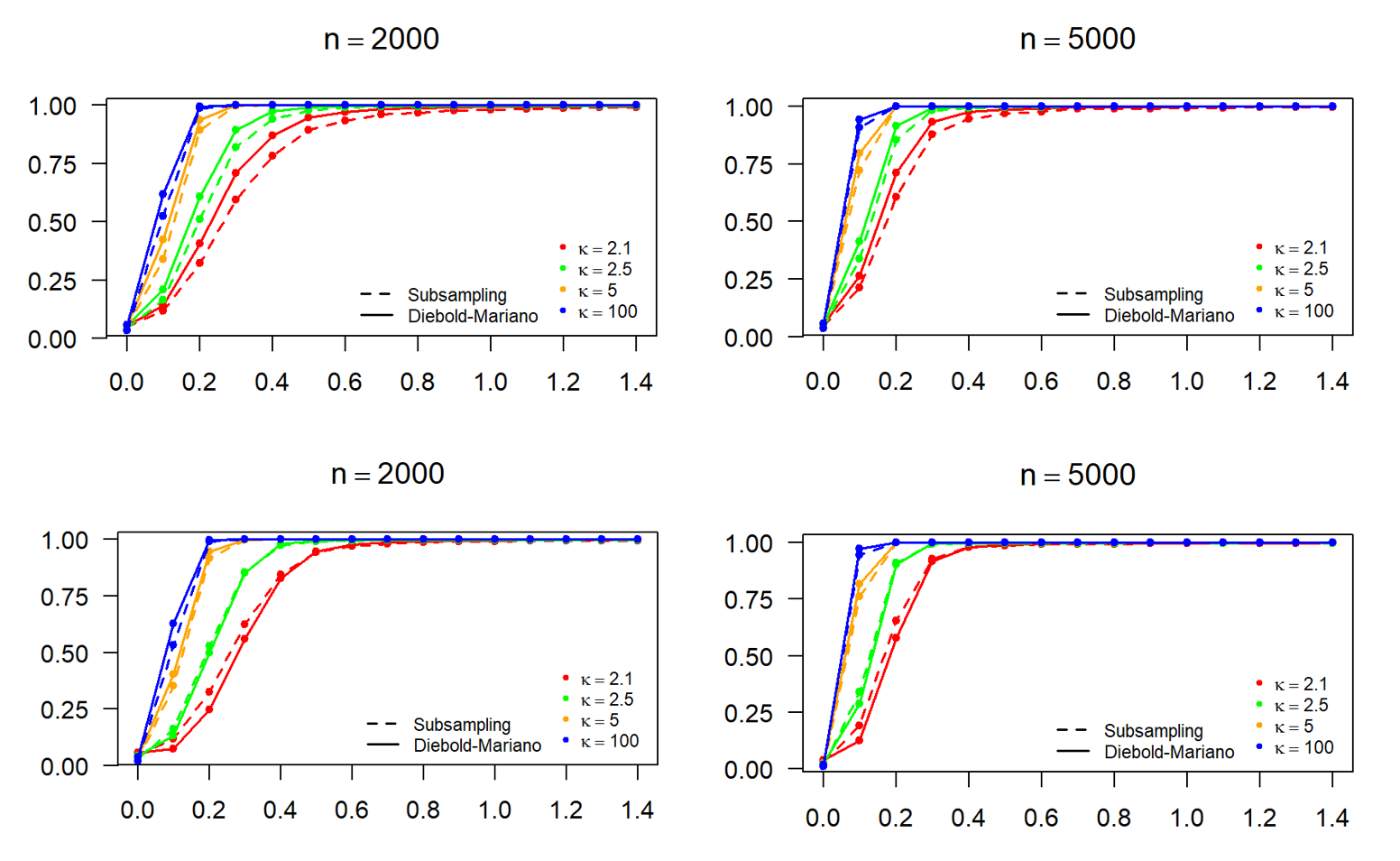}
    \caption{Rejection frequencies under alternatives for tests based on Algorithm \ref{algo} and the Diebold--Mariano test. The tests are carried out at a 5\% nominal level. The DGP is given by \eqref{eq:sim_AR} with $(Z_t)$ iid and $\delta \in [0,0.7]$. In the first row $Z_t$ is $\textrm{Student}(\kappa)$-distributed. In the second row $Z_t$ is \textrm{Skt}$(\kappa)$-distributed. The x-axes indicate the values of $\mathbb{E}[X_t]=2\delta$. The rejection frequencies are based on $M=10^4$ Monte-Carlo replications.}
    \label{fig:Algo vs Newey}
\end{figure}

\section{Empirical illustration: Comparisons of risk forecasts for emerging-market foreign exchange rates }\label{sec:empirical}

Quantification and hedging of foreign currency risk are important for international portfolio management\footnote{See, e.g., \cite{christensen_varneskov_2021_FX} and the references therein.}, and we consider in this section an application of Algorithm \ref{algo} to comparing risk forecasts for foreign exchange (FX) rate returns.  FX rates are typically volatile, and their returns are heavy-tailed. For instance, \cite{ibragimov2013emerging_markets} find that certain emerging-market FX rate returns are regularly varying, and standard confidence intervals for their tail index include values below two; that is, returns may have infinite variance. Among such FX rates is the South Korean won against U.S. dollars (KRWUSD), which we focus on here. 
Let $Y_t$ denote the FX rate log return on day $t$. Given some information set $\mathcal{F}_{t-1}$ available at time $t-1$ denote by $Q_{t-1,\tau}$ the $\tau$-quantile of the conditional distribution of $Y_t$ given $\mathcal{F}_{t-1}$:
\begin{eqnarray}\label{eq:def_VaR}
    \mathbb{P}(Y_{t}\leqslant Q_{t-1,\tau}|\mathcal{F}_{t-1})=\tau,\quad Q_{t-1,\tau}\in \mathcal{F}_{t-1},\quad \tau\in(0,1).
\end{eqnarray}
 In the context of financial risk management, this quantile is the Value-at-Risk (VaR) at risk level $\tau$. Following Example \ref{ex:EPA} one may evaluate the time $t$ VaR forecast using the tick loss function,
\begin{equation}\label{eq:def_tick_loss}
    \mathcal{L}_{\mathrm{tick},\tau}(Y_t,Q_{t-1,\tau})= (\tau - \mathbf{1}\{Y_t-Q_{t-1,\tau}<0\})(Y_t-Q_{t-1,\tau}).
\end{equation}
A common approach to quantifying VaR is by means of a GARCH-type model: suppose that $Y_t = \sigma_t Z_t$, with $\sigma_t > 0$ a function of past returns and $Z_t$ a mean-zero variable independent of $\sigma_t \in \mathcal{F}_{t-1}$. Then the VaR of $Y_t$ is given by $Q_{t-1,\tau}=c_{\tau}\sigma_t$, where the constant $c_\tau$ is the $\tau$-quantile of the distribution of $Z_t$. Consequently, the tick loss is 
\begin{equation*}
   \mathcal{L}_{\mathrm{tick},\tau}(Y_t,Q_{t-1,\tau}) = (\tau - \mathbf{1}
    \{Z_t<c_{\tau}\})(Z_t-c_\tau)\sigma_t.
\end{equation*}
For standard GARCH-type processes, $\sigma_t^2$ can be written as a stochastic recurrence equation (see, e.g., Section \ref{sec:SRE} in the Supplemental Appendix), and under mild conditions, $\sigma_t$ is regularly varying. By Breiman's Lemma, the tick loss $\mathcal{L}_{\mathrm{tick},t}(Y_t,Q_{t-1,\tau})$ is also regularly varying with index $\kappa$, and, likewise, loss differentials may be regularly varying.
Similar to \citet[Section 5]{patton2019_es_var}, we compare the VaR forecasts from ten different methods. The methods can be divided into three groups: rolling window-based estimation, GARCH models, and score-driven models. The rolling window methods, labeled "RW-$H$", compute the risk forecasts using a simple rolling window of $H$ days with $H\in\{125,250,500\}$. Three of the GARCH-based methods, labeled "G-N", "G-Skt", and "G-EDF", compute the forecasts using a standard GARCH(1,1) model with, respectively, the $\mathrm{N}(0,1)$-distribution, a skewed Student's $t$-distribution, and the empirical distribution of the in-sample standardized residuals. The method "G-FZ" computes the risk forecast based on a GARCH(1,1) model estimated under a so-called Fissler-Ziegel loss function; see \citet[Section 2.5.1]{patton2019_es_var}. The three methods based on score-driven models are labeled "FZ-2F", "FZ-1F", and "Hybrid" and are described in \citet[Section 2]{patton2019_es_var}.\footnote{Along the lines of \citet[Section 5]{patton2019_es_var}, the conditional mean of the returns is given by a constant for the GARCH and score-driven models.}  Our dataset\footnote{The data were retrieved from the Board of Governors of the Federal Reserve System webpage, \url{https://www.federalreserve.gov/datadownload/}.} contains daily continuously compounded returns on the KRWUSD rate from 4 January 1990 to 26 December 2025. The first ten years of data (4 January 1990 to 31 December 1999) are used for the estimation of the model parameters. These estimates are kept fixed for the entire out-of-sample period (1 January 2000 to 26 December 2025). As in \citet[Section 5]{patton2019_es_var}, we remove all observations with returns exactly zero, yielding a total of 2225 in-sample observations for parameter estimation and $n=6333$ out-of-sample observations for forecast comparisons. Throughout, the risk level $\tau$ is set to 0.05. 
For any two competing VaR forecasts $Q_{1, t-1,\tau}$ and $Q_{2, t-1,\tau}$, the time $t$ tick loss differential is
\begin{equation}\label{eq:tickloss_diff}
    X_t = \mathcal{L}_{\mathrm{tick},\tau}(Y_t,Q_{1, t-1,\tau}) - \mathcal{L}_{\mathrm{tick},\tau}(Y_t,Q_{2, t-1,\tau}).
\end{equation}
Following Example \ref{ex:EPA}, the two series have equal predictive ability under $\mathsf{H}_0:\mathbb{E}[X_t]=0$, and we test this EPA hypothesis for all pairs of the 10 forecast series. Panel A of Table \ref{tab:EPA_KRWUSD_tickLoss} presents Diebold--Mariano $t$-statistics on the loss differentials\footnote{As in \citet[Section 5]{patton2019_es_var} the statistics are based on the Newey-West estimator with a fixed lag-length of 20 days.}, whereas Panel B presents the statistic $T_n$ on the differentials along with subsampling-based critical values computed by Algorithm \ref{algo} with block size given by \eqref{eq:b_n_choice}. The statistics are computed as "row method minus column method", so that a negative number indicates that the row method has a smaller average loss than the column method. All grey-shaded entries imply a rejection of the EPA hypothesis at a 5\% nominal level. Notably, Panel B has fewer grey-shaded entries than Panel A, that is, the EPA hypothesis is rejected for fewer pairs when relying on the subsampling-based test. For instance, based on the reported Diebold--Mariano statistics and their sign, one would tend to conclude that the forecasts based on the G-FZ method are more accurate than each of the nine competing forecast series. In contrast, applying the subsampling-based test, the EPA hypothesis is not rejected when comparing the G-FZ method with any of the rolling window-based forecasts. A plausible explanation for the discrepancy between the outcomes of the two tests can be found in Figure \ref{fig:hill_KRWUSD_tickloss}, containing the loss differentials between the RW-500 and G-FZ forecasts and a Hill plot of the absolute value of the loss differentials. The loss differentials have an estimated tail index around 1.5, indicating that the series has infinite variance but a well-defined mean. We also note that the loss differential series appears to exhibit extremal skewness. To see this, the solid red lines indicate plus/minus the empirical 99th percentile of the absolute value of the loss differentials. Most of the loss differentials exceeding these thresholds (in magnitude) are positive, suggesting that rolling window-based forecasts are much more prone to exhibiting extreme losses than their counterparts. Following \citet[Section 2.1]{davis2018tail_process_inference} the tail-balance coefficient $p_+$ in \eqref{RV_marg} can be estimated by the proportion of positive loss differentials exceeding a large threshold. Using the aforementioned empirical 99th percentile as the threshold, we obtain an estimate $\hat{p}_+\approx 0.86$. This extremal skewness is also reflected in the reported subsampled equal-tailed critical values, suggesting that the statistic $T_n$, and hence the loss differentials, has a skewed distribution. Recall from the simulation experiments in Section \ref{sec:Monte_Carlo} that in such a setting, even with $n\geqslant5000$, the Diebold--Mariano test tends to over-reject, whereas the subsampling-based test has rejection frequencies closer to its nominal level.

\begin{figure}[H]
    \centering
    \includegraphics[width=1\linewidth]{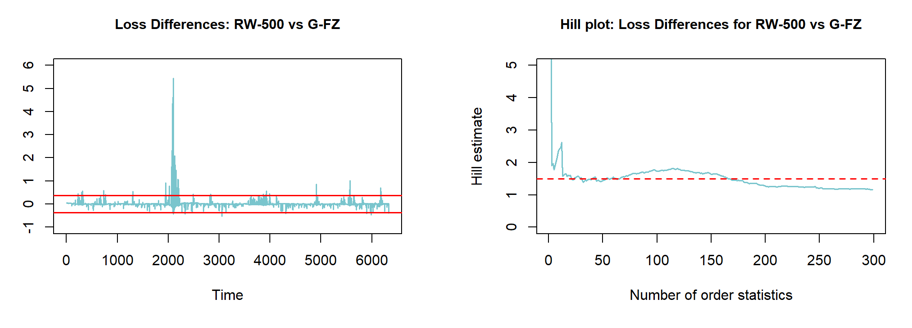}
    \caption{Left: The series of loss differentials when comparing $\mathrm{VaR}$-forecasts based on rolling-window and GARCH approaches, respectively. The red lines indicate plus/minus the empirical 99th percentile of the series in absolute values. Right: Hill plot based on the absolute values, where the dashed red line indicates $\kappa =1.5$.}\label{fig:hill_KRWUSD_tickloss}
\end{figure}

 \newpage

\begin{table}[H]
\centering
\caption{Test statistics for EPA tests}\label{tab:EPA_KRWUSD_tickLoss}
\scriptsize
\setlength{\tabcolsep}{3pt}
\renewcommand{\arraystretch}{1.35}
\resizebox{\linewidth}{!}{%
\begin{tabular}{lcccccccccc}
\multicolumn{11}{c}{\small{A: Diebold--Mariano $t$-statistics, $T_n^{\mathrm{DM}}$}}\\
\toprule
  & \textbf{RW-125} & \textbf{RW-250} & \textbf{RW-500} & \textbf{G-N} & \textbf{G-Skt} & \textbf{G-EDF} & \textbf{FZ-2F} & \textbf{FZ-1F} & \textbf{G-FZ} & \textbf{Hybrid} \\
\midrule
\textbf{RW-125} &   &
    \cellcolor{siggray}{-2.177} &
    \cellcolor{siggray}{-2.774} &
    \cellcolor{siggray}{2.089} &
    1.028 &
    1.879 &
    \cellcolor{siggray}{-3.206} &
    -0.475 &
    \cellcolor{siggray}{2.962} &
    0.879 \\
\textbf{RW-250} &
    \cellcolor{siggray}{2.177} &   &
    \cellcolor{siggray}{-2.470} &
    \cellcolor{siggray}{2.525} &
    1.769 &
    \cellcolor{siggray}{2.382} &
    -0.860 &
    0.853 &
    \cellcolor{siggray}{3.197} &
    1.558 \\
\textbf{RW-500} &
    \cellcolor{siggray}{2.774} &
    \cellcolor{siggray}{2.470} &   &
    \cellcolor{siggray}{3.212} &
    \cellcolor{siggray}{2.691} &
    \cellcolor{siggray}{3.118} &
    1.156 &
    \cellcolor{siggray}{2.181} &
    \cellcolor{siggray}{3.764} &
    \cellcolor{siggray}{2.437} \\
\textbf{G-N} &
    \cellcolor{siggray}{-2.089} &
    \cellcolor{siggray}{-2.525} &
    \cellcolor{siggray}{-3.212} &   &
    \cellcolor{siggray}{-4.860} &
    \cellcolor{siggray}{-3.113} &
    \cellcolor{siggray}{-4.165} &
    \cellcolor{siggray}{-3.471} &
    \cellcolor{siggray}{2.531} &
    \cellcolor{siggray}{-2.474} \\
\textbf{G-Skt} &
    -1.028 &
    -1.769 &
    \cellcolor{siggray}{-2.691} &
    \cellcolor{siggray}{4.860} &   &
    \cellcolor{siggray}{5.629} &
    \cellcolor{siggray}{-3.361} &
    -1.850 &
    \cellcolor{siggray}{4.864} &
    -0.060 \\
\textbf{G-EDF} &
    -1.879 &
    \cellcolor{siggray}{-2.382} &
    \cellcolor{siggray}{-3.118} &
    \cellcolor{siggray}{3.113} &
    \cellcolor{siggray}{-5.629} &   &
    \cellcolor{siggray}{-4.061} &
    \cellcolor{siggray}{-3.103} &
    \cellcolor{siggray}{3.203} &
    -1.781 \\
\textbf{FZ-2F} &
    \cellcolor{siggray}{3.206} &
    0.860 &
    -1.156 &
    \cellcolor{siggray}{4.165} &
    \cellcolor{siggray}{3.361} &
    \cellcolor{siggray}{4.061} &   &
    1.799 &
    \cellcolor{siggray}{5.093} &
    \cellcolor{siggray}{2.555} \\
\textbf{FZ-1F} &
    0.475 &
    -0.853 &
    \cellcolor{siggray}{-2.181} &
    \cellcolor{siggray}{3.471} &
    1.850 &
    \cellcolor{siggray}{3.103} &
    -1.799 &   &
    \cellcolor{siggray}{5.495} &
    \cellcolor{siggray}{2.047} \\
\textbf{G-FZ} &
    \cellcolor{siggray}{-2.962} &
    \cellcolor{siggray}{-3.197} &
    \cellcolor{siggray}{-3.764} &
    \cellcolor{siggray}{-2.531} &
    \cellcolor{siggray}{-4.864} &
    \cellcolor{siggray}{-3.203} &
    \cellcolor{siggray}{-5.093} &
    \cellcolor{siggray}{-5.495} &   &
    \cellcolor{siggray}{-5.005} \\
\textbf{Hybrid} &
    -0.879 &
    -1.558 &
    \cellcolor{siggray}{-2.437} &
    \cellcolor{siggray}{2.474} &
    0.060 &
    1.781 &
    \cellcolor{siggray}{-2.555} &
    \cellcolor{siggray}{-2.047} &
    \cellcolor{siggray}{5.005} &   \\
\midrule
\\[1em] % Adds 1em of extra vertical space
\multicolumn{11}{c}{\small{B: $T_n$ and equal-tailed critical values}} \\
\toprule
  & \textbf{RW-125} & \textbf{RW-250} & \textbf{RW-500} & \textbf{G-N} & \textbf{G-Skt} & \textbf{G-EDF} & \textbf{FZ-2F} & \textbf{FZ-1F} & \textbf{G-FZ} & \textbf{Hybrid} \\
\midrule
\textbf{RW-125} &
  &
$\underset{\{-6.06,\,2.76\}}{-3.24}$ &
$\underset{\{-7.47,\,2.39\}}{-4.55}$ &
$\underset{\{-2.05,\,2.49\}}{2.38}$ &
$\underset{\{-1.98,\,2.07\}}{1.18}$ &
$\underset{\{-1.93,\,2.32\}}{2.16}$ &
\cellcolor{siggray}$\underset{\{-2.28,\,1.58\}}{-3.30}$ &
$\underset{\{-3.43,\,2.43\}}{-0.62}$ &
$\underset{\{-1.98,\,3.61\}}{3.42}$ &
$\underset{\{-4.22,\,3.38\}}{1.01}$ \\
\textbf{RW-250} &
$\underset{\{-2.76,\,6.06\}}{3.24}$ &
  &
$\underset{\{-8.01,\,4.37\}}{-4.24}$ &
$\underset{\{-2.03,\,4.29\}}{3.17}$ &
$\underset{\{-2.07,\,3.52\}}{2.26}$ &
$\underset{\{-2.01,\,3.96\}}{3.02}$ &
$\underset{\{-2.08,\,2.27\}}{-1.02}$ &
$\underset{\{-3.66,\,3.50\}}{1.07}$ &
$\underset{\{-2.11,\,4.89\}}{3.98}$ &
$\underset{\{-3.42,\,5.02\}}{1.89}$ \\
\textbf{RW-500} &
$\underset{\{-2.39,\,7.47\}}{4.55}$ &
$\underset{\{-4.37,\,8.01\}}{4.24}$ &
  &
$\underset{\{-2.04,\,5.14\}}{4.54}$ &
$\underset{\{-1.93,\,4.21\}}{3.91}$ &
$\underset{\{-1.93,\,4.83\}}{4.47}$ &
$\underset{\{-1.96,\,3.02\}}{1.69}$ &
$\underset{\{-2.93,\,5.20\}}{2.85}$ &
$\underset{\{-2.09,\,6.29\}}{5.23}$ &
$\underset{\{-4.60,\,7.02\}}{3.20}$ \\
\textbf{G-N} &
$\underset{\{-2.49,\,2.05\}}{-2.38}$ &
$\underset{\{-4.29,\,2.03\}}{-3.17}$ &
$\underset{\{-5.14,\,2.04\}}{-4.54}$ &
  &
\cellcolor{siggray}$\underset{\{-1.75,\,0.98\}}{-4.42}$ &
\cellcolor{siggray}$\underset{\{-1.56,\,1.44\}}{-2.72}$ &
\cellcolor{siggray}$\underset{\{-1.79,\,1.15\}}{-4.52}$ &
\cellcolor{siggray}$\underset{\{-2.77,\,1.39\}}{-3.95}$ &
\cellcolor{siggray}$\underset{\{-2.02,\,1.89\}}{2.24}$ &
$\underset{\{-2.84,\,1.09\}}{-2.37}$ \\
\textbf{G-Skt} &
$\underset{\{-2.07,\,1.98\}}{-1.18}$ &
$\underset{\{-3.52,\,2.07\}}{-2.26}$ &
$\underset{\{-4.21,\,1.93\}}{-3.91}$ &
\cellcolor{siggray}$\underset{\{-0.98,\,1.75\}}{4.42}$ &
  &
\cellcolor{siggray}$\underset{\{-0.81,\,1.88\}}{5.25}$ &
\cellcolor{siggray}$\underset{\{-1.97,\,1.91\}}{-3.88}$ &
$\underset{\{-2.68,\,1.38\}}{-1.99}$ &
\cellcolor{siggray}$\underset{\{-1.76,\,1.79\}}{4.14}$ &
$\underset{\{-2.36,\,1.22\}}{-0.05}$ \\
\textbf{G-EDF} &
$\underset{\{-2.32,\,1.93\}}{-2.16}$ &
$\underset{\{-3.96,\,2.01\}}{-3.02}$ &
$\underset{\{-4.83,\,1.93\}}{-4.47}$ &
\cellcolor{siggray}$\underset{\{-1.44,\,1.56\}}{2.72}$ &
\cellcolor{siggray}$\underset{\{-1.88,\,0.81\}}{-5.25}$ &
  &
\cellcolor{siggray}$\underset{\{-1.85,\,1.47\}}{-4.51}$ &
\cellcolor{siggray}$\underset{\{-2.75,\,1.31\}}{-3.49}$ &
\cellcolor{siggray}$\underset{\{-2.14,\,1.77\}}{2.75}$ &
$\underset{\{-2.65,\,1.14\}}{-1.67}$ \\
\textbf{FZ-2F} &
\cellcolor{siggray}$\underset{\{-1.58,\,2.28\}}{3.30}$ &
$\underset{\{-2.27,\,2.08\}}{1.02}$ &
$\underset{\{-3.02,\,1.96\}}{-1.69}$ &
\cellcolor{siggray}$\underset{\{-1.15,\,1.79\}}{4.52}$ &
\cellcolor{siggray}$\underset{\{-1.91,\,1.97\}}{3.88}$ &
\cellcolor{siggray}$\underset{\{-1.47,\,1.85\}}{4.51}$ &
  &
\cellcolor{siggray}$\underset{\{-1.40,\,1.63\}}{2.02}$ &
\cellcolor{siggray}$\underset{\{-0.42,\,1.76\}}{5.33}$ &
\cellcolor{siggray}$\underset{\{-1.15,\,1.45\}}{2.62}$ \\
\textbf{FZ-1F} &
$\underset{\{-2.43,\,3.43\}}{0.62}$ &
$\underset{\{-3.50,\,3.66\}}{-1.07}$ &
$\underset{\{-5.20,\,2.93\}}{-2.85}$ &
\cellcolor{siggray}$\underset{\{-1.39,\,2.77\}}{3.95}$ &
$\underset{\{-1.38,\,2.68\}}{1.99}$ &
\cellcolor{siggray}$\underset{\{-1.31,\,2.75\}}{3.49}$ &
\cellcolor{siggray}$\underset{\{-1.63,\,1.40\}}{-2.02}$ &
  &
\cellcolor{siggray}$\underset{\{-1.05,\,3.35\}}{6.56}$ &
$\underset{\{-2.75,\,2.94\}}{2.30}$ \\
\textbf{G-FZ} &
$\underset{\{-3.61,\,1.98\}}{-3.42}$ &
$\underset{\{-4.89,\,2.11\}}{-3.98}$ &
$\underset{\{-6.29,\,2.09\}}{-5.23}$ &
\cellcolor{siggray}$\underset{\{-1.89,\,2.02\}}{-2.24}$ &
\cellcolor{siggray}$\underset{\{-1.79,\,1.76\}}{-4.14}$ &
\cellcolor{siggray}$\underset{\{-1.77,\,2.14\}}{-2.75}$ &
\cellcolor{siggray}$\underset{\{-1.76,\,0.42\}}{-5.33}$ &
\cellcolor{siggray}$\underset{\{-3.35,\,1.05\}}{-6.56}$ &
  &
\cellcolor{siggray}$\underset{\{-4.08,\,0.62\}}{-5.21}$ \\
\textbf{Hybrid} &
$\underset{\{-3.38,\,4.22\}}{-1.01}$ &
$\underset{\{-5.02,\,3.42\}}{-1.89}$ &
$\underset{\{-7.02,\,4.60\}}{-3.20}$ &
$\underset{\{-1.09,\,2.84\}}{2.37}$ &
$\underset{\{-1.22,\,2.36\}}{0.05}$ &
$\underset{\{-1.14,\,2.65\}}{1.67}$ &
\cellcolor{siggray}$\underset{\{-1.45,\,1.15\}}{-2.62}$ &
$\underset{\{-2.94,\,2.75\}}{-2.30}$ &
\cellcolor{siggray}$\underset{\{-0.62,\,4.08\}}{5.21}$ &
  \\
\bottomrule
\end{tabular}}
\end{table}
\vspace{-0.5em}
\begin{singlespace}
\noindent {\footnotesize Notes: Panel A contains $t$-statistics, $T_n^{\mathrm{DM}}$, from the Diebold--Mariano test of EPA, using the tick loss function with risk level $\tau=0.05$, over the out-of-sample period from January 2000 to December 2025, for ten different forecasting methods. Method "RW-$H$" computes the risk forecasts using a simple rolling window of $H$ days. Methods "G-N", "G-Skt", and "G-EDF" compute the forecasts using a standard GARCH(1,1) model with, respectively, the $N(0,1)$-distribution, a skewed Student's $t$-distribution, and the empirical distribution of the in-sample standardized residuals. The "FZ-2F", "FZ-1F", "G-FZ", and "Hybrid" methods are described in \citet[Section 2]{patton2019_es_var}. The statistics are based on the Newey-West HAC estimator with 20 lags. A positive value indicates that the row method has higher average loss than the column method.  Values exceeding 1.96 in absolute value (grey shaded) imply a rejection of the EPA hypothesis at a 5\% nominal level. Values along the main diagonal are all, by convention, identically zero and are omitted. Panel B contains $t$-statistics, $T_n$, for the null of EPA along with equal-tailed  critical values (in brackets) at 5\% nominal level computed according to Algorithm \ref{algo}. Values outside of the interval formed by the critical values (grey shaded) imply a rejection of the EPA hypothesis.} 
\end{singlespace}

 \section{Conclusion}
Many key variables in economics and finance exhibit heavy-tailed behavior, and the forecast errors and loss differentials derived from them can inherit this feature. We show that standard tools for forecast evaluation, including the widely used Diebold--Mariano test, can be severely unreliable in such settings: under empirically plausible heavy-tailed data-generating processes, a test at the 5\% nominal level can reject a true null of equal predictive ability as often as 70\% of the time. The root cause is that the standard normal critical values are justified by a finite-variance central limit theorem that breaks down when the loss differentials have infinite variance.

To address this, we develop a novel stable limit theorem for strongly mixing time series processes under infinite (or potentially undefined) variance, extending the classical central limit theorem to the heavy-tailed setting. Building on this result, we derive the limiting distribution of the test statistic, which is a ratio of dependent non-Gaussian random variables. The quantiles of this ratio can differ markedly from those of a standard normal distribution. We then consider a subsampling-based inferential procedure that is valid regardless of tail-heaviness and requires no prior knowledge of the tail index, the convergence rate of the sample mean, or estimation of long-run variances.

An empirical application to risk forecasts for emerging-market foreign exchange rates, where the loss differentials appear to have infinite variance, illustrates the practical relevance: the standard Diebold--Mariano test and the subsampling-based test yield substantively different conclusions about equal predictive ability. In particular, simple rolling-window risk forecasts are not found to be inferior to forecasts based on more sophisticated GARCH and score-driven models once the heavy-tailed nature of the loss differentials is properly accounted for.

\subsubsection*{Declaration of generative AI and AI-assisted technologies in the manuscript preparation process}
During the preparation of this manuscript, the authors used Opus 4.6 as a language-editing tool to improve clarity and readability, particularly in the Introduction. All content was subsequently reviewed and edited by the authors, who take full responsibility for the manuscript.
\newpage 
\bibliographystyle{ecta}
\bibliography{Mybibliography}

\begin{thebibliography}{45}
\newcommand{\enquote}[1]{``#1''}
\expandafter\ifx\csname natexlab\endcsname\relax\def\natexlab#1{#1}\fi

\bibitem[\protect\citeauthoryear{Bai, Taqqu, and Zhang}{Bai
  et~al.}{2016}]{bai2016unified}
\textsc{Bai, S., M.~S. Taqqu, and T.~Zhang} (2016): \enquote{A unified approach
  to self-normalized block sampling,} \emph{Stochastic Processes and their
  Applications}, 126, 2465--2493.

\bibitem[\protect\citeauthoryear{Barendse and Patton}{Barendse and
  Patton}{2022}]{BarendsePatton2022}
\textsc{Barendse, S. and A.~J. Patton} (2022): \enquote{Comparing predictive
  accuracy in the presence of a loss function shape parameter,} \emph{Journal
  of Business \& Economic Statistics}, 40, 1057--1069.

\bibitem[\protect\citeauthoryear{Bartkiewicz, Jakubowski, Mikosch, and
  Wintenberger}{Bartkiewicz et~al.}{2011}]{bartkiewicz2011stable}
\textsc{Bartkiewicz, K., A.~Jakubowski, T.~Mikosch, and O.~Wintenberger}
  (2011): \enquote{Stable limits for sums of dependent infinite variance random
  variables,} \emph{Probability Theory and Related Fields}, 150, 337--372.

\bibitem[\protect\citeauthoryear{Basrak and Segers}{Basrak and
  Segers}{2009}]{basrak2009regularly}
\textsc{Basrak, B. and J.~Segers} (2009): \enquote{Regularly varying
  multivariate time series,} \emph{Stochastic Processes and their
  Applications}, 119, 1055--1080.

\bibitem[\protect\citeauthoryear{Bingham, Goldie, and Teugels}{Bingham
  et~al.}{1989}]{bingham1989regular}
\textsc{Bingham, N.~H., C.~M. Goldie, and J.~L. Teugels} (1989): \emph{Regular
  Variation}, Cambridge: Cambridge University Press.

\bibitem[\protect\citeauthoryear{Buraczewski, Damek, and Mikosch}{Buraczewski
  et~al.}{2016}]{buraczewski2016stochastic}
\textsc{Buraczewski, D., E.~Damek, and T.~Mikosch} (2016): \emph{Stochastic
  Models with Power-Law Tails: The Equation $X=AX+B$}, Switzerland: Springer.

\bibitem[\protect\citeauthoryear{Carrasco and Chen}{Carrasco and
  Chen}{2002}]{Carrasco_Chen_2002}
\textsc{Carrasco, M. and X.~Chen} (2002): \enquote{Mixing and moment properties
  of various GARCH and stochastic volatility models,} \emph{Econometric
  Theory}, 18, 17--39.

\bibitem[\protect\citeauthoryear{Christensen and Varneskov}{Christensen and
  Varneskov}{2021}]{christensen_varneskov_2021_FX}
\textsc{Christensen, B.~J. and R.~T. Varneskov} (2021): \enquote{Dynamic global
  currency hedging,} \emph{Journal of Financial Econometrics}, 19, 97--127.

\bibitem[\protect\citeauthoryear{Davis}{Davis}{1983}]{Davis1983stablelimit}
\textsc{Davis, R.~A.} (1983): \enquote{Stable limits for partial sums of
  dependent random variables,} \emph{The Annals of Probability}, 11, 262--269.

\bibitem[\protect\citeauthoryear{Davis, Drees, Segers, and
  Warcho^^c5^^82}{Davis et~al.}{2018}]{davis2018tail_process_inference}
\textsc{Davis, R.~A., H.~Drees, J.~Segers, and M.~Warcho^^c5^^82} (2018):
  \enquote{Inference on the tail process with application to financial time
  series modeling,} \emph{Journal of Econometrics}, 205, 508--525.

\bibitem[\protect\citeauthoryear{Davis and Hsing}{Davis and
  Hsing}{1995}]{davis1995point}
\textsc{Davis, R.~A. and T.~Hsing} (1995): \enquote{Point process and partial
  sum convergence for weakly dependent random variables with infinite
  variance,} \emph{The Annals of Probability}, 23, 879--917.

\bibitem[\protect\citeauthoryear{Diebold and Mariano}{Diebold and
  Mariano}{1995}]{diebold1995paring}
\textsc{Diebold, F.~X. and R.~S. Mariano} (1995): \enquote{Comparing predictive
  accuracy,} \emph{Journal of Business \& Economic Statistics}, 13, 253--263.

\bibitem[\protect\citeauthoryear{Doukhan}{Doukhan}{1994}]{Doukhan1994mixing}
\textsc{Doukhan, P.} (1994): \emph{Mixing: Properties and Examples}, Lecture
  Notes in Statistics, New York, NY: Springer-Verlag.

\bibitem[\protect\citeauthoryear{Feller}{Feller}{1971}]{Feller1971}
\textsc{Feller, W.} (1971): \emph{An Introduction to Probability Theory and Its
  Applications}, vol.~2, New York, NY: Wiley.

\bibitem[\protect\citeauthoryear{Fern{\'a}ndez and Steel}{Fern{\'a}ndez and
  Steel}{1998}]{fernandez1998bayesian}
\textsc{Fern{\'a}ndez, C. and M.~F. Steel} (1998): \enquote{On Bayesian
  modeling of fat tails and skewness,} \emph{Journal of the American
  Statistical Association}, 93, 359--371.

\bibitem[\protect\citeauthoryear{Gabaix}{Gabaix}{2009}]{gabaix2009power}
\textsc{Gabaix, X.} (2009): \enquote{Power laws in economics and finance,}
  \emph{Annual Review of Economics}, 1, 255--93.

\bibitem[\protect\citeauthoryear{Gabaix, Gopikrishnan, Plerou, and
  Stanley}{Gabaix et~al.}{2006}]{gabaix2006_QJE}
\textsc{Gabaix, X., P.~Gopikrishnan, V.~Plerou, and H.~E. Stanley} (2006):
  \enquote{Institutional investors and stock market volatility,} \emph{The
  Quarterly Journal of Economics}, 121, 461--504.

\bibitem[\protect\citeauthoryear{Giacomini and Komunjer}{Giacomini and
  Komunjer}{2005}]{giacomini_komunjer2005quantile}
\textsc{Giacomini, R. and I.~Komunjer} (2005): \enquote{Evaluation and
  combination of conditional quantile forecasts,} \emph{Journal of Business \&
  Economic Statistics}, 23, 416--431.

\bibitem[\protect\citeauthoryear{Giacomini and White}{Giacomini and
  White}{2006}]{GiacominiWhite2006}
\textsc{Giacomini, R. and H.~White} (2006): \enquote{Tests of conditional
  predictive ability,} \emph{Econometrica}, 74, 1545--1578.

\bibitem[\protect\citeauthoryear{Hansen}{Hansen}{2005}]{hansen2005SPA}
\textsc{Hansen, P.~R.} (2005): \enquote{A test for superior predictive
  ability,} \emph{Journal of Business \& Economic Statistics}, 23, 365--380.

\bibitem[\protect\citeauthoryear{Hansen and Lunde}{Hansen and
  Lunde}{2005}]{hansen_lunde2005}
\textsc{Hansen, P.~R. and A.~Lunde} (2005): \enquote{A forecast comparison of
  volatility models: does anything beat a GARCH(1,1)?} \emph{Journal of Applied
  Econometrics}, 20, 873--889.

\bibitem[\protect\citeauthoryear{Hansen, Lunde, and Nason}{Hansen
  et~al.}{2011}]{hansen_lunde_nason2011MCS}
\textsc{Hansen, P.~R., A.~Lunde, and J.~M. Nason} (2011): \enquote{The model
  confidence set,} \emph{Econometrica}, 79, 453--497.

\bibitem[\protect\citeauthoryear{Harvey, Leybourne, and Newbold}{Harvey
  et~al.}{1998}]{harvey1998encompassing}
\textsc{Harvey, D.~I., S.~J. Leybourne, and P.~Newbold} (1998): \enquote{Tests
  for forecast encompassing,} \emph{Journal of Business \& Economic
  Statistics}, 16, 254--259.

\bibitem[\protect\citeauthoryear{Ibragimov}{Ibragimov}{1962}]{ibragimov1962clt}
\textsc{Ibragimov, I.~A.} (1962): \enquote{Some limit theorems for stationary
  processes,} \emph{Theory of Probability \& Its Applications}, 7, 349--382.

\bibitem[\protect\citeauthoryear{Ibragimov, Ibragimov, and Kattuman}{Ibragimov
  et~al.}{2013}]{ibragimov2013emerging_markets}
\textsc{Ibragimov, M., R.~Ibragimov, and P.~Kattuman} (2013): \enquote{Emerging
  markets and heavy tails,} \emph{Journal of Banking \& Finance}, 37,
  2546--2559.

\bibitem[\protect\citeauthoryear{Ibragimov, Ibragimov, and Walden}{Ibragimov
  et~al.}{2015}]{ibragimov2015heavytails}
\textsc{Ibragimov, M., R.~Ibragimov, and J.~Walden} (2015): \emph{Heavy-Tailed
  Distributions and Robustness in Economics and Finance}, Lecture Notes in
  Statistics, Heidelberg: Springer.

\bibitem[\protect\citeauthoryear{Kim, Meddahi, and Yamashita}{Kim
  et~al.}{2021}]{kim2021forecast}
\textsc{Kim, J., N.~Meddahi, and M.~Yamashita} (2021): \enquote{Forecast
  comparison tests under fat-tails,} Working paper accessed on 24 November 2025
  via
  \url{https://www.ecb.europa.eu/press/conferences/shared/pdf/20210615_11th_cft/Yamashita_paperT11.pdf}.

\bibitem[\protect\citeauthoryear{Kokoszka and Wolf}{Kokoszka and
  Wolf}{2004}]{kokoszka2004subsampling}
\textsc{Kokoszka, P. and M.~Wolf} (2004): \enquote{Subsampling the mean of
  heavy-tailed dependent observations,} \emph{Journal of Time Series Analysis},
  25, 217--234.

\bibitem[\protect\citeauthoryear{Li, Liao, and Quaedvlieg}{Li
  et~al.}{2021}]{Li2021CSPA}
\textsc{Li, J., Z.~Liao, and R.~Quaedvlieg} (2021): \enquote{Conditional
  superior predictive ability,} \emph{The Review of Economic Studies}, 89,
  843--875.

\bibitem[\protect\citeauthoryear{Logan, Mallows, Rice, and Shepp}{Logan
  et~al.}{1973}]{Logan1973self-norm}
\textsc{Logan, B.~F., C.~L. Mallows, S.~O. Rice, and L.~A. Shepp} (1973):
  \enquote{Limit distributions of self-normalized sums,} \emph{The Annals of
  Probability}, 1, 788--809.

\bibitem[\protect\citeauthoryear{Matsui, Mikosch, and Wintenberger}{Matsui
  et~al.}{2025{\natexlab{a}}}]{matsui2025moments}
\textsc{Matsui, M., T.~Mikosch, and O.~Wintenberger} (2025{\natexlab{a}}):
  \enquote{Moments for self-normalized partial sums,} \emph{Stochastic
  Processes and their Applications}, 198, 104810.

\bibitem[\protect\citeauthoryear{Matsui, Mikosch, and Wintenberger}{Matsui
  et~al.}{2025{\natexlab{b}}}]{matsui2025self}
---\hspace{-.1pt}---\hspace{-.1pt}--- (2025{\natexlab{b}}):
  \enquote{Self-normalized partial sums of heavy-tailed time series,}
  \emph{Stochastic Processes and their Applications}, 190, 104729.

\bibitem[\protect\citeauthoryear{McElroy and Politis}{McElroy and
  Politis}{2002}]{McElroy_Politis_2002_subsampling}
\textsc{McElroy, T. and D.~N. Politis} (2002): \enquote{Robust inference for
  the mean in the presence of serial correlation and heavy-tailed
  distributions,} \emph{Econometric Theory}, 18, 1019--1039.

\bibitem[\protect\citeauthoryear{Mikosch and Wintenberger}{Mikosch and
  Wintenberger}{2024}]{mikosch2024extreme}
\textsc{Mikosch, T. and O.~Wintenberger} (2024): \emph{Extreme Value Theory for
  Time Series: Models with Power-Law Tails}, Heidelberg: Springer.

\bibitem[\protect\citeauthoryear{Newey and West}{Newey and
  West}{1987}]{newey_west_1987}
\textsc{Newey, W.~K. and K.~D. West} (1987): \enquote{A simple, positive
  semi-definite, heteroskedasticity and autocorrelation consistent covariance
  matrix,} \emph{Econometrica}, 55, 703--708.

\bibitem[\protect\citeauthoryear{Newey and West}{Newey and
  West}{1994}]{newey_west_1994_automatic}
---\hspace{-.1pt}---\hspace{-.1pt}--- (1994): \enquote{Automatic lag selection
  in covariance matrix estimation,} \emph{The Review of Economic Studies}, 61,
  631--653.

\bibitem[\protect\citeauthoryear{Odendahl, Rossi, and Sekhposyan}{Odendahl
  et~al.}{2023}]{odendahl2023forecast}
\textsc{Odendahl, F., B.~Rossi, and T.~Sekhposyan} (2023): \enquote{Evaluating
  forecast performance with state dependence,} \emph{Journal of Econometrics},
  237, 105220.

\bibitem[\protect\citeauthoryear{Patton and Timmermann}{Patton and
  Timmermann}{2012}]{Patton_timmermann_2012rationality}
\textsc{Patton, A.~J. and A.~Timmermann} (2012): \enquote{Forecast rationality
  tests based on multi-horizon bounds,} \emph{Journal of Business \& Economic
  Statistics}, 30, 1--17.

\bibitem[\protect\citeauthoryear{Patton, Ziegel, and Chen}{Patton
  et~al.}{2019}]{patton2019_es_var}
\textsc{Patton, A.~J., J.~F. Ziegel, and R.~Chen} (2019): \enquote{Dynamic
  semiparametric models for expected shortfall (and Value-at-Risk),}
  \emph{Journal of Econometrics}, 211, 388--413.

\bibitem[\protect\citeauthoryear{Politis, Romano, and Wolf}{Politis
  et~al.}{1999}]{politis2012subsampling}
\textsc{Politis, D.~N., J.~P. Romano, and M.~Wolf} (1999): \emph{Subsampling},
  New York, NY: Springer.

\bibitem[\protect\citeauthoryear{Samorodnitsky and Taqqu}{Samorodnitsky and
  Taqqu}{1994}]{samorodnitsky1994stable}
\textsc{Samorodnitsky, G. and M.~S. Taqqu} (1994): \emph{Stable Non-Gaussian
  Random Processes: Stochastic Models with Infinite Variance}, Boca Raton, FL:
  Chapman \& Hall.

\bibitem[\protect\citeauthoryear{Sato}{Sato}{1999}]{sato1999}
\textsc{Sato, K.-I.} (1999): \emph{L{\'e}vy Processes and Infinitely Divisible
  Distributions}, Cambridge: Cambridge University Press.

\bibitem[\protect\citeauthoryear{West}{West}{1996}]{west1996asymptotic}
\textsc{West, K.~D.} (1996): \enquote{Asymptotic inference about predictive
  ability,} \emph{Econometrica}, 64, 1067--1084.

\bibitem[\protect\citeauthoryear{White}{White}{2000}]{white2000snooping}
\textsc{White, H.} (2000): \enquote{A reality check for data snooping,}
  \emph{Econometrica}, 68, 1097--1126.

\bibitem[\protect\citeauthoryear{Zhang, Ho, Wendler, and Wu}{Zhang
  et~al.}{2013}]{zhang2013blocksampling}
\textsc{Zhang, T., H.-C. Ho, M.~Wendler, and W.~B. Wu} (2013): \enquote{Block
  sampling under strong dependence,} \emph{Stochastic Processes and their
  Applications}, 123, 2323--2339.

\end{thebibliography}

\appendix

\section{Proof of Theorem \ref{theo:inf_var_jointCLT}}\label{sec:appendix_main_theorem}
Throughout the appendix, with $X= \Re(X)+i\Im(X)$ and $Y=\Re(Y)+i\Im(Y)$ complex-valued random variables, the (pseudo-)covariance between $X$ and $Y$ is given by $\mathrm{cov}(X,Y) = \mathbb{E}[XY]-\mathbb{E}[X]\mathbb{E}[Y]$, with $\mathbb{E}[X]=\mathbb{E}[\Re(X)]+i\mathbb{E}[\Im(X)]$. 
\bigskip

Suppose initially that $\mathbb{E}[X_0]=0$ if $\kappa\in(1,2)$, such that $\mu=0$. Then the result follows by verifying the conditions of \citet[Theorem 3.2]{matsui2025self}. Since $(X_t)_{t\in\mathbb{Z}}$ is regularly varying with index $\kappa\in(0,2)\setminus\{1\}$ (Assumption \ref{ass_RV}) and the process obeys the anti-clustering condition in Assumption \ref{ass:AC+MX}, it remains to show that $(X_t)_{t\in\mathbb{Z}}$ obeys a mixing condition similar to (3.21) in \cite{matsui2025self}. Specifically, with $(a_n)_{n\geqslant 1}$ defined by \eqref{eq:def_a_n}, let
\begin{equation}
    \Psi_n(u,\lambda) :=\mathbb{E}[\exp(ia_n^{-1}uS_n-a_n^{-2}\lambda\gamma_{n}^2)],\qquad (u,\lambda)\in \mathbb{R}\times\mathbb{R}_+,
\end{equation}
with $S_n=\sum_{t=1}^nX_t$ and $\gamma_{n}=(\sum_{t=1}^nX_t^2)^{1/2}$. 
Then it suffices to show that for the integer sequence $r_n \to \infty$ in Assumption \ref{ass:AC+MX} and $k_n= \lfloor n/r_n\rfloor\to\infty$ as $n\to \infty$, it holds for all $(u,\lambda)\in \mathbb{R}\times\mathbb{R}_+$ that
\begin{align}\label{eq:mix_condition}
    \Psi_n(u,\lambda)  =  \left(\mathbb{E}[\exp(ia_n^{-1}uS_{r_n}-a_n^{-2}\lambda\gamma_{r_n}^2)]\right)^{k_n}+o(1).
\end{align}
Without loss of generality, suppose throughout that $n/r_n$ is an integer, such that $k_nr_n=n$. Then we have the following block structure of $(X_t)_{t=1,\dots,n}$:
\begin{align*}
&\underset{\text{Large block 1}}{\underbrace{X_{1},\dots ,X_{r_{n}-\ell _{n}}}%
},~\underset{\text{Small block 1}}{\underbrace{X_{r_{n}-\ell _{n}+1},\dots
,X_{r_{n}}}},~\dots, \underset{\text{Large block }k_{n}}{\underbrace{%
X_{r_{n}\left( k_{n}-1\right) +1},\dots ,X_{r_{n}k_{n}-\ell _{n}}}},\underset%
{\text{Small block }k_{n}}{\underbrace{X_{r_{n}k_{n}-\ell _{n}+1},\dots
,X_{r_{n}k_{n}}}}.
\end{align*}%
We write $A_{j}^{L}$ and $A_{j}^{S}$ for the index sets of the $j$th large and
small blocks, respectively, and define
\begin{equation}
    \Psi_{L,n}(u,\lambda) :=\mathbb{E}\left[ \mathrm{e}^{ iua_{n}^{-1}\sum_{s=1}^{k_{n}}\sum_{t\in
A_{s}^{L}}X_{t}-a_{n}^{-2}\lambda \sum_{s=1}^{k_{n}}\sum_{t\in
A_{s}^{L}}X_{t}^{2} }\right],\qquad (u,\lambda)\in \mathbb{R}\times\mathbb{R}_+.
\end{equation}
To prove \eqref{eq:mix_condition}, it suffices to show that the following holds for all $(u,\lambda)\in \mathbb{R}\times\mathbb{R}_+$:
\begin{align}
   & \Psi_n(u,\lambda) = \Psi_{L,n}(u,\lambda) + o(1), \label{eq:mix_step1} \\
& \Psi_{L,n}(u,\lambda) = \left( \mathbb{E}\left[ \mathrm{e}%
^{ iua_{n}^{-1}\sum_{t\in A_{1}^{L}}X_{t}-a_{n}^{-2}\lambda \sum_{t\in
A_{1}^{L}}X_{t}^{2} }\right] \right) ^{k_{n}} + o(1), \label{eq:mix_step2} \\
& \left( \mathbb{E}\left[ \mathrm{e}%
^{ iua_{n}^{-1}\sum_{t\in A_{1}^{L}}X_{t}-a_{n}^{-2}\lambda \sum_{t\in
A_{1}^{L}}X_{t}^{2} }\right] \right) ^{k_{n}} =  \left(\mathbb{E}[\mathrm{e}^{ia_n^{-1}uS_{r_n}-a_n^{-2}\lambda\gamma_{r_n}^2}]\right)^{k_n}+o(1). \label{eq:mix_step3}
\end{align}

\noindent\textit{Proof of \eqref{eq:mix_step1}:}
Using the inequality $|\exp(ix-y)|\leqslant1$ for $x\in\mathbb{R}$ and $y>0$, we have that
\begin{align*}
    &\left\vert \Psi_n(u,\lambda) - \Psi_{L,n}(u,\lambda)\right\vert \\
    & = \left\vert \mathbb{E}\left[\mathrm{e}^{ iua_{n}^{-1}\sum_{s=1}^{k_{n}}\sum_{t\in
A_{s}^{L}}X_{t}-a_{n}^{-2}\lambda \sum_{s=1}^{k_{n}}\sum_{t\in
A_{s}^{L}}X_{t}^{2} }\left(\mathrm{e}^{ iua_{n}^{-1}\sum_{s=1}^{k_{n}}\sum_{t\in
A_{s}^{S}}X_{t}-a_{n}^{-2}\lambda \sum_{s=1}^{k_{n}}\sum_{t\in
A_{s}^{S}}X_{t}^{2} } - 1\right)     \right]   \right\vert \\
&\leqslant \mathbb{E}\left[\left\vert \mathrm{e}^{ iua_{n}^{-1}\sum_{s=1}^{k_{n}}\sum_{t\in
A_{s}^{S}}X_{t}-a_{n}^{-2}\lambda \sum_{s=1}^{k_{n}}\sum_{t\in
A_{s}^{S}}X_{t}^{2} } - 1\right\vert     \right].   
\end{align*}
Since $\vert \mathrm{e}^{ iua_{n}^{-1}\sum_{s=1}^{k_{n}}\sum_{t\in
A_{s}^{S}}X_{t}-a_{n}^{-2}\lambda \sum_{s=1}^{k_{n}}\sum_{t\in
A_{s}^{S}}X_{t}^{2} } - 1\vert \leqslant 2$, \eqref{eq:mix_step1} holds by dominated convergence provided that 
\begin{equation}\label{eq:conv_small_sums}
    \left(a_{n}^{-1}\sum_{s=1}^{k_{n}}\sum_{t\in
A_{s}^{S}}X_{t}, a_{n}^{-2}\sum_{s=1}^{k_{n}}\sum_{t\in
A_{s}^{S}}X_{t}^{2}  \right) \overset{\mathbb{P}}{\to} (0,0),\qquad  n \to \infty.
\end{equation}
From \citet[Proof of Theorem 3.3]{matsui2025self} it holds that under the anti-clustering condition in \eqref{AC}, 
\begin{equation}\label{eq:block_conv1}
 \left( \mathbb{E}[\mathrm{e}^{ia_n^{-1}uS_{r_n}-a_n^{-2}\lambda\gamma_{r_n}^2}] \right)^{k_n} \to \mathbb{E}\big[\mathrm{e}^{iu\xi_\kappa-\lambda\zeta_{\kappa/2}^2}\big], \qquad n\to\infty,
\end{equation}
with $\mathbb{E}\big[\mathrm{e}^{iu\xi_\kappa-\lambda\zeta_{\kappa/2}^2}\big]$ given in \eqref{eq:hybrid_chf_Laplace}.
For later use, we provide an intuitive explanation of why this result holds. 
Consider the following change of expression
\begin{align}
    \left(\mathbb{E}[\mathrm{e}^{ia_n^{-1}uS_{r_n}-a_n^{-2}\lambda\gamma_{r_n}^2}]\right)^{k_n} 
    &= \left(1+k_n^{-1} \left( k_n \big(\mathbb{E}[\mathrm{e}^{ia_n^{-1}uS_{r_n}-a_n^{-2}\lambda\gamma_{r_n}^2}]-1\big)\right)
    \right)^{k_n} \notag \\
    &=:\left(1+k_n^{-1} \psi_n(u,\lambda)\right)^{k_n}.
\end{align}
Under the anti-clustering condition, it follows from \citet[Proof of Theorem 3.3]{matsui2025self} that
\begin{equation}\label{eq:conv_psi_n}
\psi_n(u,\lambda) \to \log \mathbb{E}\big[\mathrm{e}^{iu\xi_\kappa-\lambda\zeta_{\kappa/2}^2}\big],\qquad n\to \infty.    
\end{equation}
Then the well-known formula $\lim_{k\to\infty}(1+k^{-1}c)^k\to e^c,\,c>0$ implies \eqref{eq:block_conv1} as $n\to\infty$. Now, since $\ell_n = o(r_n)$, the anti-clustering condition \eqref{AC} also holds with $r_n$ replaced by $\ell_n$, we also have that 
\[
\left( \mathbb{E}[\mathrm{e}^{ia_n^{-1}uS_{\ell_n}-a_n^{-2}\lambda\gamma_{\ell_n}^2}] \right)^{\tilde{k}_n} \to \mathbb{E}\big[\mathrm{e}^{iu\xi_\kappa-\lambda\zeta_{\kappa/2}^2}\big], \qquad n\to\infty,
\]
with $\tilde{k}_n = n/\ell_n$. Using that $k_n/\tilde{k}_n = \ell_n/r_n=o(1)$, we obtain
\begin{align*}
    \left( \mathbb{E}[\mathrm{e}^{ia_n^{-1}uS_{\ell_n}-a_n^{-2}\lambda\gamma_{\ell_n}^2}] \right)^{k_n} &= \left( \mathbb{E}[\mathrm{e}^{ia_n^{-1}uS_{\ell_n}-a_n^{-2}\lambda\gamma_{\ell_n}^2}] \right)^{\tilde{k}_n\cdot k_n/\tilde{k}_n} \notag \\ 
    &= \left(\mathbb{E}\big[\mathrm{e}^{iu\xi_\kappa-\lambda\zeta_{\kappa/2}^2}\big]+o(1)\right)^{\ell_n/r_n} =1+o(1).
\end{align*}
Noting that the small blocks are of length $\ell_n$, by stationarity, this implies that
\begin{equation}\label{eq:conv_indep_small_blocks}
     \left( \mathbb{E}\left[\mathrm{e}^{ia_n^{-1}u\sum_{t\in A_1^S}X_t-a_n^{-2}\lambda\sum_{t\in A_1^S}X_t^2}\right] \right)^{k_n} = 1+o(1).
\end{equation}
Consequently, \eqref{eq:conv_small_sums} now follows provided that
\begin{equation}\label{eq:small_blocks}
    \mathbb{E}\left[ \mathrm{e}^{ iua_{n}^{-1}\sum_{s=1}^{k_{n}}\sum_{t\in
A_{s}^{S}}X_{t}-a_{n}^{-2}\lambda \sum_{s=1}^{k_{n}}\sum_{t\in
A_{s}^{S}}X_{t}^{2} }\right]  =  \left( \mathbb{E}\left[\mathrm{e}^{ia_n^{-1}u\sum_{t\in A_1^S}X_t-a_n^{-2}\lambda\sum_{t\in A_1^S}X_t^2}\right] \right)^{k_n} + o(1).
\end{equation}
To show \eqref{eq:small_blocks}, we note that
\begin{align*}
  & \left\vert \mathbb{E}\left[ \mathrm{e}^{ iua_{n}^{-1}\sum_{s=1}^{k_{n}}\sum_{t\in
A_{s}^{S}}X_{t}-a_{n}^{-2}\lambda \sum_{s=1}^{k_{n}}\sum_{t\in
A_{s}^{S}}X_{t}^{2} }\right]  -  \left( \mathbb{E}\left[\mathrm{e}^{ia_n^{-1}u\sum_{t\in A_1^S}X_t-a_n^{-2}\lambda\sum_{t\in A_1^S}X_t^2}\right] \right)^{k_n} \right\vert \\
 & = \left\vert \mathbb{E}\left[ \prod_{s=1}^{k_n} \mathrm{e}^{ iua_{n}^{-1}\sum_{t\in
A_{s}^{S}}X_{t}-a_{n}^{-2}\lambda \sum_{t\in
A_{s}^{S}}X_{t}^{2} }\right]  -  \left( \mathbb{E}\left[\mathrm{e}^{ia_n^{-1}u\sum_{t\in A_1^S}X_t-a_n^{-2}\lambda\sum_{t\in A_1^S}X_t^2}\right] \right)^{k_n} \right\vert \\
 & \leqslant \sum_{s=2}^{k_n}\left\vert \mathrm{cov}\left(\prod_{j=1}^{s-1}\mathrm{e}^{ iua_{n}^{-1}\sum_{t\in
A_{j}^{S}}X_{t}-a_{n}^{-2}\lambda \sum_{t\in
A_{j}^{S}}X_{t}^{2} } , \mathrm{e}^{ iua_{n}^{-1}\sum_{t\in
A_{s}^{S}}X_{t}-a_{n}^{-2}\lambda \sum_{t\in
A_{s}^{S}}X_{t}^{2} } \right) \right\vert \\
& \leqslant (k_n-1)\sup_{A\subset \left\{
-n,\dots ,-1,0\right\} }\left\vert \mathrm{cov}\left( \mathrm{e}
^{ia_{n}^{-1}u\sum_{t\in A} X_{t}-a_{n}^{-2}\lambda \sum_{t\in A}X_{t}^{2}},
\mathrm{e}^{ia_{n}^{-1}u\sum_{t=r_n-\ell_n}^{r_{n}}X_{t}-a_{n}^{-2}\lambda \sum_{t=r_n-\ell_n }^{r_{n}}X_{t}^{2}}\right) \right\vert \\
&  \leqslant 4(k_n - 1)\alpha_{r_n-\ell_n},
\end{align*}
where the first inequality follows by Lemma \ref{lem:product_inequality}, the second inequality follows by stationarity and using that the small blocks are separated by $r_n - \ell_n$ observations. The third inequality follows by Lemma \ref{lem:mixing_bound}. Now \eqref{eq:small_blocks} follows by \eqref{alpha_coefs_req}, using that $k_n=n/r_n$ and $\ell_n=o(r_n)$. \\
\smallskip
\noindent\textit{Proof of \eqref{eq:mix_step2}:} The result follows by arguments similar to the ones given to show \eqref{eq:small_blocks} above. In particular, we have that
\begin{align*}
    & \left\vert \Psi_{L,n}(u,\lambda) - \left( \mathbb{E}\left[ \mathrm{e}
^{ iua_{n}^{-1}\sum_{t\in A_{1}^{L}}X_{t}-a_{n}^{-2}\lambda \sum_{t\in
A_{1}^{L}}X_{t}^{2} }\right] \right) ^{k_{n}} \right\vert \\
& \leqslant \sum_{s=2}^{k_n}\left\vert \mathrm{cov}\left(\prod_{j=1}^{s-1}\mathrm{e}^{ iua_{n}^{-1}\sum_{t\in
A_{j}^{L}}X_{t}-a_{n}^{-2}\lambda \sum_{t\in
A_{j}^{L}}X_{t}^{2} } , \mathrm{e}^{ iua_{n}^{-1}\sum_{t\in
A_{s}^{L}}X_{t}-a_{n}^{-2}\lambda \sum_{t\in
A_{s}^{L}}X_{t}^{2} } \right) \right\vert \\
& \leqslant (k_n-1)\sup_{A\subset \left\{
-n,\dots ,-1,0\right\} }\left\vert \mathrm{cov}\left( \mathrm{e}
^{ia_{n}^{-1}u\sum_{t\in A} X_{t}-a_{n}^{-2}\lambda \sum_{t\in A}X_{t}^{2}},
\mathrm{e}^{ia_{n}^{-1}u\sum_{t=\ell_n}^{r_{n}}X_{t}-a_{n}^{-2}\lambda \sum_{t=\ell_n }^{r_{n}}X_{t}^{2}}\right) \right\vert \\
&  \leqslant 4(k_n - 1)\alpha_{\ell_n},
\end{align*}
where the first inequality follows by Lemma \ref{lem:product_inequality}, the second inequality follows by stationarity and the fact that the large blocks are separated by $\ell_n$ observations, and the third inequality follows by Lemma \ref{lem:mixing_bound}. The result now follows by \eqref{alpha_coefs_req}, using that $k_n=n/r_n$. 

\noindent \textit{Proof of \eqref{eq:mix_step3}:}  From \eqref{eq:conv_psi_n} we have that 
\begin{align}
\label{largeO:inv:kn}
    & \left( \mathbb{E}\left[\mathrm{e}^{ia_n^{-1}uS_{r_n}-a_n^{-2}\lambda\gamma_{r_n}^2}\right]   - 1 \right) = O(1/k_n),
\end{align}
and by arguments identical to the ones used for deriving \eqref{eq:conv_psi_n}, it holds that
\begin{eqnarray*}
  \left( \mathbb{E}\left[ \mathrm{e}%
^{ iua_{n}^{-1}\sum_{t\in A_{1}^{L}}X_{t}-a_{n}^{-2}\lambda \sum_{t\in
A_{1}^{L}}X_{t}^{2} }\right] -1 \right) = O(1/k_n).
\end{eqnarray*}
With $(r_n,k_n)$ replaced by $(\ell_n,\tilde k_n)$, an argument similar to the one used for deriving \eqref{largeO:inv:kn} from \eqref{eq:conv_psi_n} yields that $ \big(\mathbb{E}[\mathrm{e}^{ia_n^{-1}uS_{\ell_n}-a_n^{-2}\lambda\gamma_{\ell_n}^2}]-1\big) = O(1/\tilde{k}_n)=o(1/k_n)$,
where the last equality follows by using that $k_n/\tilde{k}_n=o(1)$.
Using that
\begin{equation*}
    \mathbb{E}\left[\mathrm{e}^{ia_n^{-1}uS_{r_n}-a_n^{-2}\lambda\gamma_{r_n}^2}\right] = \mathbb{E}\left[ \mathrm{e}%
^{ iua_{n}^{-1}\sum_{t\in A_{1}^{L}}X_{t}-a_{n}^{-2}\lambda \sum_{t\in
A_{1}^{L}}X_{t}^{2} } \mathrm{e}%
^{ iua_{n}^{-1}\sum_{t\in A_{1}^{S}}X_{t}-a_{n}^{-2}\lambda \sum_{t\in
A_{1}^{S}}X_{t}^{2} }\right],
\end{equation*}
the result follows by Lemma \ref{lem:large_small_blocks} below. 
We conclude that the theorem holds for $\mu=0$.

Next, consider the case where $\kappa\in(1,2)$ and $\mathbb{E}[X_0]\ne0$. Then the centered process $(X_t-\mu)_{t\in\mathbb{Z}}$ is stationary and has the same mixing coefficients, tail index, and spectral tail process as $(X_t)_{t\in\mathbb{Z}}$. Consequently, it suffices to show that the centered process obeys the anti-clustering condition as in \eqref{AC}, that is, 
\begin{equation}\label{AC2}
    \lim_{k\to\infty}\limsup_{n\to\infty}\sum_{j=k}^{r_n}\mathbb{E}[(|a_n^{-1}(X_j-\mu)|\wedge 1\big)\big(|a_n^{-1}(X_0-\mu)|\wedge 1)]=0,
\end{equation}
where $x \wedge y:=\min(x,y)$. Using that $(|x-y|\wedge1)\leq (|x|\wedge 1)+(|y|\wedge1)$, by stationarity
\begin{align*}
    &\mathbb{E}[(|a_n^{-1}(X_j-\mu)|\wedge 1)(|a_n^{-1}(X_0-\mu)|\wedge 1)] \\
    &\quad \leq \mathbb{E}[(|a_n^{-1}X_j|\wedge 1)(|a_n^{-1}X_0|\wedge 1)]+(|a_n^{-1}\mu|\wedge1)^2+2\mathbb{E}[(|a_n^{-1}X_0|\wedge 1)](|a_n^{-1}\mu|\wedge 1)] \\
    &\quad=:J_1+J_2+J_3.
\end{align*}
By regular variation of $X_0$, $a_n\sim L(n)n^{1/\kappa}$ as $n\to\infty$, and we have for $\kappa\in(1,2)$,
\begin{equation}\label{eq:rate_a_n}
    n/a_n^2 \to 0\quad \text{and}\quad n/a_n\to \infty.
\end{equation}
From \eqref{AC}, we have that $\lim_{k\to\infty}\limsup_{n\to\infty}\sum_{j=k}^{r_n}J_1=0$. Furthermore, it holds that
\begin{equation*}
    \lim_{k\to\infty}\limsup_{n\to\infty}\sum_{j=k}^{r_n}J_2= \lim_{k\to\infty}\limsup_{n\to\infty}\dfrac{(r_n-k+1)}{n}(n/a_n^{2})(a_n^2\wedge\mu^2)=0,
\end{equation*}
where we used that $r_n=o(n)$ and \eqref{eq:rate_a_n}. Finally, we have that
\begin{align*}
    \mathbb{E}[(|a_n^{-1}X_0|\wedge 1)] &= \mathbb{E}[a_n^{-1}|X_0|\mathbf{1}(|X_0|\leq a_n)]+\mathbb{P}(|X_0|>a_n) \\
    & \leqslant a_n^{-1} \left( \mathbb{E}[|X_0|] + (a_n/n)[n\mathbb{P}(|X_0|>a_n)]  \right), 
\end{align*}
    such that
\begin{align*}
\lim_{k\to\infty}&\limsup_{n\to\infty}\sum_{j=k}^{r_n}J_3\\
    &=\lim_{k\to\infty}\limsup_{n\to\infty}2(r_n-k+1)\mathbb{E}[(|a_n^{-1}X_0|\wedge 1)](|a_n^{-1}\mu|\wedge 1) \\
    &\leqslant \lim_{k\to\infty}\limsup_{n\to\infty}2(r_n-k+1)(|a_n^{-1}\mu|\wedge 1)a_n^{-1} \left( \mathbb{E}[|X_0|] + (a_n/n)[n\mathbb{P}(|X_0|>a_n)] \right) \\
    &= \lim_{k\to\infty}\limsup_{n\to\infty}2\frac{(r_n-k+1)}{n}(|\mu|\wedge a_n)(n/a_n^{2}) \left( \mathbb{E}[|X_0|] + (a_n/n)[n\mathbb{P}(|X_0|>a_n)] \right) =0,
\end{align*}
using that $r_n=o(n)$, \eqref{eq:rate_a_n}, and \eqref{eq:def_a_n}. Collecting terms, we have that \eqref{AC2} holds.\hfill $\square$

\begin{lemma}\label{lem:mixing_bound}
    Let $(X_t)_{t\in\mathbb{Z}}$ be a real-valued stationary process. For $n\geqslant1 $ and $(u,\lambda) \in \mathbb{R}\times \mathbb{R}_+$, define for $\overline{\ell} \geqslant\ell \geqslant 1$
\begin{equation*}
\alpha_{n,\ell,\overline{\ell} }( u,\lambda ) :=\sup_{A\subset \left\{
-n,\dots ,-1,0\right\} }\left\vert \mathrm{cov}\left( \mathrm{e}%
^{ia_{n}^{-1}u\sum_{t\in A}X_{t}-a_{n}^{-2}\lambda \sum_{t\in A}X_{t}^{2}},
\mathrm{e}^{ia_{n}^{-1}u\sum_{t=\ell }^{\overline{\ell}}X_{t}-a_{n}^{-2}\lambda
\sum_{t=\ell }^{\overline{\ell}}X_{t}^{2}}\right) \right\vert.
\end{equation*}  
It holds that for any $(n, u,\lambda) \in \mathbb{N}\times \mathbb{R}\times \mathbb{R}_+$, $\alpha_{n,\ell,\overline{\ell} }( u,\lambda ) \leqslant 4 \alpha_\ell$,
where $(\alpha_k)_{k\geqslant1}$ are the strong mixing coefficients of $(X_t)_{t\in\mathbb{Z}}$ defined in \eqref{mix_coef}.
\end{lemma}
\textbf{Proof of Lemma \ref{lem:mixing_bound}:} 
For the ease of notation, with $A\subset \left\{
-n,\dots ,-1,0\right\}$ define $Y_0 := a_{n}^{-1}u\sum_{t\in A}X_{t}$, $\tilde{Y}_0 := a_{n}^{-2}\lambda \sum_{t\in A}X_{t}^{2}$,  $Y_\ell:= a_{n}^{-1}u\sum_{t=\ell }^{\overline{\ell}}X_{t}$, and $\tilde{Y}_\ell := a_{n}^{-2}\lambda
\sum_{t=\ell }^{\overline{\ell}}X_{t}^{2}$. Then
\begin{align*}
    &  \mathrm{cov}\left( \mathrm{e}%
^{ia_{n}^{-1}u\sum_{t\in A}X_{t}-a_{n}^{-2}\lambda \sum_{t\in A}X_{t}^{2}},%
\mathrm{e}^{ia_{n}^{-1}u\sum_{t=\ell }^{\overline{\ell}}X_{t}-a_{n}^{-2}\lambda
\sum_{t=\ell }^{\overline{\ell}}X_{t}^{2}}\right) \\
& = \mathrm{cov}\left( \cos(Y_0)\mathrm{e}^{-\tilde{Y}_0} + i\sin(Y_0)\mathrm{e}^{-\tilde{Y}_0}, \cos(Y_\ell)\mathrm{e}^{-\tilde{Y}_\ell} + i\sin(Y_\ell)\mathrm{e}^{-\tilde{Y}_\ell}   \right) \\
&  = \mathrm{cov}\left( \cos(Y_0)\mathrm{e}^{-\tilde{Y}_0}, \cos(Y_\ell)\mathrm{e}^{-\tilde{Y}_\ell} \right) + i \mathrm{cov}\left(\sin(Y_0)\mathrm{e}^{-\tilde{Y}_0}, 
\cos(Y_\ell)\mathrm{e}^{-\tilde{Y}_\ell}   \right) \\
&  \quad +i\mathrm{cov}\left(\cos(Y_0)\mathrm{e}^{-\tilde{Y}_0}, \sin(Y_\ell)\mathrm{e}^{-\tilde{Y}_\ell}  \right) - \mathrm{cov}\left( \sin(Y_0)\mathrm{e}^{-\tilde{Y}_0}, \sin(Y_\ell)\mathrm{e}^{-\tilde{Y}_\ell}   \right).
\end{align*}
The triangle inequality gives that 
\begin{align*}
     &  \left\vert \mathrm{cov}\left( \mathrm{e}%
^{ia_{n}^{-1}u\sum_{t\in A}X_{t}-a_{n}^{-2}\lambda \sum_{t\in A}X_{t}^{2}},%
\mathrm{e}^{ia_{n}^{-1}u\sum_{t=\ell }^{\overline{\ell}}X_{t}-a_{n}^{-2}\lambda \sum_{t=\ell }^{\overline{\ell}}X_{t}^{2}}\right) \right\vert \\
& \leqslant \left\vert \mathrm{cov}\left( \cos(Y_0)\mathrm{e}^{-\tilde{Y}_0}, \cos(Y_\ell)\mathrm{e}^{-\tilde{Y}_\ell} \right) \right\vert + \left\vert \mathrm{cov}\left(\sin(Y_0)\mathrm{e}^{-\tilde{Y}_0}, 
\cos(Y_\ell)\mathrm{e}^{-\tilde{Y}_\ell}   \right) \right\vert \\
&  \quad + \left\vert \mathrm{cov}\left(\cos(Y_0)\mathrm{e}^{-\tilde{Y}_0}, \sin(Y_\ell)\mathrm{e}^{-\tilde{Y}_\ell}  \right) \right\vert  + \left\vert \mathrm{cov}\left( \sin(Y_0)\mathrm{e}^{-\tilde{Y}_0}, \sin(Y_\ell)\mathrm{e}^{-\tilde{Y}_\ell}   \right) \right\vert.
\end{align*}
Since  $\cos(Y_0)\mathrm{e}^{-\tilde{Y}_0}$ and $ \cos(Y_\ell)\mathrm{e}^{-\tilde{Y}_\ell}$ are functions of, respectively, $(X_t)_{t\leqslant 0} $ and $(X_t)_{t\geqslant \ell} $, and bounded in absolute value by one, $\alpha_{n,\ell, \overline{\ell} }( u,\lambda ) \leqslant 4 \sup_{|g_1|,|g_2|\leq1}\left|\mathrm{cov}\left[g_1(\{X_t\}_{t\leqslant 0}),g_2(\{X_t\}_{t\geqslant \ell})\right]\right|  = 4\alpha_{\ell}$.
\hfill $\square$
\begin{lemma}\label{lem:product_inequality}
    Let $Z_1,\dots,Z_m$ be identically distributed complex-valued random variables with $\mathbb{P}(|Z_1|\leqslant 1)=1$. Then for $m\geqslant2$,
    \begin{equation*}
        \left\vert \mathbb{E}\left[\prod_{i=1}^m Z_i \right] - (\mathbb{E}[Z_1])^m \right\vert \leqslant \sum_{i=2}^m \left\vert \mathrm{cov}\left( \prod_{j=1}^{i-1}Z_j, Z_i \right)\right\vert.
    \end{equation*}
   
\end{lemma}
\textbf{Proof of Lemma \ref{lem:product_inequality}:}
For complex-valued $a_1,\dots,a_m$ and $b_1,\dots,b_m$, one has the identity
\begin{equation}\label{eq:identity}
\prod_{i=1}^{m}a_{i}-\prod_{i=1}^{m}b_{i}=\sum_{i=1}^{m}%
\prod_{j=1}^{i-1}a_{j}\left( a_{i}-b_{i}\right) \prod_{j=i+1}^{m}b_{j},
\end{equation}
with the convention that $\prod_{j=1}^{0}a_{j} = \prod_{j=m+1}^{m}b_{j} =1$. By stationarity and \eqref{eq:identity},
\begin{align*}
     \left\vert \mathbb{E}\left[\prod_{i=1}^m Z_i \right] - (\mathbb{E}[Z_1])^m \right\vert
    & = \left\vert \mathbb{E}\left[\prod_{i=1}^m Z_i \right] - \prod_{i=1}^m \mathbb{E}[Z_i] \right\vert \\
    &= \left\vert  \sum_{i=1}^{m} \mathbb{E}\left[\prod_{j=1}^{i-1}Z_{j}\left( Z_{i}-\mathbb{E}[Z_{i}]\right) \right] \prod_{j=i+1}^{m}\mathbb{E}[Z_{i}]  \right\vert \\
 &\leqslant   \sum_{i=1}^{m} \left\vert \mathrm{cov}\left(
\prod_{j=1}^{i-1}Z_{j}, Z_{i}\right) \right\vert \left\vert  (\mathbb{E}[Z_{1}])^m  \right\vert \\
 &  \leqslant \sum_{i=1}^{m} \left\vert \mathrm{cov}\left(
\prod_{j=1}^{i-1}Z_{j}, Z_{i}\right) \right\vert   = \sum_{i=2}^{m} \left\vert \mathrm{cov}\left(
\prod_{j=1}^{i-1}Z_{j}, Z_{i}\right) \right\vert,
\end{align*}
where the first inequality follows by the triangle inequality, and the second inequality follows from the fact that $|\mathbb{E}[Z_1]|\leqslant1$.\hfill $\square$

\begin{lemma}\label{lem:large_small_blocks}
    For $n\in \mathbb{N}$, let $a_n,b_n,c_n,d_n$  be real-valued random variables with $\mathbb{P}(c_n\geqslant0)=\mathbb{P}(d_n\geqslant0)=1$ $\forall n$, and let $(k_n)_{n\in\mathbb{N}}$ be an integer-valued deterministic sequence satisfying $k_n\to\infty$ as $n\to\infty$.  Define
    \begin{equation*}
        X_n=\mathrm{e}^{ia_n-c_n},\quad x_n=\mathbb{E}[X_n],\quad Y_n=\mathrm{e}^{ib_n-d_n},\quad y_n=\mathbb{E}[Y_n],\quad z_n=\mathbb{E}[X_nY_n],  
    \end{equation*}
    and assume that $x_n-1=O(1/k_n)$, $y_n-1=o(1/k_n)$, and $z_n -1 = O(1/k_n)$, as $n\to\infty$.
    Then $z_n^{k_n} - x_n^{k_n} = o(1)$ as $n\to\infty$.
\end{lemma}
\textbf{Proof of Lemma \ref{lem:large_small_blocks}:} From the identity in \eqref{eq:identity}, it holds that
\begin{equation*}
     z_n^{k_n} - x_n^{k_n} = (z_n-x_n)\left(x_n^{k_n-1} +\sum_{s=2}^{k_n}z_n^{s-1}x_n^{k_n-s} \right),
\end{equation*}
which gives the bound
\begin{equation}\label{eq:small_large_bound}
     \vert z_n^{k_n} - x_n^{k_n} \vert \leqslant k_n|z_n-x_n|\max(\vert z_n \vert,\vert x_n \vert)^{k_n-1}.
\end{equation}
By assumption, there exists a constant $C>0$ such that for all $n$ sufficiently large $\max(\vert z_n \vert,\vert x_n \vert) \leqslant ( 1+ C/k_n )$,
where we note that $\lim_{n\to\infty}( 1+ C/k_n )^{k_n-1} = \mathrm{e}^{C}$. Consequently, $\max(\vert z_n \vert,\vert x_n \vert)^{k_n-1} = O(1)$, and given the bound in \eqref{eq:small_large_bound}, it remains to show that
\begin{equation}\label{eq:z_minus_x}
    \vert z_n - x_n \vert = o(1/k_n).
\end{equation}
We have that $z_n-x_n = \mathbb{E}[X_n(Y_n-1)] = \mathrm{cov}(X_n,Y_n) + x_n(y_n-1)$, where the second term is $x_n(y_n-1)=o(1/k_n)$ by assumption. Hence, \eqref{eq:z_minus_x} holds provided that
\begin{eqnarray}\label{eq:cov_rate}
    |\mathrm{cov}(X_n,Y_n)|=o(1/k_n).
\end{eqnarray}
By the Cauchy-Schwarz inequality,
\begin{equation}\label{eq:cov_bound}
    |\mathrm{cov}(X_n,Y_n)| \leqslant \sqrt{\mathbb{E}[|X_n-x_n|^2]\mathbb{E}[|Y_n-y_n|^2] }.
\end{equation}
With ${x}^{\ast}$ denoting the complex conjugate of complex-valued $x$, we have that
\begin{equation*}
    |X_n-x_n|^2 = |X_n-1|^2 + |x_n-1|^2 - 2 \Re\left((X_n-1){(x_n-1)}^{\ast}  \right), 
\end{equation*}
such that $\mathbb{E}[|X_n-x_n|^2] = \mathbb{E}[|X_n-1|^2] - |x_n-1|^2 = \mathbb{E}[|X_n-1|^2] + O(1/k_n^2)$, where the second equality follows by assumption. Moreover, using that, by definition, $\mathbb{P}(|X_n|\leqslant 1)=1$,
\begin{equation*}
    \mathbb{E}[|X_n-1|^2] = \mathbb{E}\left[|X_n|^2 + 1 - 2\Re(X_n)\right] \leqslant 2(1 - \Re(x_n)) = O(1/k_n),
\end{equation*}
where the last equality follows by assumption. Hence,
\begin{equation}\label{eq:X_bound}
\mathbb{E}[|X_n-x_n|^2]=O(1/k_n).    
\end{equation}
 Likewise, by identical arguments,
    $\mathbb{E}[|Y_n-y_n|^2] = \mathbb{E}[|Y_n-1|^2] - |y_n-1|^2 = \mathbb{E}[|Y_n-1|^2] + o(1/k_n^2)$, where    $\mathbb{E}[|Y_n-1|^2] = E\left[|Y_n|^2 + 1 - 2\Re(Y_n)\right] \leqslant 2(1 - \Re(y_n)) = o(1/k_n)$, and we have that
\begin{equation}\label{eq:Y_bound}
\mathbb{E}[|Y_n-y_n|^2]=o(1/k_n).    
\end{equation}
Combining the bounds in \eqref{eq:cov_bound}-\eqref{eq:Y_bound}, we conclude that \eqref{eq:cov_rate} holds. \hfill $\square$

\newpage

\setcounter{page}{1}
\begin{center}
\large{Supplemental Appendix to\\
"Heavy Tails and Predictive Ability Testing"}
\end{center}
\begin{center}
by Jonas F. Frederiksen, Muneya Matsui, and Rasmus S. Pedersen
\end{center}

\section{Additional derivations and proofs for Section \ref{sec:limit_theory}}
\subsection{Derivations of \eqref{eq:chr fct of xi} and \eqref{eq: scale and skew}}\label{sec:chf_scale_skew}
Setting $\lambda=0$ in the hybrid characteristic function - Laplace transform in \eqref{eq:hybrid_chf_Laplace} yields the characteristic function of $\xi_\kappa$,
\begin{align}\label{eq:chf_deriv1}
        &\mathbb{E}\big[\mathrm{e}^{iu\xi_\kappa}\big]=\exp\left(\int_0^\infty \mathbb{E}\left[\mathrm{e}^{iyu\sum_{t\in\mathbb{Z}}Q_t}-1-iyu\sum_{t\in\mathbb{Z}}Q_t\boldsymbol{1}_{(1,2)}(\kappa)\right]\,d(-y^{-\kappa})\right),\qquad u\in\mathbb{R},
\end{align}
where $\boldsymbol{1}_{(1,2)}(\kappa):=\boldsymbol{1}(1<\kappa<2)$. Initially, we seek to show that
\begin{equation}\label{eq:integrability}
    \mathbb{E}\left[\int_0^\infty\Big|\mathrm{e}^{iyu\sum_{t\in\mathbb{Z}}Q_t}-1-iyu\sum_{t\in\mathbb{Z}}Q_t\mathbf{1}_{(1,2)}(\kappa)\Big|\,d(-y^{-\kappa})\right]<\infty,
\end{equation}
in order to apply Fubini's theorem to \eqref{eq:chf_deriv1}; that is, to interchange the order of integration and expectation. With $z=y|\sum_{t\in\mathbb{Z}}Q_t|$, a change of variable gives that
\begin{align}\label{eq:change_of_variable}
    \mathbb{E}&\left[\int_0^\infty\Big|\mathrm{e}^{iyu\sum_{t\in\mathbb{Z}}Q_t}-1-iyu\sum_{t\in\mathbb{Z}}Q_t\mathbf{1}_{(1,2)}(\kappa)\Big|\,d(-y^{-\kappa})\right] \notag \\[0.05in]
    &=\mathbb{E}\left[\Big|\sum_{t\in\mathbb{Z}}Q_t\Big|^\kappa\right]\,\left(\int_0^\infty \big|\mathrm{e}^{iuz}-1-iuz\mathbf{1}_{(1,2)}(\kappa)\big|\,d(-z^{-\kappa})\right) \notag \\[0.05in]
    &\leqslant \mathbb{E}\left[\Big(\sum_{t\in\mathbb{Z}}|Q_t|\Big)^\kappa\right]\,\left(\int_0^\infty \big|\mathrm{e}^{iuz}-1-iuz\mathbf{1}_{(1,2)}(\kappa)\big|\,d(-z^{-\kappa})\right),
\end{align}
where we for the equality used that $|\mathrm{e}^{iuz}-1-iuz\mathbf{1}_{(1,2)}(\kappa)|=|\mathrm{e}^{-iuz}-1+iuz\mathbf{1}_{(1,2)}(\kappa)|$. It follows by Lemma 9.1.7 in \cite{mikosch2024extreme} that $\mathbb{E}[{(\sum_{t\in\mathbb{Z}}|Q_t|)}^\kappa]<\infty$ since we assume the anti-clustering condition in Assumption \ref{ass:AC+MX} holds. For the second integral in \eqref{eq:change_of_variable}, it holds by Lemma 8.6 of \cite{sato1999} that for any $n\in\mathbb{N}$
\begin{equation}\label{eq:sato}
    \mathrm{e}^{iuz}=\sum_{k=0}^{n-1}\dfrac{(iuz)^k}{k!}+\theta\dfrac{|uz|^n}{n!},
\end{equation}
where $\theta\in\mathbb{C}$ satisfies $|\theta|\leqslant1$. Consequently, noting that $|\mathrm{e}^{iuz}-1|\leqslant2$ and using \eqref{eq:sato} with $n=1$ yields for the case $\kappa\in(0,1)$ that,
\begin{align*}
    \int_0^\infty \big|\mathrm{e}^{iuz}-1\big|\,d(-z^{-\kappa})&=\int_0^1 \big|\mathrm{e}^{iuz}-1\big|\,d(-z^{-\kappa})+\int_1^\infty \big|\mathrm{e}^{iuz}-1\big|\,d(-z^{-\kappa}) \\[0.05in]
    &\leqslant \int_0^1|u|z\,d(-z^{-\kappa})+\int_1^\infty2\,d(-z^{-\kappa})<\infty.
\end{align*}
Next, for the case where $\kappa\in(1,2)$ we similarly obtain from \eqref{eq:sato} with $n=2$ that,
\begin{align*}
    \int_0^\infty\left|\mathrm{e}^{iuz}-1-iuz\right|\,d(-z^{-\kappa})\leqslant \int_0^1(uz)^2/2\,d(-z^{-\kappa})+\int_1^\infty (2+|u|z)\,d(-z^{-\kappa})<\infty.
\end{align*}
Consequently, the integrability in \eqref{eq:integrability} holds for $\kappa\in(0,1)\cup(1,2)$. An application of Fubini's theorem now yields,
\begin{align}
    &\int_{0}^\infty \mathbb{E}\left[\mathrm{e}^{iyu\sum_{t\in\mathbb{Z}}Q_t}-1-iyu\sum_{t\in\mathbb{Z}}Q_t\mathbf{1}_{(1,2)}(\kappa)\right]\,d(-y^{-\kappa})\notag\\[0.05in]
    &\quad\;\;=\mathbb{E}\left[\int_0^\infty\left(\mathrm{e}^{iyu\sum_{t\in\mathbb{Z}}Q_t}-1-iyu\sum_{t\in\mathbb{Z}}Q_t\mathbf{1}_{(1,2)}(\kappa)\right)\,d(-y^{-\kappa})\right] \notag\\[0.1in]
    &\quad\;\;=\mathbb{E}\left[\int_0^\infty\dfrac{\mathrm{e}^{iyu\sum_{t\in\mathbb{Z}}Q_t}-1-iyu\sum_{t\in\mathbb{Z}}Q_t\mathbf{1}_{(1,2)}(\kappa)}{y^{\kappa+1}}\kappa\; dy\right].\label{eq:chf_deriv2}
\end{align}
By \citet[Lemma 14.11]{sato1999}, 
\begin{align}\label{eq:chf_deriv3}
    \int_0^\infty \dfrac{\mathrm
    e^{iyu\sum_{t\in\mathbb{Z}}Q_t}-1-iyu\sum_{t\in\mathbb{Z}}Q_t\mathbf{1}_{(1,2)}(\kappa)}{y^{\kappa+1}}dy&=-\kappa^{-1}\,\dfrac{\Gamma(2-\kappa)}{1-\kappa}\,\left(-iu\sum_{t\in\mathbb{Z}}Q_t\right)^\kappa.
\end{align}
Next, writing the imaginary unit on its polar form, we obtain
\begin{align*}
    \left(-iu\sum_{t\in\mathbb{Z}}Q_t\right)^\kappa &= \left(-i\,\Big|u\sum_{t\in\mathbb{Z}}Q_t\Big|\,\textrm{sign}\Big(u\sum_{t\in\mathbb{Z}}Q_t\Big)\right)^\kappa=\Big|u\sum_{t\in\mathbb{Z}}Q_t\Big|^\kappa\mathrm{e}^{-i\frac{\kappa\pi}{2}\,\textrm{sign}\big(u\sum_{t\in\mathbb{Z}}Q_t\big)}\\[0.1in]
    &=\left(u\sum_{t\in\mathbb{Z}}Q_t\right)^\kappa_+\mathrm{e}^{-i\frac{\kappa\pi}{2}}+\left(u\sum_{t\in\mathbb{Z}}Q_t\right)^\kappa_-\mathrm{e}^{i\frac{\kappa\pi}{2}}.
\end{align*}
Applications of Euler's formula now yield that
% \begin{align*}
%     \mathrm{e}^{i\frac{\kappa\pi}{2}}=\cos\left(\dfrac{\kappa\pi}{2}\right)+i\sin\left(\dfrac{\kappa\pi}{2}\right),\qquad \mathrm{e}^{-i\frac{\kappa\pi}{2}}=\cos\left(\dfrac{\kappa\pi}{2}\right)-i\sin\left(\dfrac{\kappa\pi}{2}\right), 
% \end{align*}
\begin{align}
    &\left(-iu\sum_{t\in\mathbb{Z}}Q_t\right)^\kappa\notag \\[0.05in] &=\cos\left(\dfrac{\kappa\pi}{2}\right)\left[\Big(u\sum_{t\in\mathbb{Z}}Q_t\Big)^\kappa_++\Big(u\sum_{t\in\mathbb{Z}}Q_t\Big)^\kappa_-\right]-i\sin\left(\dfrac{\kappa\pi}{2}\right)\left[\Big(u\sum_{t\in\mathbb{Z}}Q_t\Big)^\kappa_+-\Big(u\sum_{t\in\mathbb{Z}}Q_t\Big)^\kappa_-\right]  \notag \\[0.05in]
    &=\cos\left(\dfrac{\kappa\pi}{2}\right)|u|^\kappa\Big|\sum_{t\in\mathbb{Z}}Q_t\Big|^\kappa-i\sin\left(\dfrac{\kappa\pi}{2}\right)\left[\Big(u\sum_{t\in\mathbb{Z}}Q_t\Big)^\kappa_+-\Big(u\sum_{t\in\mathbb{Z}}Q_t\Big)^\kappa_-\right] \notag \\[0.05in]
    &=\cos\left(\dfrac{\kappa\pi}{2}\right)|u|^\kappa\Big|\sum_{t\in\mathbb{Z}}Q_t\Big|^\kappa\left(1-i\tan\left(\dfrac{\kappa\pi}{2}\right)\dfrac{\Big(u\sum_{t\in\mathbb{Z}}Q_t\Big)^\kappa_+-\Big(u\sum_{t\in\mathbb{Z}}Q_t\Big)^\kappa_-}{\Big|u\sum_{t\in\mathbb{Z}}Q_t\Big|^\kappa}\right). \label{eq:chf_deriv4}
\end{align}
Next, observe that
\begin{align}\label{eq:chf_deriv5}
    &\dfrac{\Big(u\sum_{t\in\mathbb{Z}}Q_t\Big)^\kappa_+-\Big(u\sum_{t\in\mathbb{Z}}Q_t\Big)^\kappa_-}{\Big|u\sum_{t\in\mathbb{Z}}Q_t\Big|^\kappa}
    %=\textrm{sign}(u)\;\textrm{sign}\Big(\sum_{t\in\mathbb{Z}}Q_t\Big)
    =\textrm{sign}(u)\dfrac{\Big(\sum_{t\in\mathbb{Z}}Q_t\Big)^\kappa_+-\Big(\sum_{t\in\mathbb{Z}}Q_t\Big)^\kappa_-}{\Big|\sum_{t\in\mathbb{Z}}Q_t\Big|^\kappa}.
\end{align}
With $c_\kappa$ given by \eqref{eq:def_c_kappa}, combining \eqref{eq:chf_deriv1} and   \eqref{eq:chf_deriv2}--\eqref{eq:chf_deriv5} gives that
\begin{align*}
    \mathbb{E}[\mathrm{e}^{iu\xi_\kappa}]&=\exp\left(-\mathbb{E}\left[c_\kappa|u|^\kappa\Big|\sum_{t\in\mathbb{Z}}Q_t\Big|^\kappa\left(1-i\tan\big(\dfrac{\kappa\pi}{2}\big)\textrm{sign}(u)\dfrac{\Big(\sum_{t\in\mathbb{Z}}Q_t\Big)^\kappa_+-\Big(\sum_{t\in\mathbb{Z}}Q_t\Big)^\kappa_-}{\Big|\sum_{t\in\mathbb{Z}}Q_t\Big|^\kappa}\right)\right]\right) \notag \\[0.05in]
    &=\exp\left(-|u|^\kappa c_\kappa \mathbb{E}\Big[\big|\sum_{t\in\mathbb{Z}}Q_t\Big|^\kappa\Big]+|u|^\kappa i\tan\left(\dfrac{\kappa\pi}{2}\right)\textrm{sign}(u)c_\kappa\mathbb{E}\left[\Big(\sum_{t\in\mathbb{Z}}Q_t\Big)^\kappa_+-\Big(\sum_{t\in\mathbb{Z}}Q_t\Big)^\kappa_-\right]\right).
\end{align*}
With $\beta$ and $\tilde{\sigma}^\kappa$ given as in \eqref{eq: scale and skew}, we then have that
\begin{align*}
    \mathbb{E}[\mathrm{e}^{iu\xi_\kappa}]&=\exp(-|u|^\kappa \tilde{\sigma}^\kappa+|u|^\kappa\tilde{\sigma}^\kappa i\tan(\kappa\pi/2)\textrm{sign}(u)\beta) \\
    &=\exp(-|u|^\kappa\tilde{\sigma}^\kappa(1-i\beta\textrm{sign}(u)\tan(\kappa\pi/2)),\qquad u\in\mathbb{R},
\end{align*}
and we conclude that \eqref{eq:chr fct of xi} holds.\hfill$\square$

\subsection{Proof of Proposition \ref{Master_Lemma}}
The results for Cases (1) and (2) are immediate consequences of Theorems \ref{theo:rio} and \ref{theo:inf_var_jointCLT}, respectively. It remains to prove the result for case (3). \\
\indent  \textit{Infinite variance case}: Consider first the denominator term in the definition of $T_n$ normalized by $(a_n)_{n\geqslant1}$ satisfying \eqref{eq:def_a_n}. 
We observe that
\begin{align*}
    a_n^{-1}\big(\sum_{t=1}^nX_t^2\big)^{1/2} &= \big(a_n^{-2}\sum_{t=1}^n(X_t-\mathbb{E}[X_t]+\mathbb{E}[X_t])^2\big)^{1/2} \\
    &=\Big(a_n^{-2}\sum_{t=1}^n(X_t-\mathbb{E}[X_t])^2 +\dfrac{2}{a_n^{2}}\mathbb{E}[X_t]\sum_{t=1}^nX_t - \dfrac{n}{a_n^{2}}(\mathbb{E}[X_t])^2\Big)^{1/2} \\
    &=\Big(a_n^{-2}\sum_{t=1}^n(X_t-\mathbb{E}[X_t])^2+ \dfrac{2}{a_n^{2}}\mathbb{E}[X_t]\sum_{t=1}^n(X_t-\mathbb{E}[X_t]) + \dfrac{n}{a_n^{2}}(\mathbb{E}[X_t])^2\Big)^{1/2} \\
    &\xrightarrow{d}\zeta_{\kappa/2},\quad n\rightarrow\infty,
\end{align*}
where we have applied the joint convergence result in Theorem \ref{theo:inf_var_jointCLT} for the centered process, $(X_t-\mathbb{E}[X_t])_{t\in\mathbb{Z}}$ and \eqref{eq:rate_a_n}. Next, for the numerator in the definition of $T_n$ normalized by $a_n$, we have the following result, as $n\to\infty$,
\begin{align*}
    a_n^{-1}\sum_{t=1}^nX_t  &=  a_n^{-1}\sum_{t=1}^n(X_t-\mathbb{E}[X_t]) +\dfrac{n}{a_n}\mathbb{E}[X_t] = O_{\mathbb{P}}(1) + \dfrac{n}{a_n}\mathbb{E}[X_t]  \overset{\mathbb{P}}{\to} \pm\infty
\end{align*}
where we again used Theorem \ref{theo:inf_var_jointCLT} and \eqref{eq:rate_a_n}. Noting that $\zeta_{\kappa/2}>0$,
the result now follows by Slutsky's lemma. \\
\indent \textit{Finite variance case}: By the ergodic theorem, the denominator of $T_n$ satisfies
\begin{align*}
    \dfrac{1}{\sqrt{n}}\Big(\sum_{t=1}^n X_t^2\Big)^{1/2} =\Big(n^{-1}\sum_{t=1}^n X_t^2\Big)^{1/2}\xrightarrow{\mathbb{P}}\sqrt{\mathbb{E}[X_t^2]},\qquad n\to\infty. 
\end{align*}
Applying Theorem \ref{theo:rio} to the centered process $(X_t - \mathbb{E}[X_t])_{t\in\mathbb{Z}}$, the $\sqrt{n}$-normalized numerator in $T_n$ satisfies
\begin{align*}
   \dfrac{1}{\sqrt{n}}\sum_{t=1}^n X_t  &=  \dfrac{1}{\sqrt{n}}\sum_{t=1}^n (X_t-\mathbb{E}[X_t]) + \dfrac{n}{\sqrt{n}}\mathbb{E}[X_t] =  O_{\mathbb{P}}(1) + \dfrac{n}{\sqrt{n}}\mathbb{E}[X_t] \overset{\mathbb{P}}{\to} \pm\infty, \qquad n\to\infty.
\end{align*}
\hfill $\square$

\subsection{Proofs related to Section \ref{sec:AR_details}}
\textbf{Proof of Proposition \ref{prop:AR1 - RV}:}  Set $\widetilde{X}_t=X_t-\delta/(1-\varphi)$ for all $t\in\mathbb{Z}$. Then it holds that $(\widetilde{X}_t)_{t\in\mathbb{Z}}$ is a stationary and $\alpha$-mixing (with geometric rate) AR(1) process with no constant term, in particular, it follows the recursion $\widetilde{X}_t=\varphi\,\widetilde{X}_{t-1}+Z_t,\;t\in\mathbb{Z}$. Furthermore, it follows from p.\,192 in \cite{mikosch2024extreme} that $(\widetilde{X}_t)_{t\in\mathbb{Z}}$ is regularly varying with tail index $\kappa$ and has spectral tail process $(\mathit{\Theta}_t)_{t\in\mathbb{Z}}$, which is given by,
\begin{align}
    \mathit{\Theta}_t=\mathit{\Theta}_0\,\varphi^t\mathbf{1}(t\geqslant -J),\qquad t\in\mathbb{Z}
\end{align}
with $\mathbb{P}(J=j)=\varphi^{\kappa j}\,(1-\varphi^\kappa)$, $j=0,1,\dots$, and $\mathbb{P}(\mathit{\Theta}_0=\pm1)=p_\pm$, where $p_\pm$ are the tail-balance coefficients of $Z_t$. Thus by \citet[Theorem 2.1 and Corollary 3.2]{basrak2009regularly}, we have for every $h\geq 0$,
\begin{equation*}
    \mathbb{P}(x^{-1}(\widetilde{X}_0,\dots,\widetilde{X}_h)\in\cdot\;|\;|\widetilde{X}_0|>x) \xrightarrow{w}\mathbb{P}(Y_\kappa\,(\mathit{\Theta}_0,\dots,\mathit{\Theta}_h)\in\cdot),\qquad x\to\infty,
\end{equation*}
where $Y_\kappa$ is a Pareto-distributed random variable satisfying $\mathbb{P}(Y_\kappa\leqslant y)=1-y^{-\kappa}$ for $y\geqslant1 $.
Next, using that $X_0=\widetilde{X}_0+\delta/(1-\varphi)$, we observe that,
\begin{align*}
\dfrac{\mathbb{P}(|X_0|>x)}{\mathbb{P}(|\widetilde{X}_0|>x)}&=\dfrac{\mathbb{P}(\widetilde{X}_0>x-\delta/(1-\varphi))}{\mathbb{P}(|\widetilde{X}_0|>x)}+\dfrac{\mathbb{P}(\widetilde{X}_0<-x-\delta/(1-\varphi))}{\mathbb{P}(|\widetilde{X}_0|>x)} \\[0.05in]
 &\sim p_+\left(\dfrac{x-\delta/(1-\varphi)}{x}\right)^{-\kappa}+p_-\left(\dfrac{x+\delta/(1-\varphi)}{x}\right)^{-\kappa}\rightarrow1,\qquad x\to\infty.
\end{align*}
That is, $|X_t|$ is regularly varying of index $\kappa$ with $\mathbb{P}(|\widetilde{X}_t|>x)\sim\mathbb{P}(|X_t|>x)$ as $x\to\infty$.
By similar arguments, we have, as $x\to\infty$,

\begin{align*}
    &\mathbb{P}(|X_0|>x\,|\;|\widetilde{X}_0|>x)\\[0.05in]
    &\quad= \dfrac{\mathbb{P}(|X_0|>x,|\widetilde{X}_0|>x)}{\mathbb{P}(|\widetilde{X}_0|>x)}\\
    &\quad \sim \dfrac{\mathbb{P}(\widetilde{X}_0+\delta/({1-\varphi})<-x,\widetilde{X}_0<-x)}{\mathbb{P}(|\widetilde{X}_0|>x)}+\dfrac{\mathbb{P}(\widetilde{X}_0+\delta/({1-\varphi})>x,\widetilde{X}_0>x)}{\mathbb{P}(|\widetilde{X}_0|>x)}\\
    &\quad=\dfrac{\mathbb{P}(\widetilde{X}_0<-x-\max\{\delta/(1-\varphi),0\})}{\mathbb{P}(|\widetilde{X}_0|>x)}+\dfrac{\mathbb{P}(\widetilde{X}_0>x-\min\{\delta/(1-\varphi),0\})}{\mathbb{P}(|\widetilde{X}_0|>x)} \\
    &\quad\sim p_-\left(\dfrac{x+\max\{\delta/(1-\varphi),0\}}{x}\right)^{-\kappa}+p_+\left(\dfrac{x-\min\{\delta/(1-\varphi),0\}}{x}\right)^{-\kappa}\rightarrow1,
\end{align*}
where we have used that 
$\mathbb{P}(X_0> \pm x,\,\widetilde X_0< \mp x) =o(1)$ as $x\to\infty$.
Consequently, it follows that
\begin{align*}
    \dfrac{\mathbb{P}(|X_0|>x,|\widetilde{X}_0|\leq x)}{\mathbb{P}(|\widetilde{X}_0|>x)}&=\dfrac{\mathbb{P}(|X_0|>x)}{\mathbb{P}(|\widetilde{X}_0|>x)}-\dfrac{\mathbb{P}(|X_0|>x,|\widetilde{X}_0|>x)}{\mathbb{P}(|\widetilde{X}_0|>x)}\\
    &=\dfrac{\mathbb{P}(|X_0|>x)}{\mathbb{P}(|\widetilde{X}_0|>x)}-\mathbb{P}(|X_0|>x\;|\;|\widetilde{X}_0|>x)\to0,\qquad x\to\infty.
\end{align*}
Using the above relations, we obtain for every $h\geq0$, as $x\to\infty$,
\begin{align*}
    \mathbb{P}&(x^{-1}(X_0,\dots,X_h)\in\cdot\;|\;|X_0|>x)\\
    &=\dfrac{\mathbb{P}(x^{-1}(X_0,\dots,X_h)\in\cdot\;,\;|X_0|>x)}{\mathbb{P}(|X_0|>x)} \\
    &\sim\dfrac{\mathbb{P}(x^{-1}(X_0,\dots,X_h)\in\cdot\;,\;|X_0|>x)}{\mathbb{P}(|\widetilde{X}_0|>x)} \\
    &=\dfrac{\mathbb{P}(x^{-1}(X_0,\dots,X_h)\in\cdot\;,\;|X_0|>x,|\widetilde{X}_0|>x)}{\mathbb{P}(|\widetilde{X}_0|>x)} +\dfrac{\mathbb{P}(x^{-1}(X_0,\dots,X_h)\in\cdot\;,\;|X_0|>x,|\widetilde{X}_0|\leq x)}{\mathbb{P}(|\widetilde{X}_0|>x)} \\
    &\sim\mathbb{P}(x^{-1}(X_0,\dots,X_h)\in\cdot\;|\;|\widetilde{X}_0|>x).
\end{align*}
Finally, inserting $X_t=\widetilde{X}_t+\delta/(1-\varphi)$ and using that $(\widetilde{X}_t)_{t\in\mathbb{Z}}$ is regularly varying of index $\kappa$ and has spectral tail process $(\mathit{\Theta}_t)_{t\in\mathbb{Z}}$,
\begin{align*}
    \mathbb{P}&(x^{-1}(X_0,\dots,X_h)\in\cdot\;|\;|X_0|>x)\\
    &\sim\mathbb{P}(x^{-1}\delta/(1-\varphi)+x^{-1}(\widetilde{X}_0,\dots,\widetilde{X}_h)\in\cdot\;|\;|\widetilde{X}_0|>x) \\
    &\sim\mathbb{P}(x^{-1}(\widetilde{X}_0,\dots,\widetilde{X}_h)\in\cdot\;|\;|\widetilde{X}_0|>x)  \xrightarrow{w}\mathbb{P}(Y_\kappa\,(\mathit{\Theta}_0,\dots,\mathit{\Theta}_h)\in\cdot),\qquad x\to\infty.
\end{align*}
Hence, $(X_t)_{t\in\mathbb{Z}}$ is regularly varying of index $\kappa$ with the same spectral tail process as $(\widetilde{X}_t)_{t\in\mathbb{Z}}$.
Next, by \citet[Proposition 10.1.14]{mikosch2024extreme}, it holds that $(X_t)_{t\in\mathbb{Z}}$ satisfies the anti-clustering condition \ref{AC} for any $(r_n)_{n\geq1}$ chosen such that $r_n=o(a_n^2/n)$ for every $\kappa\in(0,2)$. Furthermore, the mixing condition \eqref{alpha_coefs_req} holds for any sequence $(r_n)_{n\geq1}$ chosen such that $\log n=o(r_n)$ as $n\to\infty$. Indeed, choosing $\ell_n=d\log n$ with $d>0$, and using that $(X_t)_{t\in\mathbb{Z}}$ is $\alpha$-mixing with geometric rate,
\begin{equation*}
    \dfrac{n}{r_n}\alpha_{\ell_n}\leq c\dfrac{n}{r_n}\rho^{d\log n}=c\dfrac{n}{r_n}n^{d\log \rho}=c\dfrac{\log n}{r_n}\;\dfrac{n^{1-d|\log \rho|}}{\log n},
\end{equation*}
for some $c>0$ and $\rho\in(0,1)$. Hence, with $(r_n)_{n\geq1}$ and $d$ such that $\log n=o(r_n)$ and $d>1/|\log\rho|$, the right-hand side tends to zero as $n\to\infty$. Therefore, we conclude that $(X_t)_{t\in\mathbb{Z}}$ given by \eqref{eq:AR1_sec2} satisfies Assumption \ref{ass:AC+MX} for every $\kappa\in(0,2)$.\hfill $\square$ \bigskip

\noindent \textbf{Proof of Proposition \ref{prop:AR_lim_moments}:} For $\kappa \in (0,1)$, \eqref{eq:weak_AR1} follows by Proposition \ref{prop:AR1 - RV}, Theorem \ref{theo:inf_var_jointCLT} and the continuous mapping theorem. Likewise, for $\kappa\in(1,2)$ and $\delta =0$, \eqref{eq:weak_AR1} holds by Propositions \ref{prop:AR1 - RV} and  \ref{Master_Lemma}.(2).

That $\mathbb{E}[(\xi_\kappa/\zeta_{\kappa/2})^p]<\infty$ for every $p>0$ and every $\kappa\in(0,1)\cup(1,2)$ follows directly by \citet[Corollary 4.2]{matsui2025moments}. The same corollary implies that $\xi_\kappa/\zeta_{\kappa/2}$ has a continuous density if $\kappa\in(1,2)$.  With $(Q_t)_{t\in\mathbb{Z}}$ given by \eqref{eq:extremal_cluster_process}, let
\begin{equation}
    \Vert \mathit{Q} \Vert_p := \left( \sum_{t\in\mathbb{Z}}|Q_t|^p \right)^{1/p},\qquad p>0.
\end{equation}
Then by \citet[Remark 5.3]{matsui2025moments}, for $\kappa\in(1,2)$, 
\begin{equation*}
    \mathbb{E}\left[\dfrac{\xi_\kappa}{\zeta_{\kappa/2}}\right]=\mathbb{E}\left[W^{(2)}\right]\dfrac{\Gamma((1-\kappa)/2)}{\sqrt{\pi}\;\Gamma(1-\kappa/2)},
\end{equation*}
and
\begin{equation*}
    \mathbb{E}\left[\Big(\dfrac{\xi_\kappa}{\zeta_{\kappa/2}}\Big)^2\right]=\mathbb{E}[(W^{(2)})^2]+(\mathbb{E}[W^{(2)}])^2\dfrac{\kappa}{2}\left(\dfrac{\Gamma((1-\kappa)/2)}{\Gamma(1-\kappa/2)}\right)^2,
\end{equation*}
where
\begin{equation*}
    W^{(2)}=\dfrac{\Vert Q\Vert_2^\kappa}{\mathbb{E}[\Vert Q\Vert_2^\kappa]}\dfrac{\sum_{t\in\mathbb{Z}}Q_t}{\Vert Q\Vert_2}.
\end{equation*}
Next, we observe that $||Q_t||^\kappa_2$ is deterministic in the AR(1) case from
\begin{align*}
    ||Q||^\kappa_2&=\dfrac{||\mathit{\Theta}||_2^\kappa}{||\mathit{\Theta}||_\kappa^\kappa}=\dfrac{\left(\sum_{t\in\mathbb{Z}}(\mathit{\Theta}_0\;\varphi^t\mathbf{1}(t\geq -J))^2\right)^{\kappa/2}}{\sum_{t\in\mathbb{Z}}|\mathit{\Theta}_0\;\varphi^{ t}\mathbf{1}(t\geq -J)|^\kappa}=\dfrac{(\sum_{t=-J}^\infty \varphi^{2t})^{\kappa/2}}{\sum_{t=-J}^\infty\varphi^{\kappa t}} \\[0.1in]
    &=\dfrac{(\varphi^{-2J}\sum_{t=-J}^\infty\varphi^{2(t+J)})^{\kappa/2}}{\varphi^{-\kappa J}\sum_{t=-J}^\infty\varphi^{\kappa(t+J)}}=\dfrac{(\sum_{i=0}^\infty\varphi^{2i})^{\kappa/2}}{\sum_{i=0}^\infty\varphi^{\kappa i}},
\end{align*}
where we inserted $(\mathit{\Theta}_t)_{t\in\mathbb{Z}}$ given in \eqref{AR_spec_process}, and used $i=t+J$. Consequently, $W^{(2)}$ reduces to
\begin{align*}
    W^{(2)}=\dfrac{\sum_{t\in\mathbb{Z}}Q_t}{||Q_t||_2} =\left(\dfrac{\sum_{t\in\mathbb{Z}}\mathit{\Theta}_t}{(\sum_{t\in\mathbb{Z}}|\mathit{\Theta}_t|^\kappa)^{1/\kappa}}\right)\left(\dfrac{\sum_{t\in\mathbb{Z}}\mathit{\Theta}_t^2}{(\sum_{t\in\mathbb{Z}}|\mathit{\Theta}_t|^\kappa)^{2/\kappa}}\right)^{-1/2}=\;\left(\dfrac{\sum_{t\in\mathbb{Z}}\mathit{\Theta}_t}{(\sum_{t\in\mathbb{Z}}\mathit{\Theta}_t^2)^{1/2}}\right).
\end{align*}
Inserting the spectral tail process of the regularly varying AR(1) process now yields that
\begin{align*}
    W^{(2)}&=\dfrac{\sum_{t\in\mathbb{Z}}\mathit{\Theta}_0\;\varphi^t\mathbf{1}(t\geqslant-J)}{(\sum_{t\in\mathbb{Z}}\;\varphi^{2t}\mathbf{1}(t\geqslant-J))^{1/2}} =\mathit{\Theta}_0\dfrac{\sum_{t=-J}^\infty\varphi^t}{(\sum_{t=-J}^\infty\varphi^{2t})^{1/2}}=\mathit{\Theta}_0\dfrac{\varphi^{-J}\sum_{t=-J}^\infty\varphi^{t+J}}{(\varphi^{-2J}\sum_{t=-J}^\infty\varphi^{2(t+J)})^{1/2}}.
\end{align*}
Finally setting $i=t+J$,
\begin{align*}
    W^{(2)}=\mathit{\Theta}_0\dfrac{\sum_{i=0}^\infty\varphi^{i}}{(\sum_{i=0}^\infty\varphi^{2i})^{1/2}}=\mathit{\Theta}_0\dfrac{\sqrt{1-\varphi^2}}{1-\varphi}=\mathit{\Theta}_0\dfrac{\sqrt{1-\varphi}\sqrt{1+\varphi}}{(\sqrt{1-\varphi})^2}=\mathit{\Theta}_0\dfrac{\sqrt{1+\varphi}}{\sqrt{1-\varphi}}.
\end{align*}
Inserting this relation into the above expressions for the moments and using that $\mathbb{E}[\mathit{\Theta}_0]=p_+-p_-$ yields \eqref{eq:limit_mean}--\eqref{eq:limit_second_moment}. That the skewness parameter of $\xi_\kappa$ is given by $\beta=p_+-p_-$ follows by inserting $\mathit{\Theta}_t=\mathit{\Theta}_0\;\varphi^t$, $t\geq0$, into \eqref{eq:beta forward} of Lemma \ref{lem:sigma_beta}.\hfill $\square$

\begin{lemma}\label{lem:sigma_beta}
    Under the assumptions of Theorem \ref{theo:inf_var_jointCLT}, the skewness and scale parameters in \eqref{eq: scale and skew} can be represented in terms of the forward spectral tail process $(\mathit{\Theta}_t)_{t\geq0}$: 
    \begin{align}
    \beta&=\dfrac{\mathbb{E}[(\sum_{t=0}^\infty\mathit{\Theta}_t)_+^\kappa-(\sum_{t=1}^\infty\mathit{\Theta}_t)_+^\kappa]-\mathbb{E}[(\sum_{t=0}^\infty\mathit{\Theta}_t)_-^\kappa-(\sum_{t=1}^\infty\mathit{\Theta}_t)_-^\kappa]}{\mathbb{E}[|\sum_{t=0}^\infty\mathit{\Theta}_t|^\kappa-|\sum_{t=1}^\infty\mathit{\Theta}_t|^\kappa]}, \label{eq:beta forward}\\
    \tilde{\sigma}^\kappa&=c_\kappa\mathbb{E}\left[\Big|\sum_{t=0}^\infty\mathit{\Theta_t}\Big|^\kappa-\Big|\sum_{t=1}^\infty\mathit{\Theta_t}\Big|^\kappa\right], \label{eq:sigma forward}
\end{align}
with $c_\kappa$ given by \eqref{eq:def_c_kappa}.
\end{lemma}

\noindent \textbf{Proof of Lemma \ref{lem:sigma_beta}:} The expressions for $\beta$ and $\tilde{\sigma}^\kappa$ are stated in \citet[p.458]{mikosch2024extreme}. Detailed arguments are available upon request.  \hfill $\square$ 

\noindent \textbf{Proof of Proposition \ref{prop:AR_HAC}:} Using the autoregressive structure of $X_t$, we note that 
\begin{align*}
    \sum_{t=1}^{n-j}w_jX_tX_{t+j}&=\sum_{t=1}^{n-j}w_jX_t\big(\varphi^jX_t+\varphi^{j-1}Z_{t+1}+\cdots+Z_{t+j}\big) \\
    &=w_j\varphi^j\sum_{t=1}^{n-j}X_t^2+w_j\sum_{t=1}^{n-j}X_t\Big(\sum_{i=1}^j\varphi ^{j-i}Z_{t+i}\Big).
\end{align*}
Consequently, we have that $\hat{\sigma}_n^2$ in \eqref{eq:T_HAC} is given by
\begin{align*}
    \hat{\sigma}_n^2&=n^{-1}\sum_{t=1}^nX_t^2+2n^{-1}\sum_{j=1}^qw_j\varphi^j\sum_{t=1}^{n-j}X_t^2+2n^{-1}\sum_{j=1}^qw_j\sum_{t=1}^{n-j}X_t\Big(\sum_{i=1}^{j}\varphi^{j-i}Z_{t+i}\Big) \\
    &=n^{-1}\sum_{t=1}^{n}X_t^2\Big(1+2\sum_{j=1}^qw_j\varphi^j\Big) \\
    &\quad -2n^{-1}\sum_{j=1}^qw_j \varphi^j\Big(\sum_{t=n-j+1}^nX_t^2\Big)+2n^{-1}\sum_{j=1}^qw_j\sum_{t=1}^{n-j}X_t\Big(\sum_{i=1}^{j}\varphi^{j-i}Z_{t+i}\Big) \\
    &=:n^{-1}\gamma_n^2c_q+R_1+R_2,
\end{align*}
with $c_q:=1+2\sum_{j=1}^qw_j\varphi^j>0$. It immediately follows that
\begin{align*}
    T_n^{\textrm{HAC}}&=\dfrac{\gamma_n}{n^{1/2}\hat{\sigma}_n}\dfrac{S_n}{\gamma_n}= \left(c_q+\dfrac{(R_1+R_2)(n/a_n^2)}{a_n^{-2}\gamma_n^2}\right)^{-1/2}\;T_n,
\end{align*}
where by Proposition \ref{prop:AR_lim_moments}, $T_n\overset{d}{\to} \xi_\kappa/\zeta_{\kappa/2}$. By Theorem \ref{theo:inf_var_jointCLT}, $a_n^{-2}\gamma_n^2$ has a limit that is strictly positive almost surely, and it suffices to show that $(R_1+R_2)(n/a_n^2)=o_\mathbb{P}(1)$, where we note that $n/a_n^2\to0$. Clearly $R_1=o_\mathbb{P}(1)$ as $n\to\infty$. If $\kappa\in(1,2)$ then $R_2$ consists of $q$ averages of zero mean ergodic summands, and, hence, by the ergodic theorem $R_2=o_\mathbb{P}(1)$ as $n\to\infty$. It remains to show that $R_2(n/a_n^{2})=o_\mathbb{P}(1)$ for $\kappa\in(0,1)$. Observe that for any $\varepsilon>0$\medskip
\begin{align*}
    \mathbb{P}\left(a_n^{-2}\sum_{t=1}^{n}X_tZ_{t+j}>\varepsilon\right)&\leqslant \mathbb{P}\left(\sum_{t=1}^{n}|X_t|\;|Z_{t+j}|>\varepsilon a_n^2\right) \leqslant \mathbb{P}\left(n\;\max_{t=1,\dots,n}|X_t|\;|Z_{t+j}|>\varepsilon a_n^2\right).
\end{align*}
Since $|X_t|$ and $|Z_{t+j}|$ are independent and both are regularly varying with index $\kappa$, an application of Breiman's lemma yields that their product is also regularly varying with index $\kappa<1$ as well. Consequently, by regular variation of $|X_t|\;|Z_{t+j}|$ and the fact that $a_n/n\to\infty$ when $\kappa\in(0,1)$
\begin{align*}
    \mathbb{P}\left(a_n^{-2}\sum_{t=1}^{n}X_tZ_{t+j}>\varepsilon\right)&\leq \mathbb{P}\left(\bigcup_{t=1}^n\{|X_t|\;|Z_{t+j}|>\varepsilon\dfrac{a_n^2}{n}\}\right)\leq n\mathbb{P}\left(|X_t|\;|Z_{t+j}|>\varepsilon\dfrac{a_n^2} {n}\right) \\[0.05in]
    &\sim\varepsilon^{-\kappa}n^{1+\kappa}a_n^{-2\kappa}L(n^{-1}a_n^2)\to 0,\qquad n\to\infty.
\end{align*}
\hfill $\square$

\section{Results related to Section \ref{sec:subsampling}}

\subsection{Auxiliary results for subsampling}

Throughout this section, we let $(X_t)_{t\in\mathbb{Z}}$ denote a regularly varying stationary process with tail index $\kappa\in(1,2)$. Furthermore, let $(a_n)_{n\geqslant1}$ be a sequence of reals chosen such that $n\mathbb{P}(|X_t|>a_n)\to1$ as $n\to\infty$. Define
\begin{equation}\label{stats}
    \gamma_{n,2}^2=a_n^{-2}\sum_{t=1}^n(X_t-\mu)^2\quad\quad\text{and}\quad\quad R_n=a_n^{-1}\dfrac{\sum_{t=1}^n(X_t-\mu)}{\gamma_{n,2}},\quad\quad n\geq1,
\end{equation}
with $\mu:=\mathbb{E}[X_t]$. Assuming that the stationary process $(X_t)_{t\in\mathbb{Z}}$ satisfies Assumption \ref{ass:AC+MX}, it holds by Theorem \ref{theo:inf_var_jointCLT} that
\begin{equation}\label{eq:limits}
    \gamma_{n,2}\xrightarrow{d}\zeta_{\kappa/2}\quad\quad\text{and}\quad\quad R_n\xrightarrow{d}\xi_{\kappa}/\zeta_{\kappa/2},\quad\quad n\to\infty,
\end{equation}
respectively, where $\zeta_{\kappa/2}$ and $\xi_{\kappa}$ are defined as in the theorem.\medskip

In this section we prove the validity of the subsampling Algorithm \ref{algo}. We begin with two auxiliary lemmas -- Lemmas \ref{Sub_Lemma 1} and \ref{Sub_lemma2} -- showing that subsampling techniques can approximate the limiting distributions in \eqref{eq:limits}.
The proofs of these lemmas follow the same structure as \citet[proof of Theorem 2.2.1]{politis2012subsampling} adapted to strongly mixing,  regularly varying processes with tail index $\kappa\in(1,2)$. For a block length $b_n \geqslant 1$, the subsample versions of \eqref{stats} are
\begin{equation}\label{eq:centered_subsamples}
    \tilde{\gamma}_{i,b_n}^2=a_{b_n}^{-2}\sum_{t=i}^{i+b_n-1}(X_t-\overline{X}_n)^2\quad\text{and}\quad \tilde{R}_{i,b_n}=\dfrac{a_{b_n}^{-1}\sum_{t=i}^{i+b_n-1}(X_t-\overline{X}_n)}{\;\tilde{\gamma}_{i,b_n}},\quad i=1,\dots,q_n,
\end{equation}
where $q_n=n-b_n+1$, $\overline{X}_n =n^{-1}\sum_{t=1}^nX_t $, and $(a_{b_n})_{n\geqslant1}$ satisfies $b_n\,\mathbb{P}(|X_t|>a_{b_n})\to1$ as $b_n\to\infty$. We then show (in Lemmas \ref{Sub_Lemma 1} and \ref{Sub_lemma2}) that the limiting distributions of $\gamma_{n,2}$ and $R_n$, respectively, can be approximated by
\begin{equation}\label{approximations}
    L_{\gamma,b_n}(x)=q_n^{-1}\sum_{i=1}^{q_n}\mathbf{1}(\tilde{\gamma}_{i,b_n}\leqslant x)\quad\text{and}\quad L_{R,b_n}(x)=q_n^{-1}\sum_{i=1}^{q_n}\mathbf{1}(\tilde{R}_{i,b_n}\leqslant x).
\end{equation}
The validity of Algorithm \ref{algo} under the null hypothesis follows by replacing $\overline{X}_n$ and $\mu$ by zeros, whereas the validity of the Algorithm under the alternative hypothesis is demonstrated in the proof of Proposition \ref{prop:validity_Null_and_alt}).
\begin{lemma}\label{Sub_Lemma 1}
    Let $(X_t)_{t\in\mathbb{Z}}$ satisfy Assumptions \ref{ass_RV}--\ref{ass:AC+MX} with tail index $\kappa\in(1,2)$. Furthermore, let $L_{\gamma,b_n}(x)$ be defined as in \eqref{approximations} with $(b_n)_{n\geqslant1}$ chosen such that $b_n\to\infty$ and $b_n/n\to0$ as $n\to\infty$. Then the following holds.
    \begin{itemize}
        \item [\emph{(1)}] For every $x\in\mathbb{R}$, it holds that $L_{\gamma,b_n}(x)\xrightarrow{\mathbb{P}}\mathbb{P}(\zeta_{\kappa/2}\leqslant x)$ as $n\to\infty$,
        \item [\emph{(2)}] For $\beta\in(0,1)$, define the corresponding quantiles $c_{b_n}(1-\beta)=\inf\{x\in\mathbb{R}:L_{\gamma,b_n}(x)\geqslant1-\beta\}$. It holds that $\mathbb{P}(\gamma_{n,2}\leq c_{b_n}(1-\beta))\rightarrow1-\beta$, as $n\to\infty.$ 
    \end{itemize}
\end{lemma}
\noindent \textbf{Proof of Lemma \ref{Sub_Lemma 1}:}
\textit{Proof of (1):} Note initially from its definition that $\tilde{\gamma}_{i,b_n}$ is a.s.\;\,strictly positive. Consequently, it holds that $L_{\gamma,b_n}(x)=\mathbb{P}(\zeta_{\kappa/2}\leqslant x)=0$ a.s.\;\,for every $x\leqslant0$. It remains to show that the conclusion of (1) holds for every $x> 0$. For this, observe that
\begin{align}\label{eq:sub_gamma_eq1}
    \tilde{\gamma}_{i,b_n}^2&=a_{b_n}^{-2}\sum_{t=i}^{i+b_n-1}(X_t-\overline{X}_n)^2 \notag \\
    &=a_{b_n}^{-2}\left[\sum_{t=i}^{i+b_n-1}(X_t-\mu)^2+b_n(\overline{X}_n-\mu)^2-2(\overline{X}_n-\mu)\sum_{t=i}^{i+b_n-1}(X_t-\mu)\right],
\end{align}
and introduce the following notation: $\gamma_{i,b_n}^2=a_{b_n}^{-2}\sum_{t=i}^{i+b_n-1}(X_t-\mu)^2$, $S_{i,b_n}=\sum_{t=i}^{i+b_n-1}(X_t-\mu)$ and $E_{i,\gamma}=2a_{b_n}^{-2}(\overline{X}_n-\mu)S_{i,b_n}-b_n\;a_{b_n}^{-2}(\overline{X}_n-\mu)^2$. Now, inserting \eqref{eq:sub_gamma_eq1} into $L_{\gamma,b_n}(x)$ in \eqref{approximations}, and using that the events $\{\tilde{\gamma}_{i,b_n}\leqslant x\}$ and $\{\tilde{\gamma}^2_{i,b_n}\leqslant x^2\}$ are equivalent for every $x>0$ yields,
\begin{align*}
    L_{\gamma,b_n}(x)=q_n^{-1}\sum_{i=1}^{q_n}\mathbf{1}\left(\gamma_{i,b_n}^2\leqslant x^2+E_{i,\gamma}\right),\qquad x>0.
\end{align*}
We seek to establish upper and lower bounds for $L_{\gamma,b_n}(x)$ for every $x>0$. To obtain an upper bound, observe that we have the following set relation for every $d>0$,
\begin{equation*}
    \big\{\gamma_{i,b_n}^2\leqslant x^2+E_{i,\gamma}\big\}\subset\big\{\gamma_{i,b_n}^2\leqslant x^2+d\big\}\cup\big\{E_{i,\gamma}>d\big\}.
\end{equation*}
Using this relation and the fact that $(x^2+d)^{1/2}\leqslant x+d^{1/2}$ for all $x,d>0$, we obtain
\begin{equation}\label{Upper 1}
    L_{\gamma,b_n}(x)\leqslant q_n^{-1}\sum_{i=1}^{q_n}\mathbf{1}\big(\gamma_{i,b_n}\leqslant x+d^{1/2}\big)+q_n^{-1}\sum_{i=1}^{q_n}\mathbf{1}\big(E_{i,\gamma}>d\big).
\end{equation}
On the other hand, to obtain a lower bound for $L_{\gamma,b_n}(x)$ for every $x>0$, we initially consider the following set relation
\begin{equation*}
    \big\{\gamma_{i,b_n}^2\leqslant x^2+E_{i,\gamma}\big\}\supset\big\{\gamma_{i,b_n}^2\leqslant x^2-d\big\}\cap\big\{E_{i,\gamma}\geqslant -d\big\}.
\end{equation*}
Next, observe that the function $x\mapsto(x^2-d)^{1/2}$ is well-defined for every $x\in(d^{1/2},\infty)$ and furthermore it holds that $(x^2-d)^{1/2}\geqslant x-d^{1/2}$ on this interval. Applying these observations, we have that
\begin{align*}
     \mathbf{1}\big(\gamma_{i,b_n}^2\leqslant x^2+E_{i,\gamma}\big)&\geqslant\mathbf{1}\big(\gamma_{i,b_n}^2\leqslant x^2-d\big)-\mathbf{1}(E_{i,\gamma}< -d) \\
     &\geqslant \mathbf{1}(\gamma_{i,b_n}\leqslant x-d^{1/2})-\mathbf{1}(E_{i,\gamma}< -d),\quad\quad d>0,\;x\in[d^{1/2},\infty).
\end{align*}
Since we can choose $d>0$ arbitrarily small, we thus obtain the following lower bound for $L_{\gamma,b_n}(x)$, which holds for every $x>0$,
\begin{equation}\label{Lower 1}
    L_{\gamma,b_n}(x)\geqslant q_n^{-1}\sum_{i=1}^{q_n}\mathbf{1}\big(\gamma_{i,b_n}\leqslant x-d^{1/2}\big)-q_n^{-1}\sum_{i=1}^{q_n}\mathbf{1}\big(E_{i,\gamma}< -d\big).
\end{equation}
Next, we seek to show that the term involving $E_{i,\gamma}$ is negligible in both relations \eqref{Upper 1} and \eqref{Lower 1}. Since $d>0$ in both relations, it suffices to show that $q_n^{-1}\sum_{i=1}^{q_n}\mathbf{1}(|E_{i,\gamma}|>d)=o_\mathbb{P}(1)$ as $n\to\infty$. By the triangle inequality we have
\begin{equation*}
    |E_{i,\gamma}|=|2a_{b_n}^{-2}(\overline{X}_n-\mu)S_{i,b_n}-b_n\;a_{b_n}^{-2}(\overline{X}_n-\mu)^2|\leqslant2a_{b_n}^{-2}|\overline{X}_n-\mu|\;|S_{i,b_n}|+b_n\;a_{b_n}^{-2}(\overline{X}_n-\mu)^2.
\end{equation*}
Therefore, it follows that we have for every $\delta>0$,
\begin{align*}
    q_n^{-1}\sum_{i=1}^{q_n}\mathbf{1}\left(|E_{i,\gamma}|>d\right)&\leqslant q_n^{-1}\sum_{i=1}^{q_n}\mathbf{1}\left(2a_{b_n}^{-2}|\overline{X}_n-\mu|\;|S_{i,b_n}|+b_n\;a_{b_n}^{-2}(\overline{X}_n-\mu)^2>d\right)\\
    &=q_n^{-1}\sum_{i=1}^{q_n}\mathbf{1}\left(|a_{b_n}^{-1}S_{i,b_n}|>\dfrac{a_{b_n}\;d}{2|\overline{X}_n-\mu|}-\dfrac{b_n|\overline{X}_n-\mu|}{2a_{b_n}}\right) \\
    &\leqslant q_n^{-1}\sum_{i=1}^{q_n}\mathbf{1}\left(|a_{b_n}^{-1}S_{i,b_n}|>\dfrac{d\,b_n/4-\delta^2}{\delta}\right)+q_n^{-1}\sum_{i=1}^{q_n}\mathbf{1}\left(\dfrac{b_n|\overline{X}_n-\mu|}{2a_{b_n}}>\delta\right).
\end{align*}
Now, recall that $a_n\sim n^{1/\kappa}L(n)$ and $a_{b_n}\sim b_n^{1/\kappa}L(b_n)$ as $n\to\infty$. An application of Potter's bounds \citep[Theorem 1.5.6]{bingham1989regular} yields that $L(n)/L(b_n)\leqslant c(n/b_n)^\epsilon$ for any $c>1,\;\epsilon>0$ with $n$ sufficiently large. Choosing $\epsilon$ such that $1-1/\kappa-\epsilon>0$ (which is possible since $\kappa\in(1,2)$),
\begin{equation}\label{eq:potter_bounds}
    \lim_{n\to\infty}\dfrac{b_n\;a_n}{n\;a_{b_n}}=\lim_{n\to\infty}\left(\dfrac{b_n}{n}\right)^{1-1/\kappa}\dfrac{L(n)}{L(b_n)}\leqslant \lim_{n\to\infty}c\left(\dfrac{b_n}{n}\right)^{1-1/\kappa-\epsilon}=0.
\end{equation}
Using this together with Theorem \ref{theo:inf_var_jointCLT} we have that 
\begin{equation}\label{potter_arg}
    \dfrac{b_n(\overline{X}_n-\mu)}{a_{b_n}}=\dfrac{b_n\;a_n}{n\;a_{b_n}}a_n^{-1}\sum_{t=1}^n(X_t-\mu)\xrightarrow{\mathbb{P}}0,\quad\quad n\to\infty.
\end{equation}
Consequently, $\mathbb{E}[\mathbf{1}(b_n|\overline{X}_n-\mu|/(2a_{b_n})>\delta)]=\mathbb{P}(b_n|\overline{X}_n-\mu|/(2a_{b_n})>\delta)\to0$ as $n\to\infty$, for every $\delta>0$, and an application of Lemma \ref{lem:sub_variance} yields that,
\begin{equation*}
    q_n^{-1}\sum_{i=1}^{q_n}\mathbf{1}\left(\dfrac{b_n|\overline{X}_n-\mu|}{2a_{b_n}}>\delta\right)\xrightarrow{\mathbb{P}}0,\qquad n\to\infty.
\end{equation*}
Similarly, noting $|a_{b_n}^{-1}S_{i,b_n}|=O_\mathbb{P}(1)$ from Theorem \ref{theo:inf_var_jointCLT}, and that $b_n\to\infty$ as $n\to\infty$ gives us that $\mathbb{E}[\mathbf{1}(|a_{b_n}^{-1}S_{i,b_n}|>(d\,b_n/4-\delta^2)/\delta)]\to0$ as $n\to\infty$, and by Lemma \ref{lem:sub_variance},
\begin{equation}\label{e 2}
    q_n^{-1}\sum_{i=1}^{q_n}\mathbf{1}\left(|a_{b_n}^{-1}S_{i,b_n}|>\dfrac{d\,b_n/4-\delta^2}{\delta}\right)\xrightarrow{\mathbb{P}}0,\quad\quad n\to\infty.
\end{equation}
This proves that for every $\varepsilon>0$ it holds that $q_n^{-1}\sum_{i=1}^{q_n}\mathbf{1}(|E_{i,\gamma}|>d)<\varepsilon$ with probability tending to one. \smallskip

Using this result with $d>0$ sufficiently small together with the relations \eqref{Upper 1} and \eqref{Lower 1}, we have for arbitrarily small $\varepsilon>0$ and every $x>0$ that
\begin{equation}\label{sandwich 1}
    q_n^{-1}\sum_{i=1}^{q_n}\mathbf{1}\big(\gamma_{i,b_n}\leqslant x-d^{1/2}\big)-\varepsilon\leqslant L_{\gamma,b_n}(x)\leqslant q_n^{-1}\sum_{i=1}^{q_n}\mathbf{1}\big(\gamma_{i,b_n}\leqslant x+d^{1/2}\big)+\varepsilon,
\end{equation}
holds, with probability tending to one. Letting $d$ and $\varepsilon$ go to zero, we note that (1) of the lemma follows by \eqref{sandwich 1} if we can show that
\begin{equation}\label{final_relation}
    q_n^{-1}\sum_{i=1}^{q_n}\mathbf{1}\big(\gamma_{i,b_n}\leqslant x\pm d^{1/2}\big)\xrightarrow{\mathbb{P}}\mathbb{P}(\zeta_{\kappa/2}\leq x\pm d^{1/2}),\quad\quad n\to\infty.
\end{equation}
The convergence in \eqref{final_relation} holds by applications of Chebyshev's inequality and Portmanteau's theorem, provided that $\mathrm{var}(q_n^{-1}\sum_{i=1}^{q_n}\mathbf{1}(\gamma_{i,b_n}\leqslant x))\to0$ as $n\to\infty$. However, this follows directly by an application of Lemma \ref{lem:sub_variance}. This proves (1).\medskip

\textit{Proof of (2):} From (1), we have that $L_{\gamma,b_n}(x)\xrightarrow{\mathbb{P}}\mathbb{P}(\zeta_{\kappa/2}\leqslant x)$, and since $\zeta_{\kappa/2}$ has a continuous distribution, 
   $ c_{b_n}(1-\beta) \xrightarrow{\mathbb{P}}\inf\{x\in\mathbb{R}:\mathbb{P}(\zeta_{\kappa/2}\leqslant x)\geqslant1-\beta\} =: c(1-\beta)$.
Hence, with probability tending to one, it holds for $\varepsilon>0$ that     $c(1-\beta)-\varepsilon\leqslant c_{b_n}(1-\beta)\leqslant c(1-\beta)+\varepsilon$.
Using this result, we have for sufficiently large $n$,
\begin{equation*}
    \mathbb{P}(\gamma_{n,2}\leqslant c(1-\beta)-\varepsilon)\leqslant\mathbb{P}(\gamma_{n,2}\leqslant c_{b_n}(1-\beta))\leqslant \mathbb{P}(\gamma_{n,2}\leqslant c(1-\beta)+\varepsilon).
\end{equation*}
Observe for any choice of $\varepsilon>0$, we have that $c(1-\beta)\pm\varepsilon$ are continuity points of $\mathbb{P}(\zeta_{\kappa/2}\leqslant x)$. Then since $\gamma_{n,2}\xrightarrow{d}\zeta_{\kappa/2}$, we have by first letting $n\to\infty$ and then $\varepsilon\downarrow0$, 
\begin{align}
    \lim_{n\to\infty}\mathbb{P}(\gamma_{n,2}\leqslant c_{b_n}(1-\beta))&\geqslant \lim_{\varepsilon\downarrow0}\lim_{n\to\infty}\mathbb{P}(\gamma_{n,2}\leqslant c(1-\beta)-\varepsilon) \notag\\
    &=\lim_{\varepsilon\downarrow0}\mathbb{P}(\zeta_{\kappa/2}\leqslant c(1-\beta)-\varepsilon)=\mathbb{P}(\zeta_{\kappa/2}\leqslant c(1-\beta))=1-\beta\notag,\\[0.1in]
    \lim_{n\to\infty}\mathbb{P}(\gamma_{n,2}\leqslant c_{b_n}(1-\beta))&\leqslant \lim_{\varepsilon\downarrow0}\lim_{n\to\infty}\mathbb{P}(\gamma_{n,2}\leqslant c(1-\beta)+\varepsilon) \notag\\
    &=\lim_{\varepsilon\downarrow0}\mathbb{P}(\zeta_{\kappa/2}\leqslant c(1-\beta)+\varepsilon)=\mathbb{P}(\zeta_{\kappa/2}\leqslant c(1-\beta))=1-\beta. \notag
\end{align}
The right-hand sides hold since $\mathbb{P}(\zeta_{\kappa/2}\leqslant x)$ is continuous. This proves (2). \hfill $\square$
\begin{lemma}\label{Sub_lemma2}
    Let  $(X_t)_{t\in\mathbb{Z}}$ satisfy Assumptions \ref{ass_RV}--\ref{ass:AC+MX} with tail index $\kappa\in(1,2)$. Furthermore, let $L_{R,b_n}(x)$ be defined as in \eqref{approximations} with $(b_n)_{n\geqslant1}$ chosen such that $b_n/n\to0$ and $b_n\to\infty$ as $n\to\infty$. 
    \begin{itemize}
        \item [\emph{(1)}] If $x$ is a continuity point of $\mathbb{P}(\xi_\kappa/\zeta_{\kappa/2}\leqslant \cdot)$, then $L_{R,b_n}(x)\xrightarrow{\mathbb{P}}\mathbb{P}(\xi_\kappa/\zeta_{\kappa/2}\leqslant x)$ as $n\to\infty$.
        \item [\emph{(2)}] For $\beta\in(0,1)$, define the corresponding quantiles $c_{b_n}(1-\beta)=\inf\{x\in\mathbb{R}:L_{R,b_n}(x)\geqslant1-\beta\}$ and $c(1-\beta)=\inf\{x\in\mathbb{R}:\mathbb{P}(\xi_{\kappa}/\zeta_{\kappa/2}\leqslant x)\geqslant1-\beta\}$. 
        If $L(\xi_\kappa/\zeta_{\kappa/2}\leqslant x)$ is continuous in a neighborhood around $c(1-\beta)$, it holds that $\mathbb{P}(R_{n}\leqslant c_{b_n}(1-\beta))\rightarrow1-\beta$, as $n\to\infty.$
    \end{itemize}
\end{lemma}
\noindent \textbf{Proof of Lemma \ref{Sub_lemma2}:} \textit{Proof of (1).} Observe that $L_{R,b_n}(x)$ in \eqref{approximations} can be written as
\begin{align}\label{approx in con}
    L_{R,b_n}(x)&=q_n^{-1}\sum_{i=1}^{q_n}\mathbf{1}\left(\dfrac{a_{b_n}^{-1}\sum_{t=i}^{i+b_n-1}(X_t-\mu)-a_{b_n}^{-1}b_n(\overline{X}_n-\mu)}{\tilde{\gamma}_{i,b_n}}\leqslant x\right)\notag\\
    &=q_n^{-1}\sum_{i=1}^{q_n}\mathbf{1}\left(\dfrac{a_{b_n}^{-1}S_{i,b_n}}{\tilde{\gamma}_{i,b_n}}\leqslant x+\dfrac{a_{b_n}^{-1}b_n(\overline{X}_n-\mu)}{\tilde{\gamma}_{i,b_n}}\right),
\end{align}
where we have used the same notation as in the proof of Lemma \ref{Sub_Lemma 1}, that is, $S_{i,b_n}=\sum_{t=i}^{i+b_n-1}(X_t-\mu)$ and $\tilde{\gamma}_{i,b_n}^2=a_{b_n}^{-2}\sum_{t=i}^{i+b_n-1}(X_t-\overline{X}_n)^2$. Now, using that we have the following set relation for every $x\in\mathbb{R}$ and $y>0$,
\begin{equation*}
    \left\{\dfrac{a_{b_n}^{-1}S_{i,b_n}}{\tilde{\gamma}_{i,b_n}}\leqslant x+\dfrac{a_{b_n}^{-1}b_n(\overline{X}_n-\mu)}{\tilde{\gamma}_{i,b_n}}\right\}\subset\left\{\dfrac{a_{b_n}^{-1}S_{i,b_n}}{\tilde{\gamma}_{i,b_n}}\leqslant x+y\right\}\cup\left\{\dfrac{a_{b_n}^{-1}b_n(\overline{X}_n-\mu)}{\tilde{\gamma}_{i,b_n}}>y\right\},
\end{equation*}
we obtain for every $x\in\mathbb{R}$ and $y>0$ that
\begin{equation}\label{set_ineq}
    L_{R,b_n}(x)\leqslant q_n^{-1}\sum_{i=1}^{q_n}\mathbf{1}\left(\dfrac{a_{b_n}^{-1}S_{i,b_n}}{\tilde{\gamma}_{i,b_n}}\leqslant x+y\right)+q_n^{-1}\sum_{i=1}^{q_n}\mathbf{1}\left(\dfrac{a_{b_n}^{-1}b_n(\overline{X}_n-\mu)}{\tilde{\gamma}_{i,b_n}}>y\right).
\end{equation}
First, we seek to show that $q_n^{-1}\sum_{i=1}^{q_n}\mathbf{1}\left(a_{b_n}^{-1}b_n(\overline{X}_n-\mu)/{\tilde{\gamma}_{i,b_n}}>y\right)=o_\mathbb{P}(1)$ as $n\to\infty$. For this, we observe that for any $d>0$,
\begin{align*}
    q_n^{-1}&\sum_{i=1}^{q_n}\mathbf{1}\left(\dfrac{a_{b_n}^{-1}b_n(\overline{X}_n-\mu)}{\tilde{\gamma}_{i,b_n}}>y\right) \\&=1-q_n^{-1}\sum_{i=1}^{q_n}\mathbf{1}\left(\dfrac{a_{b_n}^{-1}b_n(\overline{X}_n-\mu)}{\tilde{\gamma}_{i,b_n}}\leqslant y\right) \\
    &=1-q_n^{-1}\sum_{i=1}^{q_n}\mathbf{1}\left(\dfrac{b_n\;a_n}{n\;a_{b_n}\;y}a_n^{-1}\sum_{t=1}^n(X_t-\mu)\leqslant \tilde{\gamma}_{i,b_n}\right) \\
    &\leqslant 1-q_n^{-1}\sum_{i=1}^{q_n}\mathbf{1}\left(d/y\leqslant \tilde{\gamma}_{i,b_n}\right) +q_n^{-1}\sum_{i=1}^{q_n}\mathbf{1}\left(\dfrac{b_n\;a_n}{n\;a_{b_n}}a_n^{-1}\sum_{t=1}^n(X_t-\mu)>d\right) 
\end{align*}
Using the relation \eqref{potter_arg}, we have by an application of Lemma \ref{lem:sub_variance} that,
\begin{equation*}
    q_n^{-1}\sum_{i=1}^{q_n}\mathbf{1}\left(\dfrac{b_n\;a_n}{n\;a_{b_n}}a_n^{-1}\sum_{t=1}^n(X_t-\mu)>d\right)\xrightarrow{\mathbb{P}}0,\qquad n\to\infty.
\end{equation*}
Furthermore, using the conclusion (1) of Lemma \ref{Sub_Lemma 1}, and letting $d\downarrow 0$, we also obtain that
\begin{equation*}
    1-\lim_{d\downarrow 0}q_n^{-1}\sum_{i=1}^{q_n}\mathbf{1}(\tilde{\gamma}_{i,b_n}\geqslant d/y)\xrightarrow{\mathbb{P}}1-\lim_{d\downarrow0}\mathbb{P}(\zeta_{\kappa/2}\geqslant d/y)=0,\qquad n\to\infty.
\end{equation*}
Collecting terms, we therefore have for every $x\in\mathbb{R}$ and $y,\delta>0$, as $n\to\infty$,
\begin{align*}
    L_{R,b_n}(x)&\leqslant q_n^{-1}\sum_{i=1}^{q_n}\mathbf{1}\left(\dfrac{a_{b_n}^{-1}S_{i,b_n}}{\tilde{\gamma}_{i,b_n}}\leqslant x+y\right)+o_{\mathbb{P}}(1) \\
    &=q_n^{-1}\sum_{i=1}^{q_n}\mathbf{1}\left(\dfrac{a_{b_n}^{-1}S_{i,b_n}}{\gamma_{i,b_n}}\leqslant(x+y)\Big(\dfrac{\tilde{\gamma}_{i,b_n}}{\gamma_{i,b_n}}-1\Big)+(x+y)\right)+o_\mathbb{P}(1) \\
    &\leqslant q_n^{-1}\sum_{i=1}^{q_n}\mathbf{1}\left(\dfrac{a_{b_n}^{-1}S_{i,b_n}}{\gamma_{i,b_n}}\leqslant(x+y)(1+\delta)\right)+q_n^{-1}\sum_{i=1}^{q_n}\mathbf{1}\left(\dfrac{\tilde{\gamma}_{i,b_n}}{\gamma_{i,b_n}}-1>\delta\right)+o_\mathbb{P}(1),
\end{align*}
where $\gamma_{i,b_n}^2=a_{b_n}^{-2}\sum_{i=1}^{i+b_n-1}(X_t-\mu)$ as in the proof of Lemma \ref{Sub_Lemma 1}. Using \eqref{eq:sub_gamma_eq1} observe,
\begin{align*}
    \tilde{\gamma}_{i,b_n}^2/\gamma_{i,b_n}^2=\dfrac{\gamma_{i,b_n}^2+b_n\;a_{b_n}^{-2}(\overline{X}_n-\mu)^2-2a_{b_n}^{-2}(\overline{X}_n-\mu)S_{i,b_n}}{\gamma_{i,b_n}^2}.
\end{align*}
With $E_{i,\gamma}:=2a_{b_n}^{-2}(\overline{X}_n-\mu)S_{i,b_n}-b_n\;a_{b_n}^{-2}(\overline{X}_n-\mu)^2$, we have that
\begin{align*}
    \{\tilde{\gamma}_{i,b_n}/\gamma_{i,b_n}-1>\delta\}&=\{(\gamma_{i,b_n}^2-E_{i,\gamma})/{\gamma_{i,b_n}^2}>(\delta+1)^2\} \\
    &= \{-\delta^2-2\delta>E_{i,\gamma}/\gamma_{i,b_n}^2\}\subset \{\delta^2+2\delta<|E_{i,\gamma}|/\gamma_{i,b_n}^2\}.
\end{align*}
By this relation and since we know $q_n^{-1}\sum_{i=1}^{q_n}\mathbf{1}(|E_{i,\gamma}|>d)=o_{\mathbb{P}}(1)$ for every $d>0$ from the proof of Lemma \ref{Sub_Lemma 1}, we have for every $\delta,d>0$,
\begin{align*}
    q_n^{-1}\sum_{i=1}^{q_n}\mathbf{1}\left(\dfrac{\tilde{\gamma}_{i,b_n}}{\gamma_{i,b_n}}-1>\delta\right)&\leqslant q_n^{-1}\sum_{i=1}^{q_n}\mathbf{1}\left(\dfrac{d}{\gamma_{i,b_n}^2}>\delta^2+2\delta\right)+q_n^{-1}\sum_{i=1}^{q_n}\mathbf{1}\left(|E_{i,\gamma}|>d\right) \\
    &=q_n^{-1}\sum_{i=1}^{q_n}\mathbf{1}\left(\dfrac{d}{\gamma_{i,b_n}^2}>\delta^2+2\delta\right)+o_\mathbb{P}(1).
\end{align*}
Letting $d\downarrow0$, an application of Lemma \ref{lem:sub_variance} to the right-hand side of this expression now yields that $q_n^{-1}\sum_{i=1}^{q_n}\mathbf{1}\left(\tilde{\gamma}_{i,b_n}/\gamma_{i,b_n}-1>\delta\right)=o_\mathbb{P}(1)$ as $n\to\infty$. Therefore, with probability tending to 1, we have the following bound for $L_{R,b_n}(x)$ for any $x\in\mathbb{R}$ and $y,\delta,\varepsilon>0$
\begin{equation}\label{Upper 2}
    L_{R,b_n}\leqslant q_n^{-1}\sum_{i=1}^{q_n}\mathbf{1}\left(\dfrac{a_{b_n}^{-1}S_{i,b_n}}{\gamma_{i,b_n}}\leqslant(x+y)(1+\delta)\right)+\varepsilon.
\end{equation}
One can similarly obtain a lower bound for $L_{R,b_n}(x)$ by symmetric arguments, yielding with probability tending to $1$, for every $x\in\mathbb{R}$ and $y,\delta,\varepsilon >0$
\begin{equation}\label{Lower 2}
    L_{R,b_n}(x)\geqslant q_n^{-1}\sum_{i=1}^{q_n}\mathbf{1}\left(\dfrac{a_{b_n}^{-1}S_{i,b_n}}{\gamma_{i,b_n}}\leqslant (x-y)(1+\delta)\right)-\varepsilon.
\end{equation}
Using the bounds in \eqref{Upper 2}--\eqref{Lower 2}, claim (1) of the lemma holds by similar arguments as those applied to \eqref{sandwich 1} in the proof of Lemma \ref{Sub_Lemma 1}, if $\mathrm{var}(q_n^{-1}\sum_{i=1}^{q_n}\mathbf{1}(a_{b_n}^{-1}S_{i,b_n}/\gamma_{i,b_n}\leqslant x)\to0$ as $n\to\infty$. The latter holds by an application of Lemma \ref{lem:sub_variance}.\smallskip

\textit{Proof of (2).} This result follows from (1) and the same quantile convergence argument as in the proof of Lemma \ref{Sub_Lemma 1}(2), using the continuity of $\mathbb{P}(\xi_{\kappa}/\zeta_{\kappa/2}\leqslant x)$ in a neighborhood of $c(1-\beta)$ and the fact that $R_n\xrightarrow{d}\xi_{\kappa}/\zeta_{\kappa/2}$ as $n\to\infty$.\hfill $\square$\bigskip

The following Proposition states the asymptotic validity of Algorithm \ref{algo} when $(X_t)_{t\in\mathbb{Z}}$ is regularly varying with tail index $\kappa\in(1,2)$.
\begin{proposition} \label{prop:validity_Null_and_alt}
    Let $(X_t)_{t\in\mathbb{Z}}$ satisfy Assumptions \ref{ass_RV}-\ref{ass:AC+MX} with tail index $\kappa \in (1,2)$. Assume that $b_n$ provided by Algorithm \ref{algo} satisfies that $b_n\rightarrow\infty$ and $b_n=o(n)$. Let $L_{n,b_n}(x)$ be defined as in \eqref{eq:def_L_n_b} and let
    \begin{equation}
         C_{n,b_n}(y)=\inf\{x:L_{n,b_n}(x) \geqslant y\}
    \end{equation}
    for some $y\in(0,1)$.
  Suppose that $\mathbb{E}[X_t]=0$.
    \begin{itemize}
        \item [\emph{(1.a)}] If $x$ is a continuity point of $\mathbb{P}(\xi_\kappa/\zeta_{\kappa/2}\leqslant \cdot)$, then $L_{n,b_n}(x)\xrightarrow{\mathbb{P}}\mathbb{P}(\xi_\kappa/\zeta_{\kappa/2}\leqslant x)$ as $n\to\infty$, where $(\xi_\kappa,\zeta_{\kappa/2})$ denotes the limiting  vector in Theorem \ref{theo:inf_var_jointCLT}. 
        \item [\emph{(1.b)}] If $\mathbb{P}(\xi_\kappa/\zeta_{\kappa/2}\leqslant x)$ is continuous in a neighborhood around $c(1-\beta):=\inf\{x\in\mathbb{R}:\mathbb{P}(\xi_{\kappa}/\zeta_{\kappa/2}\leqslant x)\geqslant1-\beta\}$, then $\mathbb{P}(T_{n}\leqslant C_{n,b_n}(1-\beta))\rightarrow1-\beta$ as $n\to\infty$.
    \end{itemize}
     Suppose instead that $\mathbb{E}[X_t]\neq0$.
     \begin{itemize}
        \item [\emph{(2)}] Then $\mathbb{P}(T_{n} \notin [C_{n,b_n}(\beta/2),C_{n,b_n}(1-\beta/2)])\to1$ as $n\to\infty$.
    \end{itemize}
\end{proposition}

\noindent\textbf{Proof of Proposition \ref{prop:validity_Null_and_alt}:} 
Since, $\mathbb{E}[X_t]=0$, points (1.a) and (1.b) hold by the arguments used for proving Lemma \ref{Sub_lemma2}, with $\overline{X}_n$ and $\mu$ replaced by zeros. To show point (2), we note that by Proposition \ref{Master_Lemma}, $|T_n|\to\infty$ in probability. Consequently, the result holds provided that $|T_{i,b_n}|=o_{\mathbb{P}}(|T_{n}|)$. 
To show this, we begin by recalling the form of the involved statistics, i.e.,
\begin{equation}\label{stats_proof}
    T_n=\dfrac{a_{n}^{-1}\sum_{t=1}^nX_t}{a_n^{-1}\left(\sum_{t=1}^nX_t^2\right)^{1/2}}=:\dfrac{S_n}{\gamma_{n}}\quad\text{and}\quad T_{i,b_n}=\dfrac{a_{b_n}^{-1}\sum_{t=i}^{i+b_n-1}X_t}{a_{b_n}^{-1}\left(\sum_{t=i}^{i+b_n-1}X_t^2\right)^{1/2}}=:\dfrac{S_{i,b_n}}{\gamma_{i,b_n}},
\end{equation}
for $n\geqslant 1$ and $i=1,\dots,q_n$. Note that $S_n,S_{i,b_n},\gamma_{n}$ and $\gamma_{i,b_n}$ are  defined here differently than in the previous proofs in this section.\smallskip

By standard calculations and the definition of $(a_n)_{n\geqslant1}$, i.e. the sequence satisfies $n\mathbb{P}(|X|>a_n)\to1$ as $n\to\infty$, we obtain for $n\to\infty$
\begin{align*}
    \gamma_{n,2}^2%&=a_n^{-2}\sum_{t=1}^n(X_t-\mathbb{E}[X_t]+\mathbb{E}[X_t])^2 \\
    &=a_n^{-2}\sum_{t=1}^n(X_t-\mathbb{E}[X_t])^2+\dfrac{n}{a_n^{2}}(\mathbb{E}[X_t])^2+2a_n^{-2}\mathbb{E}[X_t]\sum_{t=1}^n(X_t-\mathbb{E}[X_t]) \\
    &\sim a_n^{-2}\sum_{t=1}^n(X_t-\mathbb{E}[X_t])^2+\dfrac{n}{n^{2/\kappa}L(n)}(\mathbb{E}[X_t])^2+\dfrac{2\mathbb{E}[X_t]}{a_n}a_n^{-1}\sum_{t=1}^n(X_t-\mathbb{E}[X_t]):=J_1+J_2+J_3.
\end{align*}
Recall that $a_n^{-1}\left(a_n^{-1}\sum_{t=1}^n(X_t-\mathbb{E}[X_t])^2,\sum_{t=1}^n(X_t-\mathbb{E}[X_t])\right)\xrightarrow{d}(\zeta_{\kappa/2}^2,\xi_\kappa)$ as $n\to\infty$, by Theorem \ref{theo:inf_var_jointCLT} and a continuous mapping argument. A straightforward application of this result and that $a_n\to\infty$ as $n\to\infty$ yields that $J_1+J_3\xrightarrow{d}\zeta_{\kappa/2}$ as $n\to\infty$. Furthermore, from Lemma 2 in \citet[p.277]{Feller1971}, for sufficiently large $n$ and any $\delta>0$, $n^{-\delta}<L(n)<n^\delta$. Consequently, choosing $\delta$ such that $1+\delta<2/\kappa$,  $J_2\rightarrow0$ as $n\to\infty$, and we conclude that
\begin{equation}
    \gamma_{n}^2\xrightarrow{d}\zeta_{\kappa/2}^2,\quad\quad n\to\infty.
\end{equation}
By the same arguments, we can also show that $\gamma_{i,b_n}^2\xrightarrow{d}\zeta_{\kappa/2}^2$ as $n\to\infty$, for every $i=1,\dots,q_n$, by noting that $b_n\to\infty$ as $n\to\infty$. Hence, the standardized denominators of $T_n$ and $T_{i,b_n}$, respectively, are a.s. asymptotically finite. On the other hand, we have that both $|S_n|$ and $|S_{i,b_n}|$ in \eqref{stats_proof} diverges as $n\to\infty$. Specifically,
\begin{align*}
    \big|S_n\big|&=%\big|a_{n}^{-1}\sum_{t=1}^n(X_t-\mathbb{E}[X_t]+\mathbb{E}[X_t])\big|=
    \big|a_n^{-1}\sum_{t=1}^n(X_t-\mathbb{E}[X_t])+n\;a_n^{-1}\mathbb{E}[X_t]\big| =\big|n\;a_n^{-1}\mathbb{E}[X_t]+O_{\mathbb{P}}(1)\big|,
\end{align*}
and similarly $\big|S_{i,b_n}\big|=\big|b_n\;a_{b_n}^{-1}\mathbb{E}[X_t]+O_\mathbb{P}(1)\big|$ for $i=1,\dots,q_n$. These results imply that
\begin{equation*}
    \dfrac{|T_{i,b_n}|}{|T_n|}\sim\dfrac{b_n\;a_n}{n\;a_{b_n}}\sim\left(\dfrac{b_n}{n}\right)^{1-1/\kappa}\dfrac{L(n)}{L(b_n)},\quad\quad n\to\infty.
\end{equation*}
Now an application of the Potter bounds, similar to the proof of Lemma \ref{Sub_Lemma 1}, yields for $c>1$ and $1-1/\kappa-\varepsilon>0$ that 
\begin{equation*}
    \lim_{n\to\infty}\left(\dfrac{b_n}{n}\right)^{1-1/\kappa}\dfrac{L(n)}{L(b_n)}\leq \lim_{n\to\infty}c\left(\dfrac{b_n}{n}\right)^{1-1/\kappa-\varepsilon}=0.
\end{equation*}
This proves the proposition.\hfill$\square$

\bigskip

We end this section by the following result for the subsample-based critical values when the tail index $\kappa\in(0,1)$.
\begin{proposition} \label{subsample_special}
    Let $(X_t)_{t\in\mathbb{Z}}$ satisfy Assumptions \ref{ass_RV}--\ref{ass:AC+MX} with tail index $\kappa\in(0,1)$. Let $L_{n,b_n}(x)$ be defined as in \eqref{eq:def_L_n_b} with $(b_n)_{n\geqslant 1}$ chosen such that $b_n\to\infty$ and $b_n=o(n)$ as $n\to\infty$.
     \begin{itemize}
        \item [\emph{(1)}] If $x$ is a continuity point of $\mathbb{P}(\xi_\kappa/\zeta_{\kappa/2}\leqslant \cdot)$, then $L_{n,b_n}(x)\xrightarrow{\mathbb{P}}\mathbb{P}(\xi_\kappa/\zeta_{\kappa/2}\leqslant x)$ as $n\to\infty$.
        \item [\emph{(2)}] For $\beta\in(0,1)$, define the corresponding quantiles $c_{b_n}(1-\beta)=\inf\{x\in\mathbb{R}:L_{n,b_n}(x)\geqslant1-\beta\}$ and $c(1-\beta)=\inf\{x\in\mathbb{R}:\mathbb{P}(\xi_{\kappa}/\zeta_{\kappa/2}\leqslant x)\geqslant1-\beta\}$. \\
        If $\mathbb{P}(\xi_\kappa/\zeta_{\kappa/2}\leqslant x)$ is continuous in a neighborhood around $c(1-\beta)$, then it holds that $\mathbb{P}(T_{n}\leqslant c_{b_n}(1-\beta))\rightarrow1-\beta$, as $n\to\infty.$
    \end{itemize}
\end{proposition}

\noindent \textbf{Proof of Proposition \ref{subsample_special}:} From stationarity of $(X_t)_{t\in\mathbb{Z}}$ we have that $\mathbb{E}[L_{n,b_n}(x)]=\mathbb{P}\left(T_{1,b_n}\leqslant x\right)$.
Moreover, Chebychev's inequality yields for every $\varepsilon>0$,
\begin{equation*}
    \mathbb{P}\left(\big|L_{n,b_n}(x)-\mathbb{E}[L_{n,b_n}(x)]\big|>\varepsilon\right)\leqslant \mathrm{var}(L_{n,b_n}(x))/\varepsilon^2.
\end{equation*}
Hence, if we can show that $\mathrm{var}(L_{n,b_n}(x))\rightarrow0$ as $n\to\infty$, then $L_{n,b_n}(x)\xrightarrow{\mathbb{P}}\mathbb{P}\left(T_{1,b_n}\leqslant x\right).$
Next, it holds by Theorem \ref{theo:inf_var_jointCLT} that $T_{1,b_n}\xrightarrow{d}\xi_\kappa/\zeta_{\kappa/2}$. Consequently, for every continuity point $x$ of $\mathbb{P}(\xi_\kappa/\zeta_{\kappa/2}\leqslant x)$, it holds by Portmanteau's theorem that
\begin{equation*}
    \lim_{n\to\infty}\mathbb{P}\left(T_{1,b_n}\leqslant x\right)=\mathbb{P}(\xi_\kappa/\zeta_{\kappa/2}\leqslant x).
\end{equation*}
Collecting terms yields the desired relation, i.e., $L_{n,b_n}(x)\xrightarrow{\mathbb{P}}\mathbb{P}(\xi_\kappa/\zeta_{\kappa/2}\leqslant x)$ for every continuity point $x$ of the limit. Hence, it suffices to show that $\mathrm{var}(L_{n,b_n}(x))\rightarrow0$. However, this follows by Lemma \ref{lem:sub_variance}, proving (1) of the proposition. Next, using (1), the conclusion (2) follows by arguments given in the proof of Lemma \ref{Sub_lemma2}(2).\hfill$\square$

\begin{lemma}\label{lem:sub_variance}
    Let $(X_t)_{t\in\mathbb{Z}}$ be a strongly mixing process of real-valued random variables, and let $g_n$ be a real-valued, deterministic function acting on $(X_t)_{t\in\mathbb{Z}}$, satisfying $0\leqslant g_n\leqslant 1$. Define
    \begin{equation*}
        Y_n=q_n^{-1}\sum_{i=1}^{q_n}g_n(X_{i},\dots,X_{i+b_n-1}),
    \end{equation*}
    where $q_n=n-b_n+1$ and $(b_n)_{n\geqslant1}$ is chosen such that $b_n\to\infty$ and $b_n/n\to0$ as $n\to\infty$. Then it holds that $\mathrm{var}(Y_n)\to0$ as $n\to\infty$. Furthermore, if $\mathbb{E}[g_n(X_i,\dots,X_{i+b_n-1})]\to0$ then $Y_n\xrightarrow{\mathbb{P}}0$ as $n\to\infty$.
\end{lemma}

\noindent \textbf{Proof of Lemma \ref{lem:sub_variance}:} To simplify notation, we write $g_n(X_i,\dots,X_{i+b_n-1})=g_n(\mathbf{X}_{i,i+b_n-1})$ for $i=1,\dots,q_n$. Using that $(X_t)_{t\in\mathbb{Z}}$ is stationary, for $n$ sufficiently large, such that $2b_n \leqslant q_n -1$, the variance of $Y_n$ is given by
\begin{align*}
    \mathrm{var}(Y_n)&=q_n^{-1}\mathrm{var}\big(g_n(\mathbf{X}_{0,b_n-1})\big)+2q_n^{-2}\sum_{i=1}^{2b_n-1}(q_n-i)\,\mathrm{cov}\big(g_n(\mathbf{X}_{0,b_n-1}),g_n(\mathbf{X}_{i,i+b_n-1})\big) \\[0.025in]
    &\quad+2q_n^{-2}\sum_{i=2b_n}^{q_n-1}(q_n-i)\,\mathrm{cov}\big(g_n(\mathbf{X}_{0,b_n-1}),g_n(\mathbf{X}_{i,i+b_n-1})\big)=:J_1+J_2+J_3.
\end{align*}
Now, using that $|g_n|\leqslant 1$, it holds that $|\mathrm{cov}\big(g_n(\mathbf{X}_{i,i+b_n-1}),g_n(\mathbf{X}_{j,j+b_n-1})\big)|\leqslant 1$ for any $(i,j)\in\mathbb{Z}\times \mathbb{Z}$. Consequently, we have $J_1\to0$ as $n\to\infty$, and also
\begin{align*}
    J_2&\leqslant 2\,\dfrac{2b_n-1}{q_n}-\left(\dfrac{2b_n}{q_n}\right)^2+\dfrac{2b_n}{q_n^2} \\[0.025in]
    &=2\,\dfrac{2-1/b_n}{(n-b_n+1)/b_n}-\left(\dfrac{2}{(n-b_n+1)/b_n}\right)^2+\dfrac{2}{(n-b_n+1)^2/b_n}\to0,\qquad n\to\infty,
\end{align*}
where we inserted $q_n=n-b_n+1$, and used the fact that $b_n\to\infty$, $n/b_n\to\infty$ as $n\to\infty$. Next, note that for each $i=2b_n,\dots,q_n-1$, it holds that $i-(b_n-1)\geqslant b_n +1 $, consequently, by definition of the strong mixing numbers in \eqref{mix_coef},
\begin{align*}
    J_3\leqslant 2q_n^{-1}\sum_{i=2b_n}^{q_n-1}\big|\mathrm{cov}\big(g_n(\mathbf{X}_{0,b_n-1}),g_n(\mathbf{X}_{i,i+b_n-1})\big)\big|\leqslant 2\,\alpha_{b_n+1}\to0,\qquad n\to\infty.
\end{align*}
This proves that $\mathrm{var}(Y_n)\to0$, as $n\to\infty$. By Chebyshev's inequality, $|Y_n-\mathbb{E}[g_n(\mathbf{X}_{i,i+b_n-1})]|\xrightarrow{\mathbb{P}}0$. If, in addition, $\mathbb{E}[g_n(\mathbf{X}_{i,i+b_n-1})]\to0$, it also holds that
\begin{equation*}
    0\leqslant Y_n\leqslant |Y_n-\mathbb{E}[g_n(\mathbf{X}_{i,i+b_n-1})]|+\mathbb{E}[g_n(\mathbf{X}_{i,i+b_n-1})]\xrightarrow{\mathbb{P}}0,\qquad n\to\infty.
\end{equation*}
This proves the lemma.\hfill $\square$
\subsection{Proof of Theorem \ref{cor:validity_subsampling}}
Under Assumption \ref{mix_coef_ass}, the result follows directly by \citet[Theorem 3.5.1]{politis2012subsampling} setting $\tau_n = \sqrt{n}$ and $t_n(X_1,\dots,X_n)=T_n/\sqrt{n}$, and noting that $T_n$ has a Gaussian limiting distribution. Under Assumptions \ref{ass_RV}--\ref{ass:AC+MX}, the result follows by Proposition \ref{prop:validity_Null_and_alt}. \hfill $\square$

\section{Results and details related to Section \ref{sec:infinite_mean}}

\subsection{Subsampling algorithm for critical value construction}
The following Algorithm provides critical values for a test of $\mathsf{H}_0$ based on the statistic $ \widetilde{T}_n$ in \eqref{eq:modified_statistic} against alternatives where the mean may be infinite or undefined.
\begin{algorithm} \label{algo_infinite_mean}
    Given a sample $(X_t)_{t=1,\dots,n}$ of size $n > 1$, choose some integer $b_n\in (0,n)$ and some nominal level $\eta\in(0,1)$. 
    \begin{enumerate}
        \item Compute the statistic given in \eqref{eq:modified_statistic},
        \begin{equation*}
        \widetilde{T}_n= \frac{\sum_{t=1}^nX_t}{(\sum_{t=1}^nX_t^2)^{1/2}}\frac{\sum_{t=1}^n|X_t|}{n}.    
        \end{equation*}
        \item With $q_n := n-b_n+ 1$, compute the subsample statistics
        \begin{equation}\label{eq:block_stat_infinite_mean}
        \widetilde{T}_{i,{b_n}} := \dfrac{\sum_{t=i}^{i+b_n-1} X_t}{(\sum_{t=i}^{i+b_n-1} X_t^2)^{1/2}}\frac{\sum_{t=i}^{i+b_n-1}|X_t|}{b_n},\qquad i=1,\dots,q_n.
        \end{equation}
        \item Define 
        \begin{equation}\label{eq:def_L_n_beq_infinite_mean}
        \widetilde{L}_{n,{b_n}}(x):= q_n^{-1}\sum_{i=1}^{q_n}\mathbf{1}\big(|\widetilde{T}_{i,b}|\leqslant x\big),
        \end{equation}
        and compute the critical value
        \begin{equation}\label{eq:algo_C_n_b_eq:block_stat_infinite_mean}
            \widetilde{C}_{n,{b_n}}(1-\eta)=\inf\{x:\widetilde{L}_{n,{b_n}}(x)\geqslant1-\eta\}.
        \end{equation}

    \end{enumerate}
    Then a two-sided test of level $\beta$ rejects $\mathsf{H}_0$ if $|\widetilde{T}_n| >  \widetilde{C}_{n,{b_n}}(1-\eta)$. 
\end{algorithm}

\subsection{Proof of \eqref{eq:modified_subsampling_validity}}\label{sec:proof_infinite_mean}
We have that
\begin{align*}
    \frac{|\widetilde{T}_{b_n}|}{|\widetilde{T}_n|}=\frac{|T_{b_n}|\gamma_{b_n,1}/b_n}{|T_n|\gamma_{n,1}/n}=\frac{|T_{b_n}|}{|T_n|}\frac{a_{b_n}^{-1}\gamma_{b_n,1}}{a_n^{-1}\gamma_{n,1}}\frac{a_{b_n}/{b_n}}{a_n/n}=O_{\mathbb{P}}(1)\frac{a_{b_n}/{b_n}}{a_n/n},
\end{align*}
where the last equality follows by Theorem \ref{theo:inf_var_jointCLT} and Lemma \ref{lem:gamma_1}. It suffices to show that 
\begin{equation}\label{eq:fraction_a_b_a_n}
    \frac{a_{b_n}/{b_n}}{a_n/n} =o(1).
\end{equation}
Since $|X_0|$ is regularly varying with index $\kappa\in(0,1)$, it follows by the definition of $a_n$, that $a_n\sim n^{1/\kappa}L(n)$ for some slowly varying function $L(n)$. Consequently,
\begin{eqnarray*}
    \frac{a_{b_n}/{b_n}}{a_n/n} \sim \left(\frac{{b_n}}{n}\right)^{(1-\kappa)/\kappa}\frac{L({b_n})}{L(n)}.
\end{eqnarray*}
By \citet[Theorem 1.5.6]{bingham1989regular} there exist constants $C>0$ and $0<\delta<(1-\kappa)/\kappa$ such that (for $n$ and, hence, ${b_n}$ sufficiently large)
\begin{eqnarray*}
    \frac{L({b_n})}{L(n)}\leqslant C \left(\frac{{b_n}}{n}\right)^{-\delta}.
\end{eqnarray*}
Using that $(1-\kappa)/\kappa-\delta >0$ and that ${b_n}=o(n)$, we conclude that \eqref{eq:fraction_a_b_a_n} holds.

\begin{lemma}\label{lem:gamma_1}
   Suppose that Assumptions \ref{ass_RV}-\ref{ass:AC+MX} hold with $\kappa\in(0,1)$. With the deterministic sequence $(a_n)_{n\geqslant1}$ defined by \eqref{eq:def_a_n} and $\overline{\gamma}_{n}$ defined by \eqref{eq:def_gamma_1}, it holds that
   \begin{equation}
       a_n^{-1}(S_n,\overline{\gamma}_{n,1})\overset{d}{\to}(\xi_\kappa,\overline{\zeta}_\kappa), \quad \text{as} \ n\to \infty,
   \end{equation}
   where $\xi_\kappa$ and $\overline{\zeta}_\kappa$ are $\kappa$-stable random variables with $\mathbb{P}(\overline{\zeta}_\kappa>0)=1$.
\end{lemma}
\textit{Proof:} 
The proof is very similar to that of Theorem \ref{theo:inf_var_jointCLT}. Define
\begin{equation}
    \widetilde{\Psi}_n(u,\lambda) :=\mathbb{E}[\exp(ia_n^{-1}uS_n-a_n^{-1}\lambda\overline{\gamma}_{n})],\qquad (u,\lambda)\in \mathbb{R}\times\mathbb{R}_+.
\end{equation}
Then it suffices to show that for the integer sequence $r_n \to \infty$ in Assumption \ref{ass:AC+MX} and $k_n:=\lfloor n/r_n\rfloor\to\infty$ as $n\to \infty$, it holds for all $(u,\lambda)\in \mathbb{R}\times\mathbb{R}_+$ that
\begin{align}\label{eq:mix_condition_gamma_1}
    \widetilde{\Psi}_n(u,\lambda)  =  \left(\mathbb{E}[\exp(ia_n^{-1}uS_{r_n}-a_n^{-1}\lambda\overline{\gamma}_{r_n})]\right)^{k_n}+o(1).
\end{align}
Condition \eqref{eq:mix_condition_gamma_1} follows by arguments similar to the ones used to prove \eqref{eq:mix_condition}.\hfill $\square$

\subsection{Simulation results for Algorithm \ref{algo_infinite_mean}}\label{sim:infinite_mean}

\setcounter{figure}{0}
\setcounter{table}{0}
\counterwithin{figure}{section}
\counterwithin{table}{section}

\renewcommand{\thefigure}{\thesection.\arabic{figure}}
\renewcommand{\thetable}{\thesection.\arabic{table}}

We consider testing $\mathsf{H}_0$ against alternatives that allow for $\mathbb{E}[|X_t|]=\infty$. Similar to Section \ref{sec:Monte_Carlo}, the DGPs are given by the AR(1) process in \eqref{eq:sim_AR} 
with $Z_t$ (symmetric) $\textrm{Stable}(\kappa,0,1,0)$-distributed, $\kappa >0$. We note that $\mathbb{E}[|X_t|]=\infty$ if $\kappa < 1$. For such values of $\kappa$ and with $\delta = 0$, we do not expect the DM test statistic to diverge (by Proposition \ref{prop:AR_HAC}), whereas a test based on Algorithm \ref{algo_infinite_mean} should be consistent. These theoretical insights are in line with Table \ref{tab:power_sub_infmean}, which contains the empirical rejection frequencies for both tests. Across all values of $\kappa$ and sample sizes $n$, the rejection frequencies of the DM test are below the 5\% nominal level. In contrast, the rejection frequencies of the test based on Algorithm \ref{algo_infinite_mean} are increasing in $n$. We also note that the rejection frequencies for this test are decreasing in $\kappa$. Intuitively, as the tail index $\kappa$ increases toward one, the unconditional distribution of $X_t$ becomes "closer" to obeying the null hypothesis that $\mathbb{E}[X_t]=0$, since $\delta = 0$.

We also consider testing $\mathsf{H}_0$ using Algorithm \ref{algo_infinite_mean} when $\kappa\in(1,2)$ and $\delta>0$ such that $\mathbb{E}[X_t]>0$. Similar to the findings in Section \ref{sec:algo_power}, the rejection frequencies are increasing in $n$ and $\kappa$. Comparing with the test based on Algorithm \ref{algo}, the test based on Algorithm \ref{algo_infinite_mean} appears to be more powerful when $\kappa$ is close to one, which is expected given that Algorithm \ref{algo_infinite_mean} delivers non-trivial rejection frequencies even for values of $\kappa$ below one. Detailed simulation results are available upon request.

\begin{table}[H]
\centering
\caption{Rejection frequencies (in percent) of the Diebold-Mariano test (Section \ref{sec:Monte_Carlo}) and the test in Algorithm \ref{algo_infinite_mean} of $\mathsf{H}_0$. The DGP $(X_t)_{t=1,\dots,n}$ is given by \eqref{eq:sim_AR} with $\delta = 0$ and iid $\textrm{Stable}(\kappa,0,1,0)$-distributed noise $Z_t$, $\kappa \in (0,1)$. The nominal level of the tests is $5\%$. The rejection frequencies are based on $M=10^4$ Monte-Carlo replications.} \label{tab:power_sub_infmean}
\begin{tabular}{c|cc|cc|cc|cc}
\hline
\hline
$n$: & \multicolumn{2}{c|}{1000} & \multicolumn{2}{c|}{2000} & \multicolumn{2}{c|}{5000} & \multicolumn{2}{c}{10000} \\
\hline
Test: & DM & Alg \ref{algo_infinite_mean} & DM & Alg \ref{algo_infinite_mean} & DM & Alg \ref{algo_infinite_mean} & DM & Alg \ref{algo_infinite_mean}\\
\hline
\hline
$\kappa=0.1$ & 0.1 & 8.0 & 0.0 & 73.6 & 0.0 & 90.9 & 0.0 & 96.4 \\
\hline
$\kappa=0.2$ & 0.2 & 5.5 & 0.3 & 52.6 & 0.2 & 72.9 & 0.1 & 83.7 \\
\hline
$\kappa=0.3$ & 0.7 & 3.9 & 0.5 & 39.7 & 0.4 & 55.9 & 0.3 & 66.6 \\
\hline
$\kappa=0.4$ & 1.1 & 3.1 & 1.1 & 27.8 & 0.7 & 41.3 & 0.5 & 50.3 \\
\hline
$\kappa=0.5$ & 1.8 & 2.2 & 1.5 & 19.8 & 1.0 & 29.6 & 1.0 & 36.4 \\
\hline
$\kappa=0.6$ & 1.9 & 1.8 & 1.9 & 14.6 & 1.5 & 21.0 & 1.4 & 25.3 \\
\hline
$\kappa=0.7$ & 3.1 & 1.2 & 2.4 & 10.0 & 2.1 & 15.1 & 1.9 & 17.5 \\
\hline
$\kappa=0.8$ & 3.1 & 0.9 & 3.0 & 7.3 & 2.7 & 9.6 & 2.4 & 10.8 \\
\hline
$\kappa=0.9$ & 3.9 & 0.7 & 3.4 & 5.5 & 2.9 & 7.0 & 2.5 & 7.9 \\
\hline
\end{tabular}
\end{table}
\newpage
\section{Results related to Section \ref{sec:multivariate}}\label{sec:appendix_multivariate}
\setcounter{assumption}{0}
\renewcommand{\theassumption}{\Alph{section}\arabic{assumption}}

This section provides multivariate versions of Theorems \ref{theo:rio}--\ref{theo:inf_var_jointCLT} and a proof of Theorem \ref{thm:SPA}. 

\subsection{Regularly varying multivariate processes and limit theory}
We start by extending the definitions of strong mixing and  regular variation to vector-valued processes: A stationary process $(\mathbf{X}_t)_{t\in\mathbb{Z}}$ with $\mathbf{X}_t \in \mathbb{R}^m$ has strong mixing numbers
\begin{equation}\label{mix_coef_multi}
\alpha_k=\sup_{|g_1|,|g_2|\leqslant 1}\left|\mathrm{cov}\left[g_1(\dots,\mathbf{X}_{-1},\mathbf{X}_{0}),g_2(\mathbf{X}_{k},\mathbf{X}_{k+1},\dots)\right]\right|,\qquad k \geqslant 0,
\end{equation} 
and is said to be strongly mixing if $\alpha_k\to0$ as $k\to\infty$. Likewise, $(\mathbf{X}_t)_{t\in\mathbb{Z}}$ is said to be regularly varying with index $\kappa>0$, if $\Vert \mathbf{X}_0\Vert $ is regularly varying with index $\kappa>0$ and there exists a stochastic process of $\mathbb{R}^m$-valued variables $(\mathbf{\Theta}_t)_{t\in\mathbb{Z}}$ such that for all $t,h\in\mathbb{Z}$ satisfying $t\leqslant0\leqslant t+h$,
\begin{equation}\label{RV_TS_multi}
    \mathbb{P}(\Vert \mathbf{X}_0\Vert ^{-1}(\mathbf{X}_{t},\dots,\mathbf{X}_{t+h})\in\cdot\;|\;\Vert\mathbf{X}_0\Vert >x)\xrightarrow{w}\mathbb{P}((\mathbf{\Theta}_{t},\dots,\mathbf{\Theta}_{t+h})\in\cdot),\quad\quad x\to\infty.
\end{equation}
With $\mathbf{X}_t\in \mathbb{R}^m$ for some fixed  integer $k\geqslant1$, we state the multivariate versions of Assumptions \ref{mix_coef_ass}--\ref{ass:AC+MX} in terms of the the stationary process $(\mathbf{X}_t)_{t\in\mathbb{Z}}$. 
\begin{assumption}\label{ass:mixing_multivariate}
    For some $\varepsilon>0$, $\mathbb{E}[\Vert \mathbf{X}_0\Vert^{2+\varepsilon}]<\infty$, the strong mixing numbers of $(\mathbf{X}_t)_{t\in\mathbb{Z}}$ satisfy $\sum_{k=1}^\infty\alpha_k^{\varepsilon/(2+\varepsilon)}<\infty$, and with $\bm{\mu} = \mathbb{E}[\mathbf{X}_0]$,
    \begin{equation}\label{eq:def_Sigma}
        \bm{\Sigma}= \mathbb{E}[(\mathbf{X}_0 - \bm{\mu})(\mathbf{X}_0 - \bm{\mu})']+\sum_{i,j=1 \ i\neq j}^\infty\mathbb{E}[(\mathbf{X}_i-\bm{\mu})(\mathbf{X}_j-\bm{\mu})'] \ \ \text{is positive definite}.
    \end{equation}
\end{assumption}
\begin{assumption}\label{ass_RV_multivariate}
    The process $(\mathbf{X}_t)_{t\in\mathbb{Z}}$ is regularly varying with index $\kappa\in(0,2)\setminus \{1\}$ in the sense of \eqref{RV_TS_multi}.
\end{assumption}
Similar to \eqref{eq:def_a_n} we have for the multivariate case that there exists a sequence $(\tilde{a}_n)_{n\geqslant1}$ with $\tilde{a}_n\to\infty$ satisfying
\begin{equation} \label{eq:def_a_n_multivariate}
    n\prob(\Vert \mathbf{X}\Vert >\tilde{a}_n)\rightarrow1,\qquad n\rightarrow\infty.
\end{equation}
\begin{assumption}\label{AC,MX_multivariate}
    There exists a deterministic sequence $(r_n)_{n\geqslant1}$ with $r_n\rightarrow\infty$ and $r_n/n\rightarrow 0$ such that
    \begin{equation}\label{AC_multivariate}
\lim_{k\rightarrow\infty}\limsup_{n\rightarrow\infty}n\sum_{j=k}^{r_n}\mathbb{E}\big[\min\{\Vert\tilde{a}_n^{-1}\mathbf{X}_j\Vert,1\}\;\min\{\Vert \tilde{a}_n^{-1}\mathbf{X}_0\Vert,1\}\big]=0.
    \end{equation}
    Moreover, the strong mixing numbers of $(\mathbf{X}_t)_{t\in\mathbb{Z}}$ satisfy
    \begin{equation}\label{alpha_coefs_req_multi}
        \dfrac{n}{r_n}\alpha(\ell_n)\rightarrow0,
    \end{equation}
    for some sequence $\ell_n\rightarrow\infty$ satisfying $\ell_n/r_n\rightarrow0$, as $n\to \infty$.
\end{assumption}
\begin{corollary}
    \label{thm:rio_multivariate}
    Under Assumption \ref{ass:mixing_multivariate},
    \begin{equation}
        n^{-1/2}\left( \mathbf{S}_n - \bm{\mu},\tilde{\gamma}_{n} \right) \overset{d}{\to} \left(\mathbf{Z}, \sqrt{\mathbb{E}[\Vert \mathbf{X}_0\Vert^2]} \right),\qquad n\to\infty,
    \end{equation}
    where $\mathbf{Z}$ is $N(\mathbf{0},\bm{\Sigma})$-distributed with $\bm{\Sigma}$ given by \eqref{eq:def_Sigma}.
\end{corollary}
\textit{Proof:} By the ergodic theorem $n^{-1/2}\tilde{\gamma}_{n} \overset{\mathbb{P}}{\to}\sqrt{\mathbb{E}[\Vert \mathbf{X}_0\Vert^2]}$ as  $n\to\infty$. For any non-zero constant vector $\lambda\in\mathbb{R}^m$ $n^{-1/2}\lambda'\left( \mathbf{S}_n - \bm{\mu}\right) \overset{d}{\to} \lambda'\mathbf{Z}$ by Theorem \ref{theo:rio}. The result then follows by the Cramér-Wold device and Slutsky's lemma. \hfill$\square$
\begin{theorem}\label{theo:mikosch_multivariate}
    Let $(\mathbf{X}_t)_{t\in\mathbb{Z}}$ satisfy Assumptions \ref{ass_RV_multivariate}--\ref{AC,MX_multivariate}. If the tail index $\kappa\in(1,2)$, assume that $\mathbb{E}[\mathbf{X}_0]=0$. With $\tilde{a}_n\to\infty$ satisfying \eqref{eq:def_a_n_multivariate}, it holds that
    \begin{equation}\label{eq:thm_joint_conv_multivariate}
        \tilde{a}_n^{-1}(\mathbf{S}_n,\tilde{\gamma}_{n})\xrightarrow{d}(\bm{\xi}_\kappa,\tilde{\zeta}_{\kappa/2}), \qquad n\to\infty,
    \end{equation}
    where $\bm{\xi}_\kappa\in \mathbb{R}^m$ is a $\kappa$-stable random vector and $\tilde{\zeta}_{\kappa/2}^2>0$ is a $\kappa/2$-stable random variable.
\end{theorem}
\textit{Proof:} The proof is a straightforward adaptation of the proof of Theorem \ref{theo:inf_var_jointCLT}: Define
\begin{equation}
    \Psi_n(\mathbf{u},\lambda) :=\mathbb{E}[\exp(i\tilde{a}_n^{-1}\mathbf{u}'\mathbf{S}_n-\tilde{a}_n^{-2}\lambda\tilde{\gamma}_{n}^2)],\quad (\mathbf{u},\lambda)\in \mathbb{R}^m\times\mathbb{R}_+.
\end{equation}
Then it suffices to show that for the integer sequence $r_n \to \infty$ in Assumption \ref{AC,MX_multivariate} and $k_n:=\lfloor n/r_n\rfloor\to\infty$, it holds that for all $(\mathbf{u},\lambda)\in \mathbb{R}^m\times\mathbb{R}_+$
\begin{align}\label{eq:mix_condition_multivariate}
    \Psi_n(\mathbf{u},\lambda)  =  \left(\mathbb{E}[\exp(i\tilde{a}_n^{-1}\mathbf{u}'\mathbf{S}_{r_n}-\tilde{a}_n^{-2}\lambda\tilde{\gamma}_{r_n}^2)]\right)^{k_n}+o(1).
\end{align}
Condition \eqref{eq:mix_condition_multivariate} follows by arguments similar to those used for showing \eqref{eq:mix_condition}. \hfill$\square$
\subsection{Proof of Theorem \ref{thm:SPA}}
\textit{Proof of Part 1.a:} Follows by Theorem \ref{thm:rio_multivariate} and Theorem 1 of \cite{hansen2005SPA}. \\
\textit{Proof of Part 1.b:} The proof follows the structure of the proof of Theorem 1 of \cite{hansen2005SPA} carefully taking into account the stable limit result of Theorem \ref{theo:mikosch_multivariate}. Initially, note that by Theorem \ref{theo:mikosch_multivariate} and arguments identical to those used to prove Theorem \ref{theo:inf_var_jointCLT},
\begin{equation}\label{eq:multi_joint}
    \tilde{a}_n^{-1}\left(\mathbf{S}_n -n\mathbb{E}[\mathbf{X}_0] , \sqrt{\sum_{t=1}^n\Vert \mathbf{X}_t -\mathbb{E}[\mathbf{X}_0] \Vert^2}\right)\overset{d}{\to} \left(\bm{\xi}_{\kappa}, \tilde{\zeta}_{\kappa/2} \right), \qquad n\to\infty.
\end{equation}
Then using \eqref{eq:multi_joint} and that $\tilde{a}_n\to \infty$ with   $n/\tilde{a}_n^2=o(1)$ (as in \eqref{eq:rate_a_n}),
\begin{align*}
    \tilde{\gamma}_{n}^2 &= \tilde{a}_n^{-2}  \sum_{t=1}^n\Vert \mathbf{X}_t\Vert^2  = \tilde{a}_n^{-2} \sum_{t=1}^n\Vert \mathbf{X}_t -\mathbb{E}[\mathbf{X}_0] \Vert^2 + \tilde{a}_n^{-2}n\Vert \mathbb{E}[\mathbf{X}_0] +\tilde{a}_n^{-2}  \mathbb{E}[\mathbf{X}_0]'\sum_{t=1}^n\left(\mathbf{X}_t -\mathbb{E}[\mathbf{X}_0]\right) \\
     & = \tilde{a}_n^{-2} \sum_{t=1}^n\Vert \mathbf{X}_t -\mathbb{E}[\mathbf{X}_0] \Vert^2 + o_{\mathbb{P}}(1),
\end{align*}
such that
\begin{equation}\label{eq:multi_conv_gamma}
     \tilde{\gamma}_{n} \overset{d}{\to}  \tilde{\zeta}_{\kappa/2}, \qquad n\to\infty.
\end{equation}
Define $\mathbf{W}_n=(W_{n,1},\dots W_{n,m})'$ with $ W_{n,j}=\tilde{a}_n^{-1}S_{n,j}\mathbf{1}_{(\mathbb{E}[X_{0,j}]<0)}$,
and  $\mathbf{U}_n=(U_{n,1},\dots U_{n,m})'$ with $ U_{n,j}=\tilde{a}_n^{-1}S_{n,j}\mathbf{1}_{(\mathbb{E}[X_{0,j}]=0)}$, such that $\tilde{a}_n^{-1}\mathbf{S}_n=\mathbf{W}_n+\mathbf{U}_n$.
By \eqref{eq:multi_joint},
\begin{equation*}
    \mathbf{W}_n-n\tilde{a}_n^{-1}\mathbb{E}[\mathbf{X}_0]\overset{d}{\to}[\xi_{\kappa,1}\mathbf{1}_{(\mathbb{E}[X_{0,1}]<0)},\dots,\xi_{\kappa,m}\mathbf{1}_{(\mathbb{E}[X_{0,m}]<0)}]', \qquad n\to\infty,
\end{equation*}
such that
\begin{equation}\label{eq:multi_W_conv}
    \mathbf{W}_n \overset{\mathbb{P}}{\to}-[\infty\times\mathbf{1}_{(\mathbb{E}[X_{0,1}]<0)},\dots,\infty\times\mathbf{1}_{(\mathbb{E}[X_{0,m}]<0)}]', \qquad n\to\infty,
\end{equation}
where, by convention, $\infty\times0=0$, and likewise
\begin{equation}\label{eq:multi_U_conv}
    \mathbf{U}_n \overset{d}{\to} [\xi_{\kappa,1}\mathbf{1}_{(\mathbb{E}[X_{0,1}]=0)},\dots,\xi_{\kappa,m}\mathbf{1}_{(\mathbb{E}[X_{0,m}]=0)}]' :=\tilde{\mathbf{\xi}}_{\kappa}, \qquad n\to\infty.
\end{equation}
The test statistic in \eqref{eq:def_V_SPA} is given by $ V_n^{\mathrm{SPA}} = \varphi(\tilde{a}_n^{-1}\mathbf{S}_n,\tilde{a}_n^{-1}\tilde{\gamma}_{n})$ where $\varphi:\mathbb{R}^m\times(0,\infty)\to[0,\infty)$ is given by $\varphi(\mathbf{u},v) =\max\left(v^{-1}\max_{j=1,\dots,m}u_j,0\right)$.
Define $\mathbf{u}^+=(u_1^+,\dots,u_m^+)'$ with $u_j^+=\max(u_j,0)$, $j=1,\dots,m$. Then, it holds that
\begin{equation}\label{eq:multi_property_varphi}
    \varphi(\mathbf{u},v)=\varphi(\mathbf{u}^+,v),
\end{equation}
and by \eqref{eq:multi_W_conv}--\eqref{eq:multi_U_conv},
\begin{equation}\label{eq:multi_U_W_positive}
    (\mathbf{U}_n+{\mathbf{W}_n})^+=\mathbf{U}_n^+ +o_{\mathbb{P}}(1).
\end{equation}
Combining \eqref{eq:multi_conv_gamma}, \eqref{eq:multi_property_varphi}, and \eqref{eq:multi_U_W_positive},
we have that
\begin{align*}
    V_n^{\mathrm{SPA}} %&=\varphi(\tilde{a}_n^{-1}\mathbf{S}_n,\tilde{a}_n^{-1}\tilde{\gamma}_{n}) \\ 
    & = \varphi(\mathbf{U}_n,\tilde{a}_n^{-1}\tilde{\gamma}_{n}) + o_{\mathbb{P}}(1) \overset{d}{\to} \varphi(\tilde{\bm{\xi}}_\kappa,\tilde{\zeta}_{\kappa/2}) ={\max\left[\max_{j=1,\dots,m}\tilde{\xi}_{\kappa,j},0\right]}/\tilde{\zeta}_{\kappa,2}, \qquad n\to\infty,
\end{align*}
where the convergence arises from the continuous mapping theorem. \\
\textit{Proof of Part 2:} For the finite variance case, the result holds by Theorem \ref{thm:rio_multivariate} and Theorem 1 of \cite{hansen2005SPA}. For the case of $\kappa \in (1,2)$, by \eqref{eq:multi_joint} it holds that for some $j\in\{1,\dots,m\}$ where $\mathbb{E}[X_{0,j}]>0$, as $n\to\infty$,
\begin{equation*}
    U_{n,j}=\tilde{a}_n^{-1}S_{n,j}=\tilde{a}_n^{-1}(S_{n,j}-n\mathbb{E}[X_{0,j}])+n\tilde{a}_n^{-1}\mathbb{E}[X_{0,j}]= O_{\mathbb{P}}(1) + n\tilde{a}_n^{-1}\mathbb{E}[X_{0,j}]\overset{\mathbb{P}}{\to}\infty,
\end{equation*}
using that $\tilde{a}_n=o(n)$. The result now follows by noting that $\varphi(\mathbf{u}, v)\to \infty$ if $u_j\to\infty$ for some $j\in\{1,\dots,m\}$. \hfill$\square$

\section{Examples of data-generating processes}\label{sec:DGPs}
Section \ref{sec:AR_details} considered linear autoregressive processes as an example of processes satisfying Assumptions \ref{mix_coef_ass}--\ref{ass:AC+MX}. In this section, we consider additional examples.
\subsection{Affine stochastic recurrence equations} \label{sec:SRE}
Consider the process for $X_t \in \mathbb{R}$  given by
\begin{equation}
    X_t = A_t X_{t-1} + B_t,\quad t\in\mathbb{Z}, \label{eq:SRE}
\end{equation}
where $((A_t,B_t))_{t\in\mathbb{Z}}$ is an iid process with  $A_t,B_t\in\mathbb{R}$. This type of process, known as a so-called affine stochastic recurrence equation (SRE), has been extensively studied in recent books by \cite{buraczewski2016stochastic} and \cite{mikosch2024extreme}. The process is general in the sense that it nests a wide range of well-known time series processes, such as ARCH processes, double autoregressive (DAR) processes, and nonlinear GARCH processes. 
With $(A,B)$ denoting generic elements of $((A_t,B_t))_{t\in\mathbb{Z}}$, we have the following result\footnote{Several of the conditions can be relaxed or modified in order to adapt to the aforementioned processes.}.
\begin{proposition}
Suppose that
\begin{enumerate}
    \item $(A,B)$ has a Lebesgue density,
    \item $\mathbb{P}(Ax + B = x) < 1$ for all $x\in \mathbb{R}$
    \item $\mathbb{P}(A<0)>0$,
    \item there exists a constant $\kappa>0$ such that 
    \begin{equation*}
        \mathbb{E}[|A|^{\kappa}]=1, \quad \mathbb{E}[|B|^{\kappa}]<\infty, \quad \mathbb{E}[|A|^{\kappa}\max\{\log(|A|),0\}]<\infty.
    \end{equation*}

\end{enumerate}
 Then there exists a (with probability one) unique stationary causal solution $(X_t)_{t\in \mathbb{Z}}$ to \eqref{eq:SRE}. This solution has the following properties:
 \begin{enumerate}[(i)]
     \item $(X_t)_{t\in \mathbb{Z}}$ has strongly mixing coefficients with geometric decay, that is, $\alpha_k\leqslant c \rho^k$ for some constants $c\in(0,\infty)$ and $\rho\in(0,1)$, $k=0,1,\dots$,
     \item $(X_t)_{t\in \mathbb{Z}}$ is regularly varying with index $\kappa>0$ in the sense of \eqref{RV_TS}. Its spectral process $(\Theta_t)_{t\in\mathbb{Z}}$, satisfies $\mathbb{P}(\Theta_0 = 1) = \mathbb{P}(\Theta_0 = -1) = 1/2$, $\Theta_t = \Theta_0 \prod_{i=1}^t A_i $ for $t>0$, and
     \begin{equation*}
         \mathbb{P}\left( (\Theta_{-h},\dots,\Theta_{h}) \in \cdot  \right) = \mathbb{E} \left[ \frac{|\Theta_h|^\kappa}{\mathbb{E}[|\Theta_h|^\kappa]}\mathbf{1}\left\{ \frac{(\Theta_{0},\dots\Theta_{2h})}{|\Theta_h|} \in \cdot \right\} \right], \ h>0, 
     \end{equation*}
     \item $(X_t)_{t\in \mathbb{Z}}$ obeys the anti-clustering condition \eqref{AC} for any sequence $r_n$ satisfying $r_n=o(a_n^2/n)$ and $\log(n) = o(r_n)$ as $n\to\infty$. For any such choice of $r_n$ there exists another sequence $\ell_n \to \infty$ satisfying $\ell_n = o(r_n)$ such that \eqref{alpha_coefs_req} holds.
     
 \end{enumerate}
\end{proposition}
\textbf{Proof:}
 Since $A$ has a Lebesgue density, $\mathbb{P}(A=0)=0$. Moreover, it holds that $\mathbb{E}[\log(|A|)] < \kappa^{-1}\log(\mathbb{E}[|A|^\kappa])=0$ and $\mathbb{E}[\log^+|B|]\leqslant \kappa^{-1}\mathbb{E}[|B|^{\kappa}]<\infty$. Using \citet[Theorem 5.6.2]{mikosch2024extreme} it follows that the SRE admits an almost surely unique, stationary solution. In terms of the mixing numbers, let $\mathbb{P}_0$ denote the stationary distribution of $X_t$. By $1.$ and $3.$ and using \citet[Proposition 2.2.1 and Lemma 2.2.2]{buraczewski2016stochastic}, it holds that  the Markov chain $(X_t)_{t\geqslant0}$ is $\mathbb{P}_0$-irreducible. Moreover, using $4.$ there exists $\varepsilon > 0$ such that $\mathbb{E}[|A|^\varepsilon]<1$ and $\mathbb{E}[|B|^\varepsilon]<\infty$, and using \citet[Proposition 2.2.4]{buraczewski2016stochastic} $(X_t)_{t\in\mathbb{Z}}$ has strong mixing coefficients with geometric decay. The regular variation of $(X_t)_{t\in\mathbb{Z}}$ follows by \citet[Proposition 5.6.10]{mikosch2024extreme} with the specific structure of the spectral process obtained from the arguments given by \citet[pp. 266-267]{mikosch2024extreme}. The anti-clustering condition follows by \citet[proof of Theorem 9.3.6]{mikosch2024extreme}. The geometric decay on the strong mixing coefficients and the conditions on $r_n$ imply that the condition \eqref{alpha_coefs_req} holds for $\ell_n \to \infty$ satisfying $\ell_n=o(r_n)$ and $\log(n/r_n)=o(\ell_n)$. 
  \subsection{Stochastic volatility processes}
For $X_t\in\mathbb{R}$ the class of stochastic volatility (SV) processes is given by 
\begin{equation}
    X_t = \sigma_t z_t,\quad t\in\mathbb{Z}, \label{eq:X_SV}
\end{equation}
where, typically, $((\sigma_t,z_t))_{t\in \mathbb{Z}}$ is stationary and  $(\sigma_t)_{t\leqslant0}$ and $(z_t)_{t\geqslant 0}$ are independent. A special case is the stochastic autoregressive volatility process with $X_t$ given by \eqref{eq:X_SV} and
\begin{equation} \label{eq:sigma_SV}
    \log(\sigma_t) = \omega + \phi \log(\sigma_{t-1}) + (\gamma + \beta \log(\sigma_{t-1}))u_t,
\end{equation}
with $\omega, \phi, \gamma, \beta \in \mathbb{R}$, $\gamma \neq 0$ and
\begin{equation}\label{eq:innovations_SV}
    (z_t)_{t\in\mathbb{Z}} \ \text{and} \ (u_t)_{t\in\mathbb{Z}} \ \text{are mutually independent iid processes},
\end{equation}
see, e.g., \citet[Section 5]{Carrasco_Chen_2002}. Note here that the log-volatility process $\log(\sigma_t)$ in \eqref{eq:sigma_SV} obeys an affine SRE of the form \eqref{eq:SRE} with
\begin{equation} \label{eq:logSV_SRE}
 A_t = (\phi + \beta u_t) \quad \text{and}\quad B_t = \omega + \gamma u_t.     
\end{equation}
We have the following result.
\begin{proposition}\label{prop:logSV}
    Let $X_t$ and $\sigma_t$ satisfy \eqref{eq:X_SV}-\eqref{eq:innovations_SV}, and suppose that
    \begin{enumerate}
        \item the distribution of $u_t$ has a Lebesgue density positive on $\mathbb{R}$, with $\mathbb{E}[|u_t|]<\infty$,
        \item $|\phi|<1$ and $\mathbb{E}[|\phi + \beta u_t|]<1$.
    \end{enumerate}
    Then there exists a (with probability one) unique stationary causal solution, $(X_t)_{t\in\mathbb{Z}}$, to \eqref{eq:X_SV}. This solution is strongly mixing with geometric decay. \\
    Suppose in addition that 
    \begin{enumerate}  \setcounter{enumi}{2}
        \item the distribution of $z_t$ is regularly varying with index $\kappa>0$ in the sense of \eqref{RV_marg},
        \item $\mathbb{E}[\sigma_t^p]<\infty$ for some $p>\kappa$.
    \end{enumerate}
    Then 
    \begin{enumerate}[(i)]
         \item $(X_t)_{t\in \mathbb{Z}}$ is regularly varying with index $\kappa>0$ in the sense of \eqref{RV_TS}. Its spectral process $(\Theta_t)_{t\in\mathbb{Z}}$ satisfies $\mathbb{P}(\Theta_0 = 1) = q_+ $ and $\mathbb{P}(\Theta_0 = -1) = q_-$ with the constants $q_+,q_-\geqslant0$ provided by \eqref{RV_marg}, and $\mathbb{P}(\Theta_t = 0)=1$ for $t\neq 0$,
           \item $(X_t)_{t\in \mathbb{Z}}$ obeys the anti-clustering condition \eqref{AC} for a suitable sequence $r_n \to \infty$ satisfying $r_n=o(n)$ as $n\to\infty$. For this choice of sequence there exists another sequence $\ell_n \to \infty$ satisfying $\ell_n = o(r_n)$ such that \eqref{alpha_coefs_req} holds.
    \end{enumerate}
\end{proposition}
\textbf{Proof:} It holds that $\log(\sigma_t)$ obeys an SRE with $A_t$ and $B_t$ given by \eqref{eq:logSV_SRE}. Moreover, $\mathbb{P}(A=0)=0$, $\mathbb{E}[\log(|A|)]\leqslant \log(\mathbb{E}[|\phi + \beta u_t|])<0$ and $\mathbb{E}[\log^+(|B|)]\leqslant \mathbb{E}[|\omega + \gamma u_t|]\leqslant |\omega| + |\gamma| \mathbb{E}[|u_t|] < \infty$. Consequently, using \citet[Theorem 5.6.2]{mikosch2024extreme} we have that there exists an almost surely unique, stationary, causal solution $(\log(\sigma_t))_{t\in\mathbb{Z}}$ to \eqref{eq:logSV_SRE}. It  follows that there exists an almost surely unique, stationary causal solution $(X_t)_{t\in\mathbb{Z}}$, to \eqref{eq:X_SV}. By \citet[Proposition 2]{Carrasco_Chen_2002} it holds that $(\log(\sigma_t))_{t\in\mathbb{Z}}$ is strongly mixing with geometric decay. Likewise, since $\sigma_t$ is a measurable function of $\log(\sigma_t)$, $(\sigma_t)_{t\in\mathbb{Z}}$ is strongly mixing with the same decay. By arguments given by \citet[pp. 325--326]{mikosch2024extreme}, the same mixing properties hold for $(X_t)_{t\in\mathbb{Z}}$. Regular variation of $(X_t)_{t\in\mathbb{Z}}$ follows by \citet[Proposition 5.3.1]{mikosch2024extreme}. The properties in \eqref{AC}--\eqref{alpha_coefs_req} follow by arguments in \citet[Section 9.3.1]{mikosch2024extreme}.

\begin{remark}
    Condition $5.$ of Proposition \ref{prop:logSV} requires that $\sigma_t$ has a finite moment of some order $p$ exceeding the tail index $\kappa>0$. For $\beta\neq 0$ this moment condition does not necessarily hold as $\log(\sigma_t)$ itself may be regularly varying; e.g., \citet[Theorem 5.6.9]{mikosch2024extreme}. In the case where $\beta =0$, $\log(\sigma_t)$ is a stationary AR(1) process (under condition $2.$), and, for instance, with $u_t$ Gaussian, it holds that $\sigma_t$ is log-normal with moments of any order  being finite.
\end{remark}

\end{document}